\documentclass[11pt]{article}
\usepackage[utf8]{inputenc}
\usepackage[margin=1in]{geometry}

\usepackage{mathtools}
\usepackage{amsmath,amssymb,shortcuts,mathtools,cleveref,tikz,subcaption,amsthm,amssymb}

\usepackage{multirow,booktabs,url}
\usepackage{natbib}

\usepackage[figuresright]{rotating}

\RequirePackage[normalem]{ulem} %
\PassOptionsToPackage{authoryear,sort}{natbib}
\usetikzlibrary{positioning,arrows,decorations,decorations.markings,decorations.pathreplacing,decorations.pathmorphing,shadows,positioning,arrows.meta,matrix,fit,shapes.swigs}
\usetikzlibrary{shapes,snakes}
\Crefname{assumption}{Assumption}{Assumptions}
\Crefname{exmp}{Example}{Examples}

\newcommand{\Acal}{\mathcal{A}}
\newcommand{\Xcal}{\mathcal{X}}
\newcommand{\Ycal}{\mathcal{Y}}
\newcommand{\Scal}{\mathcal{S}}
\newcommand{\Ucal}{\mathcal{U}}

\newcommand{\Vcal}{\mathcal{V}}
\newcommand{\Wcal}{\mathcal{W}}
\newcommand{\Sc}{S_{3}}
\newcommand{\Sb}{S_{2}}
\newcommand{\Sa}{S_{1}}
\newcommand{\hk}{\hat{h}_k}
\newcommand{\qk}{\hat{q}_k}
\newcommand{\out}{\text{OTC}}
\newcommand{\sel}{\text{SEL}}
\newcommand{\dr}{\text{DR}}
\newcommand{\nEa}{n_{E}^{(a)}}
\newcommand{\nOka}{n_{O, k}^{(a)}}

\newcommand{\score}{\mathrm{SC}}

\newtheorem{theorem}{Theorem}
\newtheorem{corollary}{Corollary}
\newtheorem{lemma}{Lemma}
\newtheorem{proposition}{Proposition}
\newtheorem{assumption}{Assumption}
\newtheorem{definition}{Definition}
\theoremstyle{definition}
\newtheorem{example}{Example}
\newtheorem{remark}{Remark}

\newcommand*\samethanks[1][\value{footnote}]{\footnotemark[#1]}
\title{Long-term Causal Inference Under Persistent Confounding via Data Combination}
\author{Guido Imbens$^1$\footnote{Corresponding author: imbens@stanford.edu}~\thanks{Alphabetical order.}, ~Nathan Kallus$^2$\samethanks, ~Xiaojie Mao$^3$\samethanks, ~Yuhao Wang$^{4, 5}$\samethanks}
\date{\small
$^1$ Graduate School of Business, Stanford University, Stanford, CA 94305, USA \\
$^2$ Cornell Tech, Cornell University, New York, NY 10044, USA; \\
$^3$ School of Economics and Management, Tsinghua University, Beijing 100084, China; \\
$^4$ Institute for Interdisciplinary Information Sciences, Tsinghua University, Beijing 100084, China; \\
$^5$ Shanghai Qi Zhi Institute, Shanghai 200232, China.
}

\begin{document}
\maketitle

\begin{abstract}
We study the identification and estimation of long-term treatment effects by combining short-term experimental data and long-term observational data subject to unobserved confounding.
This problem arises often when concerned with long-term treatment effects since experiments are often short-term due to operational necessity while observational data can be more easily collected over longer time frames but may be subject to confounding.
In this paper, we  tackle the challenge of \emph{persistent} confounding: unobserved confounders that can simultaneously affect the treatment, short-term outcomes, and long-term outcome. In particular, persistent confounding invalidates identification strategies in previous approaches to this problem. 
To address this challenge, we exploit the sequential structure of multiple short-term outcomes and develop several novel identification strategies for the average long-term treatment effect. 
Based on these,
we develop estimation and inference methods with asymptotic guarantees.
To demonstrate the importance of handling persistent confounders, we apply our methods to estimate the effect of a job training program on long-term employment using semi-synthetic data.\\

\end{abstract}


\section{Introduction}\label{sec: intro}
Empirical researchers and decision-makers are often interested in learning the long-term treatment effects of interventions.
For example, labor economists are interested in the effect of early childhood education on lifetime earnings \citep{chetty2011does}, 
marketers are interested in the effects of promotions on long-term revenue \citep{yang2020targeting}, 
online platforms are interested in the effects of webpage designs on users' long-term behaviors \citep{hohnhold2015focusing}.
Since a long-term effect can be quite different from short-term effects \citep{kohavi2012trustworthy}, accurately evaluating the long-term effect is both difficult and crucial for comprehensively understanding the intervention of interest. 

Learning long-term treatment effects is very challenging in practice because long-term outcomes are seldom observed within the time frame of randomized experiments.
For example, randomized experiments in online platforms (often termed A/B tests within that context) usually last for only a few weeks, and practitioners in the industry commonly recognize evaluation of long-term effects as a paramount challenge \citep{gupta2019top}.
In contrast, observational data are often easier and cheaper to acquire and can be collected retroactively, so they are more likely to include long-term outcome observations. 
Nevertheless, observational data are susceptible to unmeasured confounding, which can lead to biased treatment effect estimates. 
Therefore, long-term causal inference is very challenging using only experimental or observational data, either due to missing long-term outcome (in experimental data) or unmeasured confounding (in observational data). 

In this paper, we study the identification and estimation of long-term treatment effects by combining \emph{both} experimental and observational data. 
By combining these two different types of data, we hope to leverage their complementary strengths, i.e., the randomized treatment assignments in the experimental data and the long-term observations in the observational data. 
In particular, we aim to tackle the presence of \emph{persistent confounding} in the observational data, which cannot be generally ruled out. That is, we allow some unobserved confounders to have persistent effects in the sense that they can affect \emph{not only} the short-term outcomes \emph{but also} the long-term outcome.  
Persistent confounders are prevalent in long-term studies. 
For example, in studying early childhood education's effect on lifetime earnings, students' innate intelligence and/or familial support systems can affect both short-term and long-term earnings. 
Our setup is summarized in the causal diagrams in \Cref{figure: DAG-a}.

A few previous works also consider data combination for long-term causal inference.
\cite{athey2019surrogate}, in a setting where the observational sample contains no information on the treatment, rely on a surrogate criterion first proposed by \cite{prentice1989surrogate}.  
 \cite{athey2020combining}, in the same setting as considered in the current paper, assume a latent unconfoundedness condition. 
{While these conditions make no explicit reference to persistent confounding and its absence, a nontrivial persistent confounder can generally violate these (see \Cref{sec: validation} for details). At the same time, both settings are \textit{just identified}, meaning the conditions imposed are minimal, {so that if a condition is dropped then another would be needed in its place to  guarantee identification.}

{In this paper, we leverage an assumed \emph{sequential} structure between \emph{multiple} short-term outcomes to tackle long-term causal inference in the presence of persistent confounders.
Our new identification and estimation strategies are based on using short-term outcomes as \emph{proxy variables} for the persistent confounders.}  
To the best of our knowledge, this is the first time that the internal structure of short-term outcomes is used to address unmeasured confounding in long-term causal inference. 
Indeed, although \cite{athey2019surrogate,athey2020combining} also advocate using multiple short-term outcomes, they view them as a whole without leveraging their internal structure.  
Our work therefore also provides new insights on the special role of using multiple short-term outcomes in long-term causal inference. 

Our contributions are summarized as follows:
\begin{itemize}
\item {We propose several novel identification strategies for the average long-term treatment effect in the presence of persistent confounders.
 These identification strategies rely on three groups of short-term outcomes, where two of these groups are used as informative proxy variables for the unobserved confounders (\Cref{assump: completeness}). These short-term outcomes, together with the long-term outcome, are assumed to follow a sequential structure encapsulated in a conditional independence condition (\Cref{assump: CI}).}
\item Based on the identification strategies, we propose estimators for the average long-term treatment effect. These estimators involve fitting two nuisance functions that are defined as solutions to two conditional moment equations. 
Our estimation procedures accommodate any nuisance estimator among many existing ones. We provide high level conditions for the asymptotic consistency and asymptotic normality of our estimators. 
\item We evaluate the performance of our proposed estimators based on large-scale experimental data for a job-training program with long-term employment observations. We combine part of the experimental data and some semi-synthetic observational data with realistic  persistent confounding.
We demonstrate that due to the persistent confounding, our proposed estimators have smaller error than estimators that do not handle persistent confounding.   
\end{itemize}

The rest of this paper is organized as follows. We first review the related literature in \Cref{sec: literature} and set up our problem in \Cref{sec: setup}. Then we discuss our identification strategies in \Cref{sec: identification}, where each subsection features one different identification strategy. In \Cref{sec: estimation}, we present our long-term treatment effect estimators and analyze their asymptotic properties.
We further dicuss some extensions in \Cref{sec: extension}.
In \Cref{sec: experiment}, we illustrate the performance of methods in a semi-synthetic experiment.
We finally conclude this paper in \Cref{sec: conclusion}. 

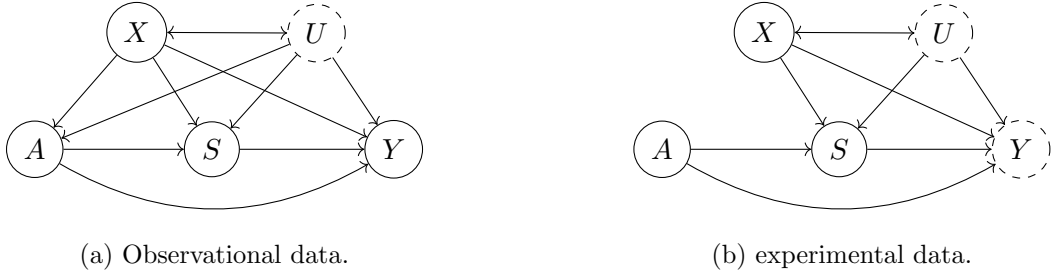
\begin{figure}[t]
\begin{subfigure}[b]{0.5\textwidth}
\centering 
\begin{tikzpicture}
\node[draw, circle, text centered] (A) {$A$};
\node[draw, circle, text centered, right=4cm of A] (Y) {$Y$};
\node[draw, circle, text centered, right=1.6cm of A] (S) {$S$};
\node[draw, circle, text centered, above right=1cm and 0.8cm of A] (X) {$X$};
\node[draw, circle, dashed, text centered, above right=1cm and 3.2cm of A] (U) {$U$};

\draw[->] (A) -- (S);
\draw[->] (A) to [bend right] (Y);
\draw[->] (S) -- (Y);
\draw[->] (X) -- (A);
\draw[->] (X) -- (S);
\draw[->] (X) -- (Y);
\draw[->] (U) -- (A);
\draw[->] (U) -- (S);
\draw[->] (U) -- (Y);
\draw[<->] (X) -- (U);
\end{tikzpicture}
\caption{Observational data.}
\label{figure: DAG-obs-a}
\end{subfigure}
\begin{subfigure}[b]{0.5\textwidth}
\centering 
\begin{tikzpicture}
\node[draw, circle, text centered] (A) {$A$};
\node[draw, dashed, circle, text centered, right=4cm of A] (Y) {$Y$};
\node[draw, circle, text centered, right=1.6cm of A] (S) {$S$};
\node[draw, circle, text centered, above right=1cm and 0.8cm of A] (X) {$X$};
\node[draw, circle, dashed, text centered, above right=1cm and 3.2cm of A] (U) {$U$};

\draw[->] (A) -- (S);
\draw[->] (A) to [bend right] (Y);
\draw[->] (S) -- (Y);
\draw[->] (X) -- (S);
\draw[->] (X) -- (Y);
\draw[->] (U) -- (S);
\draw[->] (U) -- (Y);
\draw[<->] (X) -- (U);
\end{tikzpicture}
\caption{Experimental data.}
\label{figure: DAG-exp-a}
\end{subfigure}
\caption{Causal diagrams for observational and experimental data with persistent confounders. Here $A$ denotes the treatment, $S$ denotes (multiple) short-term outcomes, $Y$ denotes the long-term outcome, $X$ denotes covariates, and $U$ denotes unobserved confounders.  Confounders $U$ in both samples and the long-term outcome $Y$ in the experimental data are unobserved, so they are indicated by dashed circles. Note that unobserved confounders $U$ can simultaneously affect short-term outcomes $S$ and the long-term outcome $Y$.}
\label{figure: DAG-a}
\end{figure}

\section{Related Literature}\label{sec: literature}
\subsection{Surrogates}
Our paper is related to a large body of biostatistics literature on surrogate outcomes; see reviews in \cite{weir2006statistical,vanderweele2013surrogate,joffe2009related}.

These literature consider using the causal  effect of an intervention on a surrogate outcome (e.g., patients' short-term health) as a proxy for its treatment effect on the  outcome of primary interest (e.g., long-term health). 
To this end, many criteria have been proposed to ensure the validity of the surrogate outcome. 
Examples include the {statistical surrogate criterion} \citep{prentice1989surrogate}, principal surrogate criterion \citep{frangakis2002principal}, consistent surrogate criterion \citep{chen2007criteria}, among many others. 
However, these criteria can easily run into a logical paradox\footnote{{The paradox refers to the phenomenon  that the sign of the treatment effect on the target outcome cannot be predicted by the sign of the treatment effect on the surrogate and the sign of the surrogate's effect on the target outcome. For example, it is possible that a treatment has a negative effect on the target outcome, even if both the treatment effect on the surrogate and the surrogate's effect on the target outcome are positive. }} \citep{chen2007criteria} or rely on unidentifiable quantities,
showing the challenge of causal inference when the primary outcome is completely missing.
When multiple surrogates are available, \cite{wang2020model,price2018estimation}
 consider transforming these surrogates to optimally approximate the primary outcome. Their approaches can avoid the surrogate paradox discussed in \cite{chen2007criteria}. 
Nevertheless, learning surrogate transformations requires having experimental data with long-term outcome observations. 

In contrast, our paper does not need long-term outcome observations in the experimental data but only need them in observational data. Moreover, our paper does not view short-term outcomes as proxies for the long-term outcome, so we avoid these previous surrogate criteria.
{Instead, we consider three groups of short-term outcomes $(S_1, S_2, S_3)$, viewing
$S_1$ and $S_3$ as proxies for unmeasured confounders and $S_2$ as a surrogate for the effect of $S_1$ on $S_3$ (it cannot,
however, serve as a surrogate for the effect of $A$ on $Y$ due to unmediated effects).}
See also discussions in \Cref{sec: proximal}.

\subsection{Data Combination for Long-term Causal Inference}\label{sec: literature-data-comb}
Following \cite{athey2019surrogate}, some  recent literature also combine experimental and observational data, and rely on the statistical surrogate criterion, either to estimate cumulative treatment effects in dynamic settings \citep{battocchi2021estimating} or learn long-term optimal treatment policies \citep{yang2020targeting,cai2021gear}. 
\cite{chen2021semiparametric} derive the efficiency lower bound for average long-term treatment effect in settings of \cite{athey2019surrogate,athey2020combining}.
\cite{singh2021finite,singh2022generalized} further develop debiased long-term treatment effect estimators based on machine learning nuisance estimation.
In contrast, \cite{kallus2020role,cai2021coda} combine two datasets that both satisfy unconfoundedness. 
Still, all of these works  rule out persistent confounding, which is the main problem tackled in  this paper.

{A concurrent and independent work by \cite{ghassami2022combining} uses alternative conditions or additional variables to alleviate latent confounding in long-term causal inference.
They propose three different identification strategies, and their proximal data fusion strategy is closely related to our approach in \Cref{sec: identify-OBF,sec: identify-SBF,sec: covariate-adaptive-identify}. 
Their approach requires auxiliary proxy variables satisfying certain generic conditions (in addition to the short-term outcomes). 
In contrast, our work specifically leverages the special sequential structure of multiple short-term outcomes and shows how such short-term outcomes can proxy the confounders. 
{This provides us with the necessary proxy variables for identification}
and allows us to understand the different types of confounders and which need to be controlled (see \cref{sec: partial}).
Importantly, we develop both estimation {and inference} methods with  theoretical guarantees and validate them in a concret case study.
Moreover, we provide an alternative control function  identification strategy in \Cref{sec: control-fun} and study how the short-term outcomes may help weaken a widely assumed external validity condition in \Cref{sec: ext-validity}. 
These results have no analogues in \cite{ghassami2022combining}.}

There is also growing interest in combining experimental and observational data to improve, rather than enable, causal inference \citep[e.g., ][]{chen2021minimax,cheng2021adaptive,yang2020elastic,yang2020improved,colnet2020causal,kallus2018removing,rosenman2022propensity,rosenman2020combining,yang2019combining}.
In these works, the outcome of interest is observed in both types of data, so causal-effect identification is already guaranteed by the experimental data. Instead, the aim of the data combination is to reduce variance.
In contrast to these works, in our setting, data combination is crucial for causal identification since any one data set alone cannot identify the long-term treatment effect. 

\subsection{Proximal Causal Inference}\label{sec: proximal}
Our identification proposals are related to how proximal causal inference 
deals with unmeasured confounding by leveraging proxy variables \citep{tchetgen2020introduction}.
The seminal work of \cite{Miao2016} demonstrated the identification of treatment effects with unobserved confounders given two different types of proxy variables: negative control outcomes, which are not affected by the treatment, and negative control treatments, which do not affect the outcome. 
Since then, a series of works have proposed a variety of different estimation methods based on this identification strategy \citep[e.g., ][]{kallus2021causal,GhassamiAmirEmad2021MKML,deaner2021proxy,SinghRahul2020KMfU,miao2018a,shi2020multiply,mastouri2021proximal,cui2020semiparametric}. 
The proximal causal inference framework has also been extended to longitudinal data analysis \citep{imbens2021controlling,ying2021proximal,shi2021theory}, mediation analysis \citep{dukes2021proximal,ghassami2021proximal}, and off-policy evaluation and learning \citep{bennett2021proximal,tennenholtz2020off,qi2021proximal,xu2021deep}.

The existing proximal causal inference literature focus on a single observational dataset. 
In contrast, in this paper we consider combining observational and experimental data.
We view short-term outcomes as proxy variables for persistent unmeasured  confounders. 
However, all of these short-term outcomes can be affected by the treatment (see \Cref{figure: DAG-b} below), so they do not satisfy the proxy conditions in \cite{Miao2016}. 
In this paper, we establish novel identification strategies that leverage the additional experimental data. See also discussions in \Cref{remark: proximal}.

\section{Problem Setup}\label{sec: setup}
We consider a binary treatment $A \in \Acal = \braces{0, 1}$ where $A = 1$ stands for the treated group and $A = 0$ stands for the control group.
We are interested in the treatment effect  on a long-term outcome. 
Using the potential outcome framework \citep{rubin1974estimating}, we postulate potential long-term outcomes  $Y(0), Y(1) \in\Ycal\subseteq\Rl$, which would be realized were the treatment assignment equal $0$ and $1$, respectively. 
In reality, we observe at most one of the potential outcomes per unit, corresponding to the actual treatment assignment, $Y = Y(A)$.

We may in fact observe neither potential long-term outcome in short-term experiments that end before these long-term outcomes can be observed.
Nevertheless,  it is usually still possible to observe some short-term outcomes. 
We postulate potential short-term outcomes $S(1) \in \Scal, S(0)\in \Scal$, and denote the observable realized short-term outcomes as $S = S(A)$. 
In this paper, we consider \emph{multiple} short-term outcomes, so we generally understand $S$ as a \emph{vector}. We discuss our assumptions on the inner structure of these short-term outcomes in \cref{sec: short-term}.
Additionally, we can observe some pre-treatment covariates  denoted as $X \in \Xcal$.

We have access to two samples: an observational (O) sample with $n_O$ units and an experimental (E) sample with $n_E$ units. 
We suppose that the observational sample is a random sample from the population of interest, where for each unit $i$ we can observe independently and identically distributed tuples $(X_i, A_i, S_i, Y_i)$. 
The experimental sample may be a selective sample from the same population, where for each unit $i$ we only observe $(X_i, A_i, S_i)$, but \emph{not} the long-term outcome. 
We use a binary indicator $G_i \in \braces{E, O}$ to denote which sample a unit $i$ belongs to. 
Without loss of generality, we consider a combined i.i.d sample of size $n = n_O + n_E$ from an artificial super-population, namely, $\mathcal{D} = \braces{(G_i, X_i, A_i, S_i, Y_i\indic{G_i = O}): i = 1, \dots, n_O + n_E}$. 
We use $\mathbb{P}$ and $\mathbb{E}$ to denote the probability and expectation with respect to this super-population, and use $p(\cdot)$ to denote the associated probability density function or probability mass function, as appropriate. 
We also denote the observational and experimental subsamples as $\mathcal{D}_O$ and $\mathcal{D}_E$, respectively. 

Our aim is to combine  the observational and experimental samples in order to learn the long-term treatment effect on the population associated with the observational data:
\begin{align}\label{eq: ATE}
&\tau = \mu\prns{1} - \mu\prns{0}, \text{ where } \mu\prns{a} = \Eb{Y\prns{a} \mid G = O}. 
 \end{align}
Our results easily extend to the average on the experimental or combined population. We focus on $\tau$ for concreteness and  we believe it captures the most commonly relevant estimand.

\subsection{Basic Assumptions for Observational and Experimental Data}
\begin{figure}[t]
\begin{subfigure}[b]{0.48\textwidth}
\centering 
\begin{tikzpicture}
\tikzset{swig hsplit={gap=3pt,
line color lower=red}}
\node[name=A, shape=swig hsplit]{
\nodepart{upper}{$A$}
\nodepart{lower}{{$a$}} };
\node[draw, ellipse, text centered, right=4cm of A] (Y) {$Y({a})$};
\node[draw, ellipse, text centered, right=1.6cm of A] (S) {$S({a})$};
\node[draw, circle, text centered, above right=1cm and 0.8cm of A] (X) {$X$};
\node[draw, circle, dashed, text centered, above right=1cm and 3.2cm of A] (U) {$U$};

\draw[->] (A) to [bend right] (S);
\draw[->] (A) to [bend right] (Y);
\draw[->] (S) -- (Y);
\draw[->] (X) -- (A);
\draw[->] (X) -- (S);
\draw[->] (X) -- (Y);
\draw[->] (U) -- (A);
\draw[->] (U) -- (S);
\draw[->] (U) -- (Y);
\draw[<->] (X) -- (U);
\end{tikzpicture}
\caption{Observational data.}
\label{figure: SWIG-obs}
\end{subfigure}
\hfill
\begin{subfigure}[b]{0.48\textwidth}
\centering 
\begin{tikzpicture}
\tikzset{swig hsplit={gap=3pt,
line color lower=red}}
\node[name=A, shape=swig hsplit]{
\nodepart{upper}{$A$}
\nodepart{lower}{{$a$}} };
\node[draw, ellipse, text centered, right=4cm of A] (Y) {$Y({a})$};
\node[draw, ellipse, text centered, right=1.6cm of A] (S) {$S({a})$};
\node[draw, circle, text centered, above right=1cm and 0.8cm of A] (X) {$X$};
\node[draw, circle, dashed, text centered, above right=1cm and 3.2cm of A] (U) {$U$};

\draw[->] (A) to [bend right] (S);
\draw[->] (A) to [bend right] (Y);
\draw[->] (S) -- (Y);
\draw[->] (X) -- (S);
\draw[->] (X) -- (Y);
\draw[->] (U) -- (S);
\draw[->] (U) -- (Y);
\draw[<->] (X) -- (U);
\end{tikzpicture}
\caption{Experimental data.}
\label{figure: SWIG-exp}
\end{subfigure}
\caption{{Single world intervention graphs (SWIG) corresponding to the causal diagrams in \Cref{figure: DAG-a}}.}
\label{figure: SWIG}
\end{figure}
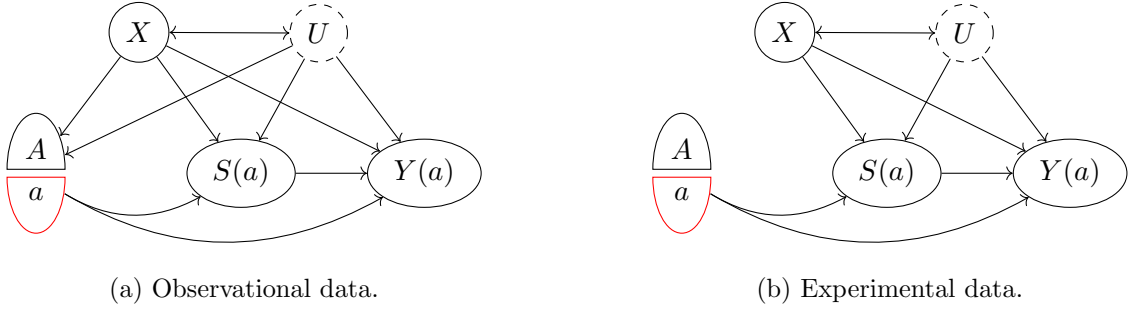

We now describe the basic assumptions that characterize the experimental and observational data sets as such. Unless otherwise stated, all of these assumptions are maintained throughout this paper. 

The observational data is generally confounded, that is, conditioning only on $X$ does \emph{not} render the treatment assignment independent of the potential short-term and long-term outcomes.
Instead, there exist some 
unobserved confounders $U \in \Ucal$ that are needed to account for the association between treatment and potential outcomes. {See \Cref{figure: SWIG-obs} for a single world intervention graph illustration \citep{richardson2013single} when intervening on the variable $A$.}

\begin{assumption}[Observational data]\label{assump: unconfound-obs}
For $a \in \braces{0, 1}$, 
\begin{align}\label{eq: unconf-obs}
\prns{Y\prns{a}, S\prns{a}} \perp A \mid U, X, G = O,
\end{align}
and $0 < \Prb{A = 1 \mid U, X, G = O} < 1$ almost surely.
\end{assumption}
\Cref{eq: unconf-obs} means that $U$ and $X$ together account for all confounding in the observational data, and generally the observed covariates $X$ alone are not enough. Moreover, we impose the overlap condition $0 < \Prb{A = 1 \mid U, X, G = O} < 1$, which is a standard assumption in causal inference literature.
Note that the existence of $U$ is without loss of generality because we can always take it to be the potential outcomes themselves.
Because of the unobserved confounders $U$, the observational data alone is not enough to identify the treatment effect parameter $\tau$ in \cref{eq: ATE}.

In contrast to the observational data, the treatments are assigned completely at random in the experimental data. {See \Cref{figure: SWIG-exp} for a single world intervention graph illustration.}
\begin{assumption}[Experimental Data]\label{assump: unconfound-exp}
For $a \in \braces{0, 1}$,
\begin{align}\label{eq: unconf-exp}
\prns{Y\prns{a}, S\prns{a}, U, X} \perp A \mid G = E,
\end{align}
and $0 < \Prb{A = 1 \mid G = E} < 1$ almost surely. 
\end{assumption}
{In \Cref{sec: extension}, we will relax \cref{assump: unconfound-exp} by allowing the treatment assignment in the experimental data to depend on the covariates $X$, so the conditions in \cref{assump: unconfound-exp} hold  conditioned on the covariates $X$.}

Although unconfounded, the experimental data do  not contain long-term outcome observations, so the experimental data alone is not enough to identify the treatment effect either. 
This motivates us to combine the observational and experimental data. 
To this end, we further impose the following assumption permitting such combination.
\begin{assumption}[External Validity]\label{assump: ext-valid}
For any $a \in \braces{0, 1}$,
\begin{align}\label{eq: external-validity-0}
\prns{S\prns{a}, U, X} \perp G,
\end{align}
and, almost surely, 
\begin{align}\label{eq: data-overlap}
\frac{p\prns{U, X \mid A = a, G = E}}{p\prns{U, X \mid A= a, G = O}} < \infty.
\end{align}
\end{assumption}

\Cref{assump: ext-valid} ensures that the two samples have enough commonality so it is meaningful to combine them.
\Cref{eq: external-validity-0} in \Cref{assump: ext-valid} means that the experimental data has \emph{external validity}, in that the distribution of $\prns{S\prns{a}, U, X}$ in the experimental data is the same as that in the observational data (i.e., the population of interest). 
Similar assumptions also appear in previous literature that attempt to combine different samples \citep[e.g., ][]{athey2020combining,athey2019surrogate,kallus2020role}.

In \Cref{sec: extension} we further relax \cref{eq: external-validity-0} to allow the distributions of covariates $X$ to be different in the two samples. 
Note that \cref{eq: external-validity-0} already allows the distributions of potential long-term outcome $Y(a)$ in the experimental and observational data to be different, so the long-term treatment effect on the experimental population can be different from our target.

\Cref{eq: data-overlap} in \Cref{assump: ext-valid} means that the conditional distributions of $(U, X) \mid A$ on the experimental and observational data have enough overlap, which is also a common assumption in missing data literature \citep{tsiatis2007semiparametric}.

\subsection{Three Groups of Short-term Outcomes}\label{sec: short-term}
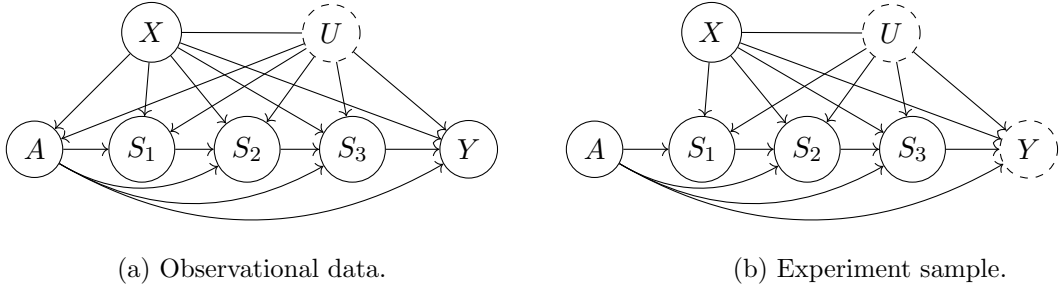
\begin{figure}[t]
\centering 
\begin{subfigure}[b]{0.49\textwidth}
\centering
\begin{tikzpicture}
\node[draw, circle, text centered] (A) {$A$};
\node[draw, circle, text centered, right=5cm of A] (Y) {$Y$};
\node[draw, circle, text centered, right=0.6cm of A] (S) {$S_1$};
\node[draw, circle, text centered, right=2cm of A] (S2) {$S_2$};
\node[draw, circle, text centered, right=3.4cm of A] (S3) {$S_3$};
\node[draw, circle, text centered, above right=1cm and 1cm of A] (X) {$X$};
\node[draw, circle, dashed, text centered, above right=1cm and 3.4cm of A] (U) {$U$};

\draw[->] (A) -- (S);
\draw[->] (A) to [bend right] (S2);
\draw[->] (A) to [bend right] (S3);
\draw[->] (A) to [bend right] (Y);
\draw[->] (X) -- (A);
\draw[->] (X) -- (S);
\draw[->] (X) -- (S2);
\draw[->] (X) -- (S3);
\draw[->] (U) -- (S);
\draw[->] (U) -- (S2);
\draw[->] (U) -- (S3);
\draw[->] (S) -- (S2);
\draw[->] (S2) -- (S3);
\draw[->] (S3) -- (Y);
\draw[->] (X) -- (Y);
\draw[->] (U) -- (A);
\draw[->] (U) -- (Y);
\draw[-] (X) -- (U);
\end{tikzpicture}
\caption{Observational data.}
\label{figure: DAG-obs-b}
\end{subfigure}
\begin{subfigure}[b]{0.49\textwidth}
\begin{tikzpicture}
\node[draw, circle, text centered] (A) {$A$};
\node[draw, dashed, circle, text centered, right=5cm of A] (Y) {$Y$};
\node[draw, circle, text centered, right=0.6cm of A] (S) {$S_1$};
\node[draw, circle, text centered, right=2cm of A] (S2) {$S_2$};
\node[draw, circle, text centered, right=3.4cm of A] (S3) {$S_3$};
\node[draw, circle, text centered, above right=1cm and 1cm of A] (X) {$X$};
\node[draw, circle, dashed, text centered, above right=1cm and 3.4cm of A] (U) {$U$};

\draw[->] (A) -- (S);
\draw[->] (A) to [bend right] (S2);
\draw[->] (A) to [bend right] (S3);
\draw[->] (A) to [bend right] (Y);
\draw[->] (X) -- (S);
\draw[->] (X) -- (S2);
\draw[->] (X) -- (S3);
\draw[->] (U) -- (S);
\draw[->] (U) -- (S2);
\draw[->] (U) -- (S3);
\draw[->] (S) -- (S2);
\draw[->] (S2) -- (S3);
\draw[->] (S3) -- (Y);
\draw[->] (X) -- (Y);
\draw[->] (U) -- (Y);
\draw[-] (X) -- (U);
\end{tikzpicture}
\caption{Experiment sample.}
\label{figure: DAG-exp-b}
\end{subfigure}
\caption{Sequential structure of three groups of short-term outcomes.}
\label{figure: DAG-b}
\end{figure}

To address general persistent confounding, we need some additional structure on the short-term outcomes. 
In this paper, we consider leveraging \emph{multiple, sequential short-term outcomes}. 
In particular, we consider a partitioning of the short-term outcomes into \emph{three groups} sorted in a temporal order, writing the potential short-term outcomes as $S(a) = \prns{\Sa(a), \Sb(a), \Sc(a)} \in \Scal_1 \times \Scal_2 \times \Scal_3$ and their observed counterparts as $S = \prns{\Sa, \Sb, \Sc}$. Given this partitioning, we assume the following conditional independence structure for the potential short-term and long-term outcomes. 
\begin{assumption}[Sequential Outcomes]\label{assump: CI}
For $a \in \braces{0, 1}$,
\begin{align}
\prns{Y(a), \Sc(a)} \perp \Sa(a) 
    &\mid \Sb(a), U, X, G = O, \label{eq: CI-potential-1}
\end{align}
\end{assumption} 
\Cref{assump: CI} requires that the effect of the first short-term outcome on the last short-term outcome and the long-term outcome is mediated by the intermediate short-term outcome.
Nonetheless, all outcomes can be related by \emph{unobserved} confounders, even in the experimental data, and the treatment can affect all outcomes both directly and indirectly.
This captures the sequential structure of the short-term and long-term outcomes (see \Cref{figure: DAG-b} for an example\footnote{{ \Cref{figure: DAG-b} is just one example diagram that satisfies \Cref{assump: CI}. The assumption could be also satisfied if there is an additional arrow from $S_2$ to $Y$. When $S_1(a), S_2(a), S_3(a), Y(a)$ follow a Markov process conditional on $U, X$ as discussed in this paragraph, the arrow from $S_2$ to $Y$ is indeed absent.}}).
For example, it holds when the potential outcomes follow autoregressive structural equations of suitable orders (see \Cref{ex: linear} below for a simple instance).
{Moreover, it holds when the potential outcomes $S_1(a), S_2(a), S_3(a), Y(a)$ follow a Markov process, conditional on $U, X$. Markov models are widely used in social sciences such as for modeling the dynamics of labor markets \citep{poterba1986reporting,mohapatra2007rise}, and in medical sciences for modeling chronic-disease progression \citep{marshall1995multi,liu2013glycemic,kay1986markov,liu2013glycemic}.}

{
Note that \Cref{assump: CI} significantly differs from the statistical surrogacy condition in \cite{athey2019surrogate}. The latter requires short-term outcomes $S$ to block all dependence between the long-term outcome $Y$ and the treatment $A$ given the covariates $X$, ruling out any unmediated direct effect of the treatment on the long-term outcome and any confounding between short- and long-term outcomes. In contrast, our \Cref{assump: CI} requires the potential short-term outcomes $S_2(a)$ to block the dependence between the potential outcomes $S_1(a)$ and $S_3(a), Y(a)$. This puts no restrictions on the treatment effect on the long-term outcome or the confounding between short- and long-term outcomes. 
See \Cref{sec: validation} for a diagram illustration that compares our assumptions with the surrogacy assumption in \cite{athey2019surrogate}. 
Moreover, previous  work usually view  multiple short-term outcomes as a single vector without any internal structure \citep[e.g.,~][]{athey2019surrogate,athey2020combining,kallus2020role}. In contrast, we introduce a sequential internal structure among the surrogates and the long-term outcome  (\Cref{assump: CI})  to address the challenge of persistent confounding, as we will demonstrate in the next section. 
}

We further assume that short-term outcomes $\prns{\Sa, \Sc}$ are sufficiently affected by the unobserved confounders $U$, formalized in the following completeness conditions.

\begin{assumption}[Completeness Conditions]\label{assump: completeness} 
For any $s_2 \in \Scal_2$, $a \in \braces{0, 1}$, $x \in \Xcal$, 
\begin{enumerate}
\item \label{assump: completeness-1} If $\Eb{g\prns{U} \mid \Sc, \Sb = s_2, A = a, X = x, G = O} = 0$ holds almost surely, then $g\prns{U} = 0$ almost surely.
\item \label{assump: completeness-2} If $\Eb{g\prns{U} \mid \Sa, \Sb = s_2, A = a, X = x, G = O} = 0$ holds almost surely, then $g\prns{U} = 0$ almost surely.
\end{enumerate}
\end{assumption}
These completeness conditions require that the short-term outcomes $\prns{\Sa, \Sc}$ are strongly dependent with the unobserved confounders, and they have  sufficient variability
relative to the unobserved confounders $U$. 
Under these conditions, $\prns{\Sa, \Sc}$ can be viewed as strong proxy variables\footnote{Note that we do not require $S_2$ to be strong proxy variables for the unobserved confounders $U$. Instead, we only require $S_2$ to block the path between $S_1$ and $S_3$ (see \Cref{figure: DAG-b}), so that \Cref{assump: CI} is plausible.} for the unobserved confounders $U$.
Completeness assumptions have been used extensively in recent literature on proximal causal inference  \citep{Miao2016,shi2020multiply,miao2018a,cui2020semiparametric,kallus2021causal}. 
However, these literature require proxy variables that are not causally affected by the treatment, termed negative controls. 
In contrast, here both of $\prns{\Sa, \Sc}$ can be  affected by the treatment and thus do not directly fit into this previous literature. 

{While in the main text we simply consider a single set of unobserved confounders $U$ that affect everything,
in \Cref{sec: partial} we further dissect persistent confounders into groups of variables and show that some unobserved confounders can be ignored and simply excluded from $U$, relaxing some of the above assumptions.}

\section{Identification}\label{sec: identification}
In this section, we establish three novel identification strategies for the average long-term treatment effect in presence of general persistent confounding. 

\subsection{Identification via Outcome Bridge Function}\label{sec: identify-OBF}
We first introduce the concept of an outcome bridge function, which will play an important role in our first identification strategy. 

\begin{assumption}[Outcome Bridge Function]\label{assump: bridge}
There exists an outcome bridge function $h_0: \Scal_3 \times \Scal_2 \times \Acal \times \Xcal \to \Rl$ defined as follows: 
\begin{align}\label{eq: bridge-U}
\Eb{Y \mid \Sb, A, U, X, G = O} = \Eb{h_0\prns{\Sc, \Sb, A, X} \mid \Sb, A, U, X, G = O}.
\end{align}
\end{assumption}
According to \Cref{eq: bridge-U}, an outcome bridge function $h_0$ gives a transformation of short-term outcomes $\prns{S_3, S_2}$, treatment $A$, and covariates $X$, such that the confounding effects of the unmeasured variables $U$  on this transformation can reproduce those on the long-term outcome $Y$. So we can expect outcome bridge functions to be useful in tackling  unmeasured confounding. 

In general nonparametric settings, \Cref{assump: bridge} holds as a consequence of \Cref{assump: completeness} condition \ref{assump: completeness-1} and some additional technical conditions.
See \Cref{sec: existence} for details. 
In some special cases detailed below, we can both directly guarantee \Cref{assump: bridge} and describe the functional form of outcome bridge functions. 

\begin{example}[Discrete Setting]\label{ex: discrete}
Suppose that $\Scal_1 = \Scal_2 = \Scal_3 = \braces{s_{(j)}: j = 1, \dots, M_s}$ and $\Ucal = \braces{u_{(k)}: k = 1, \dots, M_u}$. For any $s_2 \in \Scal_2, a \in\Acal, x\in\Xcal$, let $\Eb{{Y} \mid s_2, a, \mathbf{U}, x} \in \R{M_u}$ denote the vector  whose $k$th element is $\Eb{Y \mid S_2 = s_2, A = a, U = u_{(k)}, X= x, G = O}$ and $P(\mathbf{S}_3 \mid s_2, a, \mathbf{U}, x) \in \R{M_s\times M_u}$ the matrix whose $(j, k)$th element is 
$$\Prb{S_3 = s_{(j)} \mid S_2 = s_2, A = a, U = u_{(k)}, X = x, G = O}.$$
The existence of an outcome bridge function in \Cref{assump: bridge} is equivalent to the existence of a solution $z \in \R{M_s}$ to the following linear equation system for any $s_2 \in \Scal_2, a \in\Acal, x\in\Xcal$:
\begin{align}\label{eq: discrete-sol}
{P(\mathbf{S}_3 \mid s_2, a, \mathbf{U}, x)}^\top z = \Eb{{Y} \mid s_2, a, \mathbf{U}, x}
\end{align}
A sufficient condition for the existence of solutions to \Cref{eq: discrete-sol} is that the matrix $P(\mathbf{S}_3 \mid s_2, a, \mathbf{U}, x)$ has a full column rank for any $s_2 \in \Scal_2, a \in\Acal, x\in\Xcal$. 
This full column rank condition means that $S_3$ are strongly dependent with $U$ and it requires that the number of possible values of $S_3$ (i.e., $M_s$) is no smaller than the number of possible values of $U$ (i.e., $M_u$). 
In this example, the full column rank sufficient condition is equivalent to the completeness condition in \Cref{assump: completeness} condition \ref{assump: completeness-1}.
\end{example}

\begin{example}[Linear Model]\label{ex: linear}
Suppose that $\prns{Y, S_3, S_2, S_1}$ are generated from the following linear structural equation system: 
\begin{align*}
&Y = \tau_y A + \alpha_y^\top S_3 + \beta_y^\top X + \gamma_y^\top U + \epsilon_y, \\ 
&S_j = \tau_j A + \alpha_j S_{j-1} + \beta_j X + \gamma_j U + \epsilon_j, ~ j \in \braces{3, 2}\\
&S_1 = \tau_1 A +  \beta_1 X + \gamma_1 U + \epsilon_1,
\end{align*}
where $\tau_y, (\tau_j, \alpha_y, \beta_y, \gamma_y), (\alpha_j, \beta_j, \gamma_j)$ are scalars, vectors, and matrices of conformable sizes, respectively, and $\epsilon_y, \epsilon_j$ are independent mean-zero noise terms such that $\epsilon_y \perp (S, A, U, X)$ and $\epsilon_j \perp (S_{j-1}, \dots, S_1, A, U, X)$. 
\Cref{assump: bridge} holds if there exists a solution $\omega$ to the linear equation $\gamma_3^\top\omega = \gamma_y$, since for any such $\omega$, 
it can be easily shown that a valid outcome bridge function is
\begin{align*}
h_0\prns{s_3, s_2, a, x} = \theta_3^\top s_3 + \theta_2^\top s_2 + \theta_1 a + \theta_0^\top x,
\end{align*}
where $\theta_3 = \omega+ \alpha_y, \theta_2 = -\alpha_3^\top \omega, \theta_1 = \tau_y - \tau_3^\top \omega, \theta_0 = \beta_y - \beta_3^\top\omega$.
Therefore, a sufficient condition for the existence of outcome bridge functions is that $\gamma_3$ has a full column rank.
This full-column-rank condition again means that $S_3$ is sufficiently informative for the unobserved confounders $U$.
\end{example}

Note that outcome bridge functions in \Cref{eq: bridge-U} are defined in terms of unobserved confounders, so we cannot directly use this definition to learn outcome bridge functions from observed data. In the following lemma, we give an alternative characterization of outcome bridge functions, only in terms of distributions of observed data.
\begin{lemma}\label{lemma: bridge-obs}
Under {\Cref{assump: CI,assump: unconfound-obs}, the completeness condition in \Cref{assump: completeness} condition \ref{assump: completeness-2}}, any function $h_0$ that satisfies 
\begin{align}\label{eq: bridge-obs}
&\Eb{Y \mid \Sb, \Sa, A, X, G = O}   = \Eb{h_0\prns{\Sc, \Sb, A, X} \mid \Sb, \Sa, A, X, G = O}
\end{align}
is also a valid outcome bridge function in the sense of  \Cref{eq: bridge-U}.
\end{lemma}
In \Cref{lemma: bridge-obs}, we assume the completeness condition in \Cref{assump: completeness} condition \ref{assump: completeness-2}, which requires the short-term outcomes $S_1$ to be  informative enough for the unobserved confounders $U$. 
Under this additional assumption, outcome bridge functions can be equivalently characterized by the conditional moment equation in \Cref{eq: bridge-obs}. 
Note that \Cref{eq: bridge-obs} simply replaces the unobserved confounders $U$ in \Cref{eq: bridge-U} by the observed short-term outcomes $S_1$. 
The resulting conditional moment equation only depends on observed variables. 

We finally establish the identification of the average long-term treatment effect in the following theorem. 
\begin{theorem}\label{thm: identification1}
Under the conditions of \Cref{lemma: bridge-obs}, the average long-term treatment effect is identifiable: for any function $h_0$ satisfying \Cref{eq: bridge-obs}, 
{at least one of which exists}, we have
\begin{align}\label{eq: identification-1}
\tau 
    &= \Eb{h_0\prns{\Sc, \Sb, A, X} \mid A = 1, G = E}  - \Eb{h_0\prns{\Sc, \Sb, A, X} \mid A = 0, G = E}.
\end{align}
\end{theorem}

\Cref{thm: identification1} states that the average long-term treatment effect can be recovered by marginalizing \emph{any} outcome bridge function (which is defined on the observational data distribution) over the experimental data distribution. 
This shows how observational and experimental data can be combined together to identify the long-term treatment effect. 

\begin{remark}[Connection to \cite{athey2020combining}]\label{remark: connection}
The  proposed identification strategy in \Cref{eq: identification-1} can be viewed as a generalization of that in \cite{athey2020combining}. When there only exist short-term confounders, \cite{athey2020combining} shows that we only need a single group of short-term outcomes. We can let $S_1 = S_3 = \emptyset$ and  $S = S_2$, then $h_0\prns{S_2, A, X} = \Eb{Y \mid S, A, X, G=O}$ is the unique solution to \Cref{eq: bridge-obs}, and it can be plugged into \Cref{eq: identification-1} to identify the average long-term treatment effect. This recovers the identification strategy in Theorem 1 of \cite{athey2020combining} when specialized to the case of \Cref{assump: ext-valid} (\Cref{corollary: covariate-exp} in \Cref{sec: covariate-adaptive-identify} recovers it in the general case; see discussions therein).
Of course, when persistent confounding is present this identification fails. Instead, \Cref{thm: identification1} provides a more general identification strategy that can leverage structure in the surrogates to handle persistent confounders.
\end{remark}

\subsection{Identification via Selection Bridge Function}\label{sec: identify-SBF}

The second identification strategy involves an alternative bridge function below.

\begin{assumption}[Selection Bridge Function]\label{assump: bridge2}
There exists a selection bridge function $q_0: \Scal_2 \times \Scal_1 \times \Acal \times \Xcal \to \Rl$ defined as follows:
\begin{align}\label{eq: bridge2-U}
\frac{p\prns{\Sb, U, X\mid A, G = E}}{p\prns{\Sb, U, X \mid A, G = O}}  = \Eb{q_0\prns{\Sb, \Sa, A, X} \mid \Sb, A, U, X, G = O}.
\end{align}
\end{assumption}

According to \Cref{eq: bridge2-U}, a selection bridge function $q_0$ gives a transformation of short-term outcomes $\prns{S_2, S_1}$, treatment $A$, and covariates $X$, which can adjust for distributional differences between the experimental and observational data. In \Cref{sec: support} \Cref{lemma: data-overlap}, we prove that under \Cref{assump: ext-valid}, the density ratio in left hand side of \Cref{eq: bridge2-U} is almost surely finite, so \Cref{eq: bridge2-U} is well-defined. 

In general nonparametric models, the existence of a selection bridge function can be ensured by the completeness condition in \Cref{assump: completeness} condition \ref{assump: completeness-2} and some additional technical conditions. See \Cref{sec: existence} for details. 
This means that a selection bridge function exists when the short-term outcomes $S_1$ are sufficiently informative for the unobserved confounders $U$.  We can also derive more specialized existence conditions for \Cref{ex: linear,ex: discrete} (see \Cref{sec: selection-example}). 

Again, selection bridge functions in \Cref{eq: bridge2-U} are defined in terms of unobserved confounders.
Below, we derive alternative characterizations in terms of distributions of observed variables.

\begin{lemma}\label{lemma: bridge2-obs}
Under \cref{assump: CI,assump: ext-valid,assump: unconfound-exp,assump: unconfound-obs}, the completeness condition in \Cref{assump: completeness}  condition \ref{assump: completeness-1}, any function $q_0$ that satisfies 
\begin{align}\label{eq: bridge2-obs-1}
&\frac{p\prns{\Sc, \Sb, X\mid A, G = E}}{p\prns{\Sc, \Sb, X \mid A, G = O}}  = \Eb{q_0\prns{\Sb, \Sa, A, X} \mid \Sc, \Sb, A, X, G = O}
\end{align}
is also a valid selection bridge function in the sense of  \Cref{eq: bridge2-U}.
\end{lemma}

In \Cref{lemma: bridge2-obs}, we assume the completeness condition in \Cref{assump: completeness}  condition \ref{assump: completeness-1}, which requires the short-term outcomes $S_3$ to be  informative enough for the unobserved confounders $U$. 
Under this additional assumption, selection bridge functions can be equivalently characterized by the conditional moment equation in \Cref{eq: bridge2-obs-1}, which involves only observed variables.
\Cref{eq: bridge2-obs-1} is a direct analogue to \Cref{eq: bridge2-U}, replacing $U$ in \Cref{eq: bridge2-U} by $S_3$ in \Cref{eq: bridge2-obs-1}. We can also equivalently express \Cref{eq: bridge2-obs-1} as follows
\begin{align}\label{eq: bridge2-obs-2}
\Eb{\indic{G = O}\prns{\frac{\Prb{G = E \mid A}}{\Prb{G = O \mid A}}q_0\prns{\Sb, \Sa, A, X} + 1} \mid \Sb, \Sa, A, X} = 1.
\end{align}
\Cref{eq: bridge2-obs-2} is a more convenient formulation for estimation as it does not involve any conditional density function.

\begin{theorem}\label{thm: identification2}
Under conditions in \Cref{lemma: bridge2-obs}, the average long-term treatment effect is identifiable: for any function $q_0$ that satisfies  \Cref{eq: bridge2-obs-1} or \Cref{eq: bridge2-obs-2}, at least one of which exists, we have 
\begin{align}\label{eq: identification-2}
\tau 
    &= \Eb{q_0\prns{\Sb, \Sa, A, X}Y \mid A = 1, G = O} 
    - \Eb{q_0\prns{\Sb, \Sa, A, X}Y \mid A = 0, G = O}.
\end{align}
\end{theorem}

\Cref{thm: identification2} states that the average long-term treatment effect can be also identified by \emph{any} selection bridge function.
This provides an alternative to the identification strategy based on outcome bridge functions in \Cref{thm: identification1}. 

\begin{remark}[Comparison with Proximal Causal Inference]\label{remark: proximal}
As discussed in \Cref{sec: proximal}, our identification is related to identification in the proximal causal inference literature. 
Indeed, we also take a proxy-variable perspective, viewing short-term outcomes $(S_1, S_3)$ as proxy variables for the unobserved confounders $U$. 
Moreover, the characterization for outcome bridge function $h_0$ given in \Cref{eq: bridge-obs} has an analogue in \cite{miao2018a}.

Nevertheless, our setting is substantially different from the existing proximal causal inference literature. 
The short-term outcomes $(S_1, S_3)$ are both affected by the treatment, so they do not satisfy the proxy conditions in \cite{Miao2016}. 
Our identification strategies also feature a novel use of the experimental data. This is crucial in our setting, whereas proximal causal inference focuses on observational data only. 
Notably,  our identification in \Cref{thm: identification2} relies on a new selection bridge function. 
This bridge function, as defined in \cref{eq: bridge2-obs-1}, is specialized to our data combination setting, 
without {direct} analogue in the existing proximal causal inference literature {except the concurrent work \cite{ghassami2022combining}}. 
\end{remark}

\begin{remark}[\Cref{assump: completeness,assump: bridge,assump: bridge2} and the Conditioning on $S_2$]
{In \Cref{assump: completeness} we assume two completeness conditions and in \Cref{assump: bridge,assump: bridge2}, we assume the exsitence of outcome and selection bridge functions. These conditions roughly require $S_1, S_3$ to be strongly dependent with the unobserved confounders $U$ after accounting for $S_2, A$ and  $X$. Since $S_2$ also tend to be dependent
with $U$, conditioning on $S_2$ may explain away part of the dependence between $S_1, S_3$ and $U$. Thus \Cref{assump: completeness,assump: bridge,assump: bridge2} may be at risk if $S_2$ include very rich short-term outcomes and capture a very large amount of variations in $U$. They are more plausible as $S_1, S_3$ include richer informative short-term outcomes relative to $S_2$.}
\end{remark}

\subsection{Doubly Robust Identification}\label{sec: identify-DR}
In \Cref{sec: identify-OBF,sec: identify-SBF}, we present two different identification strategies, based on outcome bridge functions and selection bridge functions, respectively. 
We now combine them into a doubly robust identification strategy. 

\begin{theorem}\label{thm: identification-DR}
Fix functions $h: \Scal_3 \times \Scal_2 \times \Acal \times \Xcal \to \Rl$ and $q: \Scal_2 \times \Scal_1 \times \Acal \times \Xcal \to \Rl$. 
If either conditions in \Cref{thm: identification1} hold and $h = h_0$ satisfies \cref{eq: bridge-obs}, or conditions in \Cref{thm: identification2} hold and $q = q_0$ satisfies \cref{eq: bridge2-obs-1} or \cref{eq: bridge2-obs-2}, then the average long-term treatment effect is identified as:
\begin{equation}\label{eq: DR}\begin{aligned}
\tau
    &= \sum_{a\in\braces{0, 1}}(-1)^{1-a}\Eb{h\prns{\Sc, \Sb, A, X} \mid A = a, G = E}  \\
    &\phantom{=}+ \sum_{a\in\braces{0, 1}}(-1)^{1-a}\Eb{q\prns{\Sb, \Sa, A, X}\prns{Y - h\prns{\Sc, \Sb, A, X}} \mid A = a, G = O}.
\end{aligned}\end{equation}
\end{theorem}

\Cref{thm: identification-DR} shows that \Cref{eq: DR} identifies the average long-term treatment effect when it uses either a valid outcome bridge function or a valid selection bridge function. 
But it does not need both bridge functions to be valid. This is why it is called doubly robust.

\section{Estimation and Inference}\label{sec: estimation}
In this section, we provide three different estimators for the average long-term treat effect, corresponding to the three different identification strategies in \Cref{sec: identification} respectively. 
This involves combining two samples, so we assume that as $n \to \infty$, $n_E/n_O \to \lambda$ where $0 < \lambda < \infty$. This is a common assumption in the  data combination literature \cite[e.g., ][]{angrist1992effect,graham2016efficient}.

In order to estimate the average long-term  treatment effect, we need to  first estimate the outcome and/or selection bridge functions. 
Estimating these bridge functions amounts to solving the conditional moment equations in \Cref{eq: bridge2-obs-2,eq: bridge-obs} based on a finite sample of data{, which corresponds to an ill-posed inverse problem \citep{Carrasco2007}}. 
A variety of estimation strategies can be used for this task, which we review in \Cref{remark: bridge-est} below. 
For now, we consider any generic bridge function estimators, which may be any from those reviewed in \Cref{remark: bridge-est}, and discuss different ways to use these to construct the long-term treatment effect estimator.

Below, we define three different estimators for the counterfactual mean parameter $\mu(a)$, $a \in \Acal$. They all use the cross-fitting technique when constructing bridge function estimators, {so that the data used to construct the bridge function estimators are independent with the data at which the estimators are evaluated}. This technique has been widely used to accommodate complex nuisance function estimators while preserving strong asymptotic guarantees \citep[e.g., ][]{chernozhukov2019double,zheng2011cross}. {Note that we  need to split only the observational data and not the experimental data. This is because only the outcome bridge function is evaluated on the experimental data, and it is estimated only from observational data. The experimental and observational datasets are already independent, rendering the bridge-function estimate and the data on which it is evaluated independent without any cross-fitting.} 

\begin{definition}[Cross-fitted Counterfactual Mean Estimator]\label{def: mean-est}
Fix $a \in \mathcal{A}$ and an integer $K \ge 2$.
\begin{enumerate}
\item Randomly split the observational data $\mathcal{D}_O$ into $K$ (approximately) even folds, denoted as $\mathcal{D}_{O, 1}, \dots, \mathcal{D}_{O, K}$, respectively. 
\item For $k = 1, \dots, K$, use all observational data other than the $k$th fold, i.e., $\mathcal{D}_{O, -k} \coloneqq \cup_{j \ne k} \mathcal{D}_{O, j}$,  to construct the outcome bridge function estimator based on \Cref{eq: bridge-obs} and/or the selection bridge function estimator based on \Cref{eq: bridge2-obs-2}. Denote them as $\hk(S_3, S_2, A, X)$ and $\qk(S_2, S_1, A, X)$, respectively. 
\item Use any of the following counterfactual mean estimators:
\begin{align*}
&\hat\mu_\out(a) = \frac{1}{K}\sum_{k=1}^K \bracks{\frac{1}{\nEa}\sum_{i \in \mathcal{D}_E} \indic{A_i = a}{\hk\prns{S_{3, i}, S_{2, i}, A_i, X_i}}}, \\
&\hat\mu_\sel(a) = \frac{1}{K}\sum_{k=1}^K \bracks{\frac{1}{\nOka}\sum_{i \in \mathcal{D}_{O, k}}\indic{A_i = a}{\qk\prns{S_{2, i}, S_{1, i}, A_i, X_i}Y_i}}, \\
&\hat\mu_\dr(a)  =  \frac{1}{K}\sum_{k=1}^K \bracks{\frac{1}{\nEa}\sum_{i \in \mathcal{D}_E} \indic{A_i = a}{\hk\prns{S_{3,i}, S_{2, i}, A_i, X_i}}}\\
    &\phantom{\hat\mu_\dr(a)=} + 
    \frac{1}{K}\sum_{k=1}^K \bracks{\frac{1}{\nOka}\sum_{i \in \mathcal{D}_{O, k}}\indic{A_i = a}{\qk\prns{S_{2, i}, S_{1, i}, A_i, X_i}\prns{Y_i - \hk\prns{S_{3, i}, S_{2, i}, A_i, X_i}}}},
\end{align*}
where $\nEa = \sum_{i \in \mathcal{D}_E} \indic{A_i = a}$ and $\nOka = \sum_{i \in \mathcal{D}_{O, k}} \indic{A_i = a}$ are the numbers of units with treatment level $a$ in the experimental data $\mathcal{D}_E$ and the $k$-th fold of observational data $\mathcal{D}_{O, k}$, respectively. 
\end{enumerate}
\end{definition}
Based on the counterfactual mean estimators in \Cref{def: mean-est}, we can construct average long-term treatment effect estimators:
\begin{align*}
\hat\tau_\out = \hat\mu_\out(1) - \hat\mu_\out(0), ~~ \hat\tau_\sel = \hat\mu_\sel(1) - \hat\mu_\sel(0), ~~ \hat\tau_\dr = \hat\mu_\dr(1) - \hat\mu_\dr(0).
\end{align*}

{To analyze the asymptotic properties of these treatment effect estimators, we need to impose some high level conditions on the estimation errors of the bridge function estimators.
Since these estimators solve ill-posed conditional moment equations, we quantify their estimation errors in terms of both weak metrics and the strong metrics, as this is a common practice in the literature \citep[e.g., ][]{chen2012estimation,DikkalaNishanth2020MEoC,bennett2023minimax}. In particular, we define a projection operator $T$ and its adjoint operator $T^\star$ given by $[Th](S_2, S_1, A, X) = \Eb{h(S_3, S_2, A, X) \mid S_2, S_1, A, X, G = O}$ and $[T^\star q](S_3, S_2, A, X) = \Eb{q(S_2, S_1, A, X) \mid S_3, S_2, A, X, G = O}$. 
For a given outcome bridge function estimator $\hat h$ and a given selection bridge function estimator $\hat q$, we can quantify their estimation errors relative to $h$ and $q$ in terms of the weak metrics  $\|T(\hat h - h)\|_{\mathcal{L}_2(\mathbb{P})}$ and $\|T^\star (\hat q - q)\|_{\mathcal{L}_2(\mathbb{P})}$ respectively. 
We can also quantify their estimation errors in terms of the strong metrics $\|\hat h - h\|_{\mathcal{L}_2(\mathbb{P})}$ and $\|\hat q - q\|_{\mathcal{L}_2(\mathbb{P})}$ respectively. 
The strong-metric errors can be much larger (even infinitely larger) than the corresponding weak-metric errors due to ill-posedness of the conditional moment equations. {See also  \cref{remark: bridge-est} for more discussions on the error rates. }
} 

{
\begin{assumption}[Error Rates of Bridge Function Estimators]\label{assump: rate}
\begin{enumerate}
 \item \label{assump: rate-h} There exist $\tilde h \in \Scal_3 \times \Scal_2 \times \Acal \times \Xcal \to \Rl$ and sequences $\delta_{h, n} \to 0$ and $\rho_{h, n} \to 0$ such that 
\begin{align*}
 \|T(\hk - \tilde h) \|_{\mathcal{L}_2\prns{\mathbb{P}}} = O_{\mathbb{P}}\prns{\delta_{h, n}}, ~ \|\hk - \tilde h \|_{\mathcal{L}_2\prns{\mathbb{P}}} = O_{\mathbb{P}}\prns{\rho_{h, n}}, ~~ \forall k \in \braces{1, \dots, K}.
\end{align*}
 \item \label{assump: rate-q} There exist $\tilde q \in \Scal_2 \times \Scal_1 \times \Acal \times \Xcal \to \Rl$ and sequences $\delta_{q, n} \to 0$ and $\rho_{q, n} \to 0$ such that 
 \begin{align*}
  \|T^\star\prns{\qk - \tilde q}\|_{\mathcal{L}_2\prns{\mathbb{P}}}  =  O_{\mathbb{P}}\prns{\delta_{q, n}}, ~  \|\qk - \tilde q\|_{\mathcal{L}_2\prns{\mathbb{P}}}  =  O_{\mathbb{P}}\prns{\rho_{q, n}}, ~~ \forall k \in\braces{ 1, \dots, K}.
 \end{align*}
 \end{enumerate} 
\end{assumption}
}

{
\Cref{assump: rate} specifies that the outcome bridge function estimator and selection bridge function estimator converge to some limits $\tilde h$ and $\tilde q$ respectively, in terms of both weak metrics and strong metrics. 
Note that we do not necessarily require these estimators to be consistent, i.e., we allow $\tilde h \ne h_0$ or $\tilde q \ne q_0$, as we show in the following theorem. 
}

\begin{theorem}[Estimation Consistency]\label{thm: consistency}
\begin{enumerate}
\item \label{thm: consistency-h} If conditions in \Cref{thm: identification1} 
and  \Cref{assump: rate} condition \ref{assump: rate-h} hold, and $\tilde h = h_0$, then $\hat\tau_\out$ consistent. 
\item \label{thm: consistency-q} If conditions in \Cref{thm: identification2}
and  \Cref{assump: rate} condition \ref{assump: rate-q} hold, and $\tilde q = q_0$, then $\hat\tau_\sel$ is consistent. 
\item \label{thm: consistency-dr} If the conditions in either of the two  statements above hold, then $\hat\tau_\dr$ is consistent.
\end{enumerate}
\end{theorem}

\Cref{thm: consistency} shows that if the outcome bridge function estimator is consistent (i.e., $\tilde h = h_0$), then the corresponding treatment effect estimator $\hat\tau_\out$ is consistent. 
Similarly, if the selection bridge function estimator is consistent (i.e., $\tilde q = q_0$), then the corresponding treatment effect estimator $\hat\tau_\sel$ is also consistent. 
In contrast, the estimator $\hat\tau_\dr$ is more robust, in that it is consistent if \emph{either} of the two bridge function estimators is consistent. 

\Cref{thm: consistency} establishes the consistency of treatment effect estimators given only high level conditions on the bridge function estimators, regardless of how they are actually constructed. 
However, the actual ways to construct bridge function estimators generally do impact 
the asymptotic distributions of estimators $\hat\tau_\out$ and $\hat\tau_\sel$.
So we only focus on the asymptotic distribution of estimator $\hat\tau_\dr$, which can be derived even under generic high level conditions. 
\begin{theorem}[Asymptotic Distribution of Doubly Robust Estimator]\label{thm: dist-dr}
Suppose conditions in \Cref{thm: consistency} statements \ref{thm: consistency-h} and \ref{thm: consistency-q} hold and {$\min\braces{\delta_{h, n}\rho_{q, n}, \rho_{h, n}\delta_{q, n}} = o(n^{-1/2})$}. Then as $n \to \infty$, 
\begin{align*}
\sqrt{n}\prns{\hat\tau_{\dr}-\tau}  \rightsquigarrow  \mathcal{N}\prns{0, \sigma^2}, 
\end{align*}
where 
\begin{align*}
\sigma^2 
    &= \frac{1+\lambda}{\lambda}\Eb{\prns{\frac{A - \Prb{A=1\mid G=E}}{\Prb{A=1\mid G=E}}\prns{ h_0(\Sc, \Sb, A, X) - \mu(A)} }^2 \mid G = E}  \\
    &\phantom{=}+ \prns{1+\lambda} \Eb{\prns{\frac{A - \Prb{A=1\mid G=O}}{\Prb{A=1\mid G=O}}q_0\prns{\Sb, \Sa, A, X}\prns{Y - h_0(\Sc, \Sb, A, X)}}^2 \mid G = O}. \nonumber 
\end{align*}
\end{theorem}

\Cref{thm: dist-dr} shows that if both bridge function estimators are consistent (i.e., $\tilde h = h_0$ and $\tilde q = q_0$), and the product of their convergence rates {in terms of one strong-metric error and one weak-metric error} is $o(n^{-1/2})$, then the doubly robust treatment effect estimator $\hat\tau_\dr$ is asymptotically normal with a closed-form asymptotic variance.
{Note that the rate condition is weaker than requiring the product of two strong-metric error rates to be $o(n^{-1/2})$.}
We can easily estimate this asymptotic variance by plugging estimates into all unknowns therein:
\begin{align*}
    \hat\sigma^2 
        &=  \frac{n}{n_EK}\sum_{k=1}^K \braces{\frac{1}{n_E}\sum_{i \in \mathcal{D}_E} \bracks{\frac{A_i - \hat\pi_E}{\hat\pi_{E}}\prns{\hk\prns{S_{3,i}, S_{2, i}, A_i, X_i}-\hat\mu_\dr(A_i)}}^2} \\
        &+ \frac{n}{n_OK}\sum_{k=1}^K \braces{\frac{1}{n_{O, k}}\sum_{i \in \mathcal{D}_{O, k}}\bracks{\frac{A_i - \hat\pi_O}{\hat\pi_O}\qk(S_{2, i}, S_{1, i}, A_i, X_i)\prns{Y_i - \hk\prns{S_{3, i}, S_{2, i}, A_i, X_i}}}^2},
\end{align*}
where $\hat\pi_E$ and $\hat\pi_O$ are sample frequency estimates for $\Prb{A=1\mid G=E}$ and $\Prb{A=1\mid G=O}$ respectively.
Then we can accordingly construct confidence intervals.
\begin{theorem}[Confidence Interval]\label{thm: CI}
Under conditions in \Cref{thm: dist-dr}, the confidence interval 
\begin{align*}
    \op{CI} = \bracks{\hat\tau_\dr - \Phi^{-1}(1-{\alpha}/{2})\hat\sigma/\sqrt{n}, ~ \hat\tau_\dr + \Phi^{-1}(1-{\alpha}/{2})\hat\sigma/\sqrt{n}}
\end{align*}
satisfies that $\Prb{\tau \in \op{CI}} \to 1-\alpha$ as $n \to \infty$.
\end{theorem}

In the following theorem, we further show that the   asymptotic variance in \Cref{thm: dist-dr} actually attains the local semiparametric efficiency lower bound, provided that the bridge functions uniquely exist and an additional regularity condition holds. 

\begin{theorem}[Asymptotic Efficiency]\label{thm: efficiency}
{Let $\mathbb{P}$ be a distribution instance such that \Cref{assump: bridge2,assump: bridge} hold with unique bridge functions and the corresponding linear operator $T$ defined above \Cref{assump: rate}  is bijective. 
Then, 
the efficiency lower bound for the average long-term treatment effect $\tau$ under \Cref{assump: CI,assump: unconfound-obs,assump: unconfound-exp,assump: ext-valid,assump: bridge},  locally evaluated at the distribution $\mathbb{P}$,  is equal to  $\sigma^2$ given in \Cref{thm: dist-dr}.}
\end{theorem}

\Cref{thm: efficiency} implies that under the asserted assumptions, treatment effect estimator 
$\hat\tau_\dr$ is asymptotically optimal, in the sense that it achieves the smallest asymptotic variance among all regular and asymptotically linear estimators \citep{van2000asymptotic}. 

\begin{remark}[Bridge Function Estimators {and Their Estimation Errors}]\label{remark: bridge-est}
Estimating bridge functions amounts to estimating roots of the conditional moment equations in \Cref{eq: bridge2-obs-2,eq: bridge-obs}. 
This can be implemented by many methods. Examples include Generalized Method of Moments (GMM) \citep[e.g., ][]{miao2018a,cui2020semiparametric,hansen1982large}, sieve methods \citep[e.g., ][]{ai2003efficient,newey2003instrumental,Hall05IV}, kernel density estimators \citep[e.g., ][]{darolles2010nonparametric,Hall05IV}, Reproducing Kernel Hilbert Space methods \citep[e.g., ][]{NEURIPS2019_17b3c706,GhassamiAmirEmad2021MKML}, neural network methods \citep[e.g., ][]{pmlr-v70-hartford17a,NIPS2019_8615}, and more generally, adversarial learning methods \citep[e.g., ][]{bennett2020variational,DikkalaNishanth2020MEoC,kallus2021causal,bennett2023minimax}. 
We can use any of these to estimate the bridge functions. 

{The estimation errors of these  estimators are typically characterized by  weak metrics or strong metrics in the sense of \Cref{assump: rate}. Weak-metric errors quantify the degree to which estimators violate the conditional moment equations, whereas strong-metric errors measure the deviation of estimators from specific solutions to these equations. In cases of highly ill-posed conditional moment equations, an estimator may nearly satisfy the conditional moment equation (i.e., exhibit small weak-metric error) yet still diverge significantly in function values from any solution to the equation (i.e., exhibit large strong-metric error).
The difference between these two types of errors reflects the level of ill-posedness. Weak-metric error rates often resemble those found in regular regression estimation and they are readily available in many of existing works. However, strong-metric error rates generally need additional restrictions on the conditional moment equations' ill-posedness \citep{chen2012estimation,DikkalaNishanth2020MEoC,bennett2023minimax}.}

{
 \Cref{thm: dist-dr} shows that the asymptotic normality of the doubly robust estimator needs the product of one strong-metric error rate and one weak-metric error rate to be $o(n^{-1/2})$. 
 This mirrors the product rate conditions found in many existing doubly robust estimators \citep{chernozhukov2018double}, which, however, do not distinguish between weak-metric and strong-metric errors due to their focus on well-posed regression functions. Our approach differentiates between these errors since we consider ill-posed conditional moment equations.
 The condition of a fast product rate between one weak-metric error and one strong-metric error implies a trade-off between the two rates,  allowing for the ill-posedness to impact only one of the two bridge function estimators. Consequently, even if one bridge function involves a very ill-posed conditional moment equation so it has a slow strong-metric error rate, the product rate condition could be still satisfied if the other bridge function has a fast weak-metric error rate. Moreover, the condition automatically adapts to the best product of the two types of errors. Similar conditions on the product of weak-metric and strong-metric errors also appear in other recent literature for inference on parameters related to ill-posed conditional moment equations \cite[e.g.,][]{singh2021finite,bennett2022inference,bennett2023source}. 
}
\end{remark}

\begin{remark}[Non-uniqueness of Bridge Functions]
In \Cref{thm: efficiency}, we assume that bridge functions uniquely exist, which is not necessarily true in practice. As we discussed in \Cref{sec: identify-OBF,sec: identify-SBF}, bridge functions exist if the short-term outcomes $S_1$ and $S_3$ are sufficiently informative for the unobserved confounders.
But when they are more informative than necessary, there may exist more than one bridge function. 
For example, in \Cref{ex: discrete}, when the matrix $P(\mathbf{S}_3 \mid s_2, a, \mathbf{U}, x)$ has full column rank and $S_3$ has more values than the unobserved confounders $U$, \Cref{eq: discrete-sol} admits many solutions $z$ and each of them corresponds to a different outcome bridge function.

The non-uniqueness of bridge functions has important implications for asymptotic properties of treatment effect estimators. 
Almost all previous results in proximal causal inference assume unique bridge functions when studying statistical inference. 
One exception is the penalized GMM estimator in \cite{imbens2021controlling}, 
which leverages penalization to power inference even with non-unique bridge functions.
But their approach only applies to parametric estimation of bridge functions.
{\cite{bennett2022inference}  proposes methods for inference on functionals of solutions to weakly identified nonparametric conditional moment equations, and consider proximal causal inference with non-unique bridge functions as a canonical example.}
\end{remark}

\section{Extensions}\label{sec: extension}
In this section, we extend our previous identification results. 
{We first relax \Cref{assump: unconfound-exp,assump: ext-valid} in \Cref{sec: relax-assump}. Then in \Cref{sec: control-fun} we provide an alternative identification via control functions rather than bridge functions. This can identify not only the average long term treatment effect but also the entire distribution of the counterfactual long term outcomes.}

\subsection{Relaxing  \Cref{assump: unconfound-exp,assump: ext-valid}}\label{sec: relax-assump}
We now extend our 
identification results by relaxing 
 \Cref{assump: unconfound-exp,assump: ext-valid}. 
 In particular, we relax \Cref{assump: ext-valid} by allowing the covariate distribution to be different in the experimental and observational data. 
 This is an important extension because these two types of data are often collected from different environments, where the covariate distributions are  likely to be different. 
 {For example, because  observational data are usually easier to collect and have larger scale than experimental data, the observational covariate distribution may be more representative of the entire population of interest, while experimental data may only correspond to a selective sub-population.} 
 Therefore, we consider the following assumption  to allow for different covariate distributions in  two types of data. 
\begin{assumption}[External Validity, Modified]\label{assump: ext-valid2}
Suppose that for any $a \in \braces{0, 1}$,
\begin{align}\label{eq: external-validity-0b}
\prns{S\prns{a}, U} \perp G \mid X,
\end{align}
and \Cref{eq: data-overlap} holds almost surely.
\end{assumption}

Moreover, we relax \Cref{assump: unconfound-exp} by allowing the treatment assignment in the experimental data to depend on covariates $X$, instead of being completely at random. 
This permits us to also accommodate stratified randomized designs for the experimental data. 
\begin{assumption}[Experimental Data, Modified]\label{assump: unconfound-exp2}
Suppose that for any $a \in \braces{0, 1}$,
\begin{align}
\prns{Y\prns{a}, S\prns{a}, U} \perp A \mid X, G = E,
\end{align}
and $0 < \Prb{A = 1 \mid X, G = E} < 1$ almost surely. 
\end{assumption}

Below we extend the doubly robust identification  in \Cref{thm: identification-DR}, which shows that the long-term average treatment effect is identifiable under the weaker \Cref{assump: ext-valid2,assump: unconfound-exp2}.

\begin{theorem}\label{corollary: covariate-exp-dr}
Fix functions $h: \Scal_3 \times \Scal_2 \times \Acal \times \Xcal \to \Rl$, $q: \Scal_2 \times \Scal_1 \times \Acal \times \Xcal \to \Rl$, and denote $\bar h_{E}\prns{a, x} = \Eb{h\prns{\Sc, \Sb, A, X} \mid A = a, X = x, G = E}$. Suppose \Cref{assump: CI,assump: ext-valid2,assump: unconfound-exp2,assump: unconfound-obs} hold, and either of the following two conditions holds:
\begin{enumerate}
\item The completeness condition in \Cref{assump: completeness} condition \ref{assump: completeness-2} and \Cref{assump: bridge} hold, and $h = h_0$ satisfies \Cref{eq: bridge-obs};
\item The completeness condition in \Cref{assump: completeness} condition \ref{assump: completeness-1} and \Cref{assump: bridge2} hold, and $q = q_0$ satisfies \Cref{eq: bridge2-obs-1} or \Cref{eq: bridge2-obs-2}.
\end{enumerate}
Then the average long-term treatment effect is identified as:
\begin{align*}
\tau
    &= \sum_{a\in\braces{0, 1}}(-1)^{1-a}\bigg\{\Eb{\bar{h}_{E}\prns{a, X} \mid G= O}  \nonumber \\
    &\, +\Eb{\frac{\Prb{G=E}}{\Prb{G=O}}\frac{\indic{A=a}\nu(X)}{e_a(X)}\prns{h\prns{\Sc, \Sb, A, X} - \bar{h}_E\prns{A, X} }\mid G=E} \nonumber \\
    &\, +\Eb{\frac{\Prb{G=E\mid A=a}}{\Prb{G=O\mid A=a}}\frac{\indic{A=a}\nu(X)}{e_a(X)}q\left(S_{2}, S_{1}, A, X\right)\left(Y-h\left(S_{3}, S_{2}, A, X\right)\right)\mid G=O}\bigg\},
\end{align*}
where $\nu(X) = \Prb{G=O\mid X}/\Prb{G=E\mid X}$ and  $e_a(X) = \Prb{A = a \mid X, G = E}$.
\end{theorem}

\Cref{corollary: covariate-exp-dr} shows that even under the weaker \Cref{assump: ext-valid2,assump: unconfound-exp2}, outcome and selection bridge functions can  still be used to identify the average long-term treatment effect. 
This again has the doubly robust property in that it only requires one of the bridge functions to be correct rather than both. 
Compared to \Cref{thm: identification-DR}, 
\Cref{corollary: covariate-exp-dr} additionally incorporates the ratio $\nu(X) = \Prb{G=O\mid X}/\Prb{G=E\mid X}$ to adjust for the discrepancy in the covariate distribution of the two types of data (\Cref{assump: ext-valid2}).
It also uses the propensity score $e_a(X) = \Prb{A=a\mid X, G=E}$ to account for the dependence of treatment $A$ on covariates $X$ in the experimental data (\Cref{assump: unconfound-exp2}).

In \Cref{sec: covariate-adaptive-identify}, we further show that by setting $q = 0, h = h_0$ or $h = 0, q = q_0$ in \Cref{corollary: covariate-exp-dr}, we can obtain direct analogues of \Cref{thm: identification1,thm: identification2} that involve only a single bridge function.  
In \Cref{sec: covariate-adaptive-estimate}, we prove that the estimating equation based on the doubly robust identification strategy in \Cref{corollary: covariate-exp-dr} satisfies the \emph{Neyman orthogonality} property \citep{chernozhukov2019double}, and show that the resulting treatment effect estimator has appealing asymptotic properties and is amenable to inference. 

In \Cref{sec: more-extension},  we present some additional extensions.
In \Cref{sec: pretreat}, we extend our identification strategies to the setting where pre-treatment outcomes are available. In \Cref{sec: partial},  
we show that it is possible to relax completeness conditions in \Cref{assump: completeness} so that short-term outcomes need only be rich enough to capture some of the unobserved confounders rather than all of them. 

\begin{remark}[Connection to \cite{ghassami2022combining}]
{The doubly robust identification strategy in \Cref{corollary: covariate-exp-dr} and its close variants based on only a single bridge function (see \Cref{corollary: covariate-exp} in \Cref{sec: covariate-adaptive-identify}) have close analogues in the concurrent and independent work \cite{ghassami2022combining}. 
Specifically, the proximal data fusion identification strategies in \cite{ghassami2022combining} use a set of short-term outcomes $M$ and an additional set of proxies $Z$ that satisfy $Z \perp (M, Y) \mid A, X, U, G = O$. We note that under our sequential outcome condition in \Cref{assump: unconfound-obs,assump: CI}, we have $S_1 \perp (S_3, Y) \mid S_2, A, X, U, G = O$. The identification strategies in \cite{ghassami2022combining}, when their $Z$  and $M$ are replaced by $S_1$ and $S_3$ respectively and $S_2$ is additionally conditioned on everywhere, are actually equivalent to our identification strategies. 
Despite the close relations to \cite{ghassami2022combining}, our paper uniquely shows that short-term outcomes alone suffice for addressing unmeasured confounding and enables this by assuming a novel sequential outcome condition.
In addition, we also provide many additional extensions that have no analogues in \cite{ghassami2022combining}. See \Cref{sec: literature-data-comb} for a summary.} 
\end{remark}

{
\subsection{A Control Function Approach}\label{sec: control-fun}
In previous parts, we identify the long-term treatment effect using bridge functions. 
In this part, we provide an alternative identification approach based on a \emph{control function}. Control functions are special variables constructed from existing variables that can help correct for confounding bias by conditioning on them \citep{wooldridge2015control}. Control functions are often constructed from instrumental variables \citep[e.g., ][]{blundell2001endogeneity,imbens2009identification,florens2008identification}, but \cite{nagasawa2018treatment} recently proposes control functions based on proxy variables under assumptions similar to those in the proximal causal inference literature (see the review in \Cref{sec: proximal}). 
We extend this approach to our setting of long term causal inference. This extension is not straightforward, noting that the assumptions of proximal causal inference are not exactly satisfied in our setting (see \Cref{remark: proximal}). 
}

{
Specifically, we will show that we can identify  the long term treatment effect by using 
the stochastic process $\Vcal \coloneqq \braces{p(s_3 \mid S_2, S_1, A, X, G = O): s_3 \in \mathcal{S}_3}$ as a control function. Here we consider identifying the expectation of any arbitrary transformation of the long term potential outcome, a more general parameter than the average effect parameter considered so far. 
\begin{theorem}\label{thm: control-fun}
Suppose \Cref{assump: CI,assump: ext-valid2,assump: unconfound-exp2,assump: unconfound-obs} and the completeness condition in 
\Cref{assump: completeness} condition \ref{assump: completeness-1} hold. Moreover, assume for $a \in \{0, 1\}$, the support of $\Vcal$ given $S_2, A = a, X, G = O$ is identical to the support of $\Vcal$ given $S_2, X, G = O$. 
Then for any function $r: \Ycal \mapsto \R{}$ and $a \in \braces{0, 1}$, 
\begin{align}\label{eq: identification-cf}
\begin{aligned}
    &\Eb{r(Y(a)) \mid G = O} \\
        &\qquad\qquad = \Eb{\Eb{\Eb{r(Y) \mid \Vcal, S_2, A = a,X, G = O} \mid A = a, X, G = E} \mid G = O}.
\end{aligned}
\end{align}
\end{theorem}
Besides the running \Cref{assump: CI,assump: ext-valid2,assump: unconfound-exp2,assump: unconfound-obs}, \Cref{thm: control-fun} also imposes the completeness condition in 
\Cref{assump: completeness} condition \ref{assump: completeness-1} and a common support condition. 
This completeness condition requires $S_3$ to be sufficiently informative for the unobserved confounders $U$, after taking into account other relevant variables. 
The common support condition enables us to vary $A$ while holding constant the control function $\Vcal$ after conditioning on $S_2, X, G = O$.
It is equivalent to the overlap condition that $0 < \Prb{A = 1 \mid \Vcal, S_2, X, G = O} < 1$ almost surely. 
This condition is possible only when $S_1$ can induce sufficient extra variations in $\Vcal$, or alternatively, when $S_1$ has a large support and it is sufficiently informative for $U$ \citep{nagasawa2018treatment}. 
Common support conditions like this are prevalent in the control function literature. See \cite{nagasawa2018treatment,imbens2009identification} for more discussions and justifications.  
}

{
We note that the identification formula in \Cref{eq: identification-cf} can be used to identify not only the average effect, but also the entire distribution of the counterfactual long term outcome $Y(a)$. This can be achieved by applying \Cref{eq: identification-cf} to the indicator function $r(\cdot) = \indic{\cdot \le y}$ for all $y \in \Ycal$. 
Actually, under the condition $2$ in \Cref{corollary: covariate-exp-dr}, we can also use a  selection bridge function to identify the entire distribution of $Y(a)$ (see \Cref{corollary: sel-bridge-dist} in \Cref{sec: covariate-adaptive-identify}). 
The condition $2$ in \Cref{corollary: covariate-exp-dr} (i.e., the existence of a selection bridge function and the completeness condition in 
\Cref{assump: completeness} condition \ref{assump: completeness-2}) has similar qualitative implications as the completeness condition and common support condition in \Cref{thm: control-fun}: they require both $S_1$ and $S_3$ to be sufficiently strong proxies for the unobserved confounders $U$. However, these two set of conditions are in general not directly comparable. See \cite{nagasawa2018treatment} for more discussions on the connections between conditions in the control function approach and conditions in the bridge function approach. 
}

{
Finally, we remark that estimating the target parameter based on the identification formula in \Cref{thm: control-fun} may be  challenging. 
On the one hand, the common support condition may fail in practical applications \citep{chernozhukov2020semiparametric}. 
We may follow \cite{nagasawa2018treatment,newey2021control} and impose additional (semi)-parametric restrictions on the function $\Eb{r(Y) \mid \Vcal, S_2, A = a,X, G = O}$. 
These assumptions can allow for model extrapolation across different values of $\Vcal$, thereby relaxing the common support condition. 
Another possibility is to derive partial identification bounds when the common support function is violated. 
On the other hand, the control function approach requires controlling for an infinitely dimensional stochastic process $\Vcal$, which cannot be  implemented exactly in practice. \cite{nagasawa2018treatment} proposes a dimension reduction technique for the  estimation of causal effects in the proximal causal inference setting. Similar techniques may be also useful in our setting. We leave the development of practical estimation methods based on the control function for the future study. 
}

\section{Numerical Studies}\label{sec: experiment}

\subsection{Real data analysis}\label{sec: real-data}
In this section, we illustrate the performance of our proposed estimators using data for the Greater Avenues to Independence (GAIN) job training program in California. 
GAIN is a job assistance program from the late 1980s designed to help low-income population. 
To evaluate its real impacts on employment, MDRC conducted a randomized experiment in 6 California counties.
{We use the dataset analyzed in \cite{athey2019surrogate} and focus on two counties: San Diego and Riverside.}
For each experiment participant, the dataset records a binary treatment variable indicating enrollment in the GAIN program, quarterly job employment information after treatment assignment,  and other covariate information (e.g., age, education, marriage).
See \citet{hotz2006evaluating,athey2019surrogate} for more information about the GAIN program.

In our numerical studies, we consider the  San Diego data as our experimental dataset $\mathcal{D}_E$, and 
construct an observational dataset $\mathcal{D}_O$ based on the Riverside data via a \emph{biased subsampling} described below. Then we apply our proposed estimators $\hat\tau_{\out}$, $\hat\tau_{\sel}$, and $\hat\tau_{\dr}$ to estimate the average treatment effect of the GAIN program on the long-term employment. 
Since the original data are from randomized experiments, we consider the average treatment effect thereof as the ``ground truth'' and use it to evaluate the errors of different estimators.

\subsubsection{Data Preparation}\label{sec: data}
For the experimental dataset, we directly use  data from San Diego, which include $n^{(1)}_{E} = 6978$ people in the treatment group and $n^{(0)}_{E} = 1154$  people in the control group.
For the observational dataset, we subsample from the Riverside data, which originally include $N_{1} = 4405$ people in the treatment group and $N_{0} = 1040$ people in the control group. 

We  randomly subsample units from the Riverside data according to a sampling probability function $\pi(A, U) \in (0, 1)$, where $A \in \braces{0, 1}$ is the treatment assignment and $U \in \braces{0, 1, 2, 3}$ is the highest education level (``$0$'' means below 9-th grade, ``$1$'' means $9$-th to $11$-th grade, ``$2$'' means $12$-th grade, and  ``$3$'' means above $12$-th grade).
This creates dependence between the treatment assignment and the education level for the units subsampled into  $\mathcal{D}_O$.
We choose education because it is quite likely to have \emph{persistent} effects on participants' potential employment in all quarters following the treatment. 
Then we drop the education level data from $\mathcal{D}_O$ (and also $\mathcal{D}_E$).
As a result, the education level becomes a plausible persistent  unmeasured  confounder in  $\mathcal{D}_O$.

To quantify the strength of unmeasured confounding in $\mathcal{D}_O$, we index the sampling probability function $\pi\prns{A, U}$ by a non-negative parameter $\eta$. 
We set the sampling probability for control units as $\pi(0, U) =  \max \{1 - {\eta U}/{3}, 0.2\}$ and the sampling probability for treated units as $\pi(1, U)$ that satisfies the following equation: 
\[
\frac{N_0}{N_0 + N_1} \pi(0, U) + \frac{N_1}{N_0 + N_1} \pi(1, U) = \frac{N_1}{N_1 + N_0} + \frac{N_0}{N_1 + N_0} \max\{1 - \eta, 0.2\}.
\]
It is easy to show that as $\eta$ grows, the discrepancy between $\pi(0, U)$ and $\pi(1, U)$ also grows. This implies stronger dependence between $U$ and $A$ in the observational dataset $\mathcal{D}_O$, thus stronger unmeasured confounding. 
In \Cref{sec: appendixsimulation} \Cref{prop: experiment}, we prove that with this choice of $\pi(1, U)$, the  subsampling  procedure does not shift the distribution of education level $U$, so that it does not violate \Cref{assump: ext-valid}.
Moreover, the subsampling procedure does not influence \Cref{assump: unconfound-exp,assump: unconfound-obs} since  the sampling probability function only depends on $A, U$. 

In our numerical studies, we consider  the short-term otucomes $(S_1, S_2, S_3)$ as the employment status in  the first two quarters, in the third and fourth quarters, and in the fifth and sixth quarters after the treatment respectively. We consider  the long-term outcome $Y$ as the $20$-th quarter employment.  
These are all binary variables indicating whether the participants are employed in the corresponding quarters after the treatment assignments.

\subsubsection{Results}

\begin{table}[!t]
\centering 
   \begin{tabular}{ccccccccccccccccc}
    \toprule
    \multirow{2}{*}{$\eta$} & & \multicolumn{4}{c}{$\hat\tau_{\out}$} & \multicolumn{4}{c}{$\hat\tau_{\sel}$} & \multicolumn{4}{c}{$\hat\tau_{\dr}$} &  \multicolumn{2}{c}{\scriptsize Athey et al.} & \multirow{2}{*}{\scriptsize Naive}  \\
     & & 0 & .33 & .67 & 1 &  0 & .33 & .67 & 1 & 0 & .33 & .67 & 1 & NR & CV &   \\
         \cmidrule(lr){1-2} \cmidrule(lr){3-6} \cmidrule(lr){7-10} \cmidrule(lr){11-14} \cmidrule(lr){15-16} \cmidrule(lr){17-17}
    \multirow{2}{*}{0} & MAE & 67 & 89 & 84 & 82 & 81 & 81 & 80 & 80 & 71 & 95 & 90 & 88 & 11 & 17 & 0.053 \\\vspace{0.25cm}
& Med & 67 & 89 & 84 & 82 & 81 & 81 & 80 & 80 & 71 & 95 & 90 & 88 & 11 & 17 & 0.053 \\
\multirow{2}{*}{0.2} & MAE & 18 & 84 & 80 & 78 & 79 & 79 & 79 & 79 & 24 & 89 & 86 & 85 & 19 & 15 & 0.059 \\\vspace{0.25cm}
& Med & 61 & 84 & 80 & 78 & 79 & 79 & 79 & 79 & 65 & 90 & 87 & 85 & 18 & 15 & 0.059\\
\multirow{2}{*}{0.4} & MAE & 17 & 79 & 75 & 74 & 76 & 76 & 76 & 76 & 23 & 84 & 82 & 80 & 25 & 13 & 0.067 \\\vspace{0.25cm}
& Med & 62 & 79 & 76 & 74 & 77 & 76 & 76 & 76 & 65 & 85 & 83 & 82 & 25 & 13 & 0.067\\
\multirow{2}{*}{0.6} & MAE & 10 & 73 & 70 & 69 & 72 & 72 & 72 & 72 & 17 & 79 & 77 & 76 & 31 & 11 & 0.076 \\\vspace{0.25cm}
& Med & 60 & 74 & 71 & 69 & 73 & 73 & 72 & 72 & 63 & 80 & 78 & 77 & 31 & 10 & 0.076 \\
\multirow{2}{*}{0.8} & MAE & -25 & 66 & 64 & 62 & 67 & 67 & 67 & 67 & -11 & 72 & 71 & 70 & 33 & 8 & 0.088 \\\vspace{0.25cm}
& Med & 57 & 66 & 64 & 62 & 68 & 67 & 67 & 67 & 59 & 73 & 72 & 71 & 32 & 8 & 0.088 \\
\multirow{2}{*}{1} & MAE & 24 & 65 & 63 & 62 & 68 & 68 & 67 & 67 & 32 & 72 & 71 & 70 & 35 & 6 & 0.095\\\vspace{0.25cm}
& Med & 57 & 65 & 63 & 62 & 69 & 68 & 68 & 68 & 60 & 73 & 72 & 71 & 36 & 6 & 0.095  \\
\multirow{2}{*}{1.2} & MAE & -267 & 64 & 62 & 61 & 68 & 68 & 68 & 67 & -323 & 72 & 70 & 70 & 37 & 5 & 0.104 \\\vspace{0.25cm}
& Med & 56 & 65 & 62 & 61 & 70 & 69 & 69 & 68 & 59 & 74 & 72 & 71 & 38 & 5 & 0.104  \\
\multirow{2}{*}{1.4} & MAE & -13 & 62 & 59 & 58 & 69 & 68 & 67 & 67 & -12 & 71 & 70 & 69 & 38 & 4 & 0.115 \\\vspace{0.25cm} 
& Med & 51 & 63 & 60 & 58 & 72 & 71 & 71 & 70 & 53 & 75 & 74 & 73 & 38 & 4 & 0.115 \\
\multirow{2}{*}{1.6} & MAE & 5 & 61 & 58 & 56 & 68 & 68 & 67 & 66 & 10 & 71 & 70 & 68 & 40 & 4 & 0.124 \\
& Med & 49 & 61 & 58 & 56 & 71 & 71 & 70 & 68 & 52 & 74 & 73 & 72 & 40 & 3 & 0.124\\
\bottomrule
\end{tabular} 
\caption{Percentage improvement in error over the naive unadjusted difference-in-mean estimator for different estimators: our proposed estimators $\hat\tau_{\out}$, $\hat\tau_{\sel}$ and  $\hat\tau_{\dr}$, and the estimator proposed in \cite{athey2020combining}.
Larger percentage decrease means better performance.
For reference, the last column shows the error of the naive unadjusted estimator. 
For our estimators, we fit bridge functions either using no regularization (the ``0'' column) or ridge regularization with $\lambda = 0.33 / n^{(a)}_O, 0.67 / n^{(a)}_O$ and $1/n^{(a)}_O$ (the ``.33'', ``.67'' and ``1'' columns) respectively. For \citet{athey2020combining}, we considered using no regularization (the ``NR'' column) and using ridge regularization where the regularization parameter is selected by cross validation (the ``CV'' column). 
}
\label{table: numerical}
\end{table}

\Cref{table: numerical} reports the performance of different estimators over 1000 replications\footnote{{When $\eta = 0$, the sampling probabilities satisfy $\pi(0, U) = \pi(1, U) = 1$, so there is no subsampling and all replications are identical.}} of the data subsampling.
Each replication results in  a different observational dataset $\mathcal{D}_{O}$ with different number of treated units $n_O^{(1)} < N_1$ and different number of control units $n_O^{(0)} < N_0$.
For evaluation we consider two criterions over the $1000$ replications:  Mean Absolute Error (MAE) and Median of Abolute Errors (MedAE).

In \Cref{table: numerical}, we compare the performance of our proposed estimators $\hat\tau_{\out}$, $\hat\tau_{\sel}$ and $\hat\tau_{\dr}$  in \Cref{sec: estimation} with two benchmarks: the naive difference-in-mean estimator that uses only the observational dataset and the imputation estimator proposed in Section 4.1 of \cite{athey2020combining}, which uses both datasets and information of all short-term outcomes $S=\prns{S_1, S_2, S_3}$. 
The naive estimator completely ignores confounding, and the estimator in \cite{athey2020combining} can only account for  short-term confounding but not persistent confounding. To evaluate the performance of our estimators and \citet{athey2020combining}, we consider the percentage decrease in either of our error criteria relative to the naive difference-in-mean  estimator. 
A positive value corresponds to improvement over the naive estimator, and a larger value indicates better performance. A negative value means worse error than the naive estimator. 

In our estimators and the imputation estimator in \citet{athey2020combining}, we need to first estimate some nuisance functions. 
We specify the outcome bridge function in our estimators and the imputation function in \cite{athey2020combining} to be linear functions, and specify the selection bridge function in our estimators to be of the form $q(s_2, s_1, a, x) = \exp(\beta^\top_{2, a} s_2 + \beta_{1, a}^\top s_1 + \beta_{0, a}^\top x + \gamma_a)$.
Since these are all simple parametric functions, we do not need the cross-fitting technique described in \Cref{sec: estimation}, but instead use the same data for nuisance estimation and the final plug-in estimation.
To estimate the bridge functions, we employ the generalized method of moment (GMM) approach in \cite{cui2020semiparametric}. We consider a standard GMM apporach and the approach with additional ridge regularization, i.e., regularizing the $L_2$ norms of bridge function coefficients in the GMM objectives, as suggested by \cite{imbens2021controlling}. When we estimate the bridge function corresponding to the treatment level $a\in\braces{0, 1}$, we set the regularization tuning parameter as $\lambda = \lambda_0 (n^{(a)}_O)^{-1}$ for $\lambda_0 \in\{ 0, 0.33, 0.67, 1\}$ (here $\lambda_0 = 0$ corresponds to no regularization).
For the imputation function of \cite{athey2020combining}, we implement it using either ordinary least squares or cross-validated ridge regression, for which we use
the default options in the \emph{R}  package \texttt{glmnet}~\citep{glmnet}.

From \Cref{table: numerical}, we observe that with $\lambda_0 \neq 0$, the performance of our proposed estimators $\hat\tau_{\out}, \hat\tau_{\sel}, \hat\tau_{\dr}$ is stable. They consistently outperform the benchmarks, in terms of both criteria. In particular, the doubly robust estimator $\hat\tau_{\dr}$ performs the best, reducing the estimation errors of benchmark methods by large margins. 
Notably, although the benchmark estimator proposed by \cite{athey2020combining} improves upon the naive estimator, it is always outperformed by our proposed estimators.
This may be due to the fact that  the estimator in \cite{athey2020combining} cannot handle persistent confounding.
We also observe that as the unmeasured confounding becomes stronger  (i.e., as $\eta$ grows), all estimators have higher estimation errors, especially the naive estimator. 

We observe that the MAE of our estimators when not using regularization is sometimes worse than the estimator of \citet{athey2020combining} and even the naive estimator. 
This is because unregularized estimators can be unstable and MAE is sensitive to outlier estimates. 
Indeed, estimating bridge functions requires solving  inverse problems defined by conditional moment equations, which can be intrinsically difficult. 
This problem is common in proximal causal inference, and regularization has been shown to be sometimes key for valid inference \citep{imbens2021controlling}.
Nevertheless, the MedAE, which is robust to outliers, for our estimators is still lower than the benchmarks. This shows that our proposed estimators, regularized or not, all  effectively address the confounding bias. {In the supplementary material \Cref{sec: appendixsimulation}, we heuristically probe the plausibility of \Cref{assump: bridge,assump: bridge2} according to the characterization of bridge functions in a discrete setting (see \cref{ex: discrete}). Moreover, we also 
report the performance of different estimators by varying the number of quarters used for surrogate construction.
Our result shows that our approach is consistently more accurate than the approach in \citet{athey2020combining}.}

\subsection{A simulation study}\label{sec: simulation}

{In \Cref{sec: real-data}, we focus on parametric estimation of bridge functions. In this part, we use a simulation study to further demonstrate the performance of our approach when bridge functions are nonlinear and estimated by more flexible neural networks.}

{
Specifically, for both the experimental and observational data, we first generate random vectors $\tilde X$ and $U$ from the multivariate normal distribution with mean zero and covariance matrix $0.5 \mathbf{I}$, where $\mathbf{I}$ is an identity matrix with suitable size. We fix the dimension of $U$ as $5$ and vary the dimension of $\tilde{X}$ over $\{5, 10, 15, 20\}$. We further generate $Y(a) \in \R{}, \tilde S_1(a) \in \R{5}, \tilde S_2(a) \in \R{}, \tilde S_3(a) \in \R{5}$ according to the following process:
\begin{align*}
&Y(a) = \tau_y a + \alpha_y^\top \tilde S_3(a) + \beta_y^\top \tilde X + \gamma_y^\top U + \epsilon_y, \\ 
&\tilde S_j(a) = \tau_j a + \alpha_j \tilde S_{j-1}(a) + \beta_j \tilde X + \gamma_j U + \epsilon_j, ~ j \in \braces{3, 2}\\
&\tilde S_1(a) = \tau_1 a +  \beta_1 \tilde X + \gamma_1 U + \epsilon_1,
\end{align*}
where $\tau_y, (\tau_j, \alpha_y, \beta_y, \gamma_y), (\alpha_j, \beta_j, \gamma_j)$ are scalers, vectors, and matrices of conformable sizes, and $\epsilon_y, \epsilon_j$ are independent mean-zero Gaussian terms with variance $0.5$. We generate the entries in $\tau_y, (\tau_j, \alpha_y, \beta_y, \gamma_y), (\alpha_j, \beta_j, \gamma_j)$ by first drawing numbers from the uniform distribution over the $[0, 1]$ interval and then rescaling them so that the $\ell_2$-norms of the vectors $(\tau_j, \alpha_y, \beta_y, \gamma_y)$ and the columns of $(\alpha_j, \beta_j, \gamma_j)$ are all equal to $0.5$. 
Moreover, we draw the treatment indicator $A$ according to $\mathbb{P}({A = 1 \mid \tilde{X}, U, G = E}) = \frac{1}{2}$ and $\mathbb{P}({A = 1 \mid \tilde{X}, U, G = O}) = ({1 + \exp(\kappa_1^\top \tilde{X} + \kappa_2^\top U)})^{-1}$, where the coefficients $\kappa_1$ and $\kappa_2$ are similarly generated by sampling and rescaling. 
According to \Cref{ex: linear} and \Cref{sec: linear}, the outcome and selection bridge functions exist under certain rank conditions. Moreover, the outcome bridge function is linear in $\tilde S_3, \tilde S_2, \tilde X, A$ and the selection bridge function is an exponential transformation of a linear function of $\tilde S_2, \tilde S_1, \tilde X, A$. 
To introduce nonlinear bridge functions, we apply a nonlinear transformation $g(\cdot) = \op{sign}(\cdot)\abs{\cdot}^q$ for $q \in \{1, 1.5, 2\}$ to each element of $\tilde X, \tilde S_1, \tilde S_2, \tilde S_3$, leading to $X,  S_1,  S_2,  S_3$ respectively. This is an invertible transformation that ensures a one-to-one correspondence between the original variables and transformed variables. With these transformations, the bridge functions with respect to the transformed variables  $X,  S_1,  S_2,  S_3$ are linear when $q = 1$ and nonlinear when $q = 1.5$ or $2$.
}

{We repeat generating data according to the process above for 200 times. In each replicate, we draw new values for all the model parameters and generate observational data and experimental data accordingly with equal sizes $n_O = n_E = 2000$. 
We apply our proposed doubly robust estimator and associated confidence intervals to the datasets. We estimate the bridge functions in two ways. One way is to use the minimax learning estimators in \cite{kallus2021causal,DikkalaNishanth2020MEoC}. Specifically, a minimax bridge function estimator is obtained as the solution to a minimax optimization problem derived from the corresponding conditional moment equation. In our study, we follow \cite{kallus2021causal} and specify the outer minimization function class (i.e., the class used to model the bridge function) as a neural network class and the inner maximization function class (i.e., the class used to guarantee the equivalence between the minimax optimization formulation and the  conditional moment equation formulation) as a Reproducing Kernel Hilbert Space (RKHS). For implementation details, we refer the readers to \Cref{sec: appendixsimulation}.
The other way is to use parametric estimators for the bridge functions, where the specifications are identical to those in \Cref{sec: real-data}. These model  specifications are correct when $q = 1$ but wrong when $q \in \{1.5, 2\}$. In both approaches, we also use a ridge regularization with $\lambda = 1 / n_O^{(a)}$, i.e., same as the ``1'' column in \Cref{table: numerical}. 
}

{
Table~\ref{table: simulation} reports the performance of our estimator and confidence intervals based on two kinds of bridge estimators, over the $200$ replicates. 
When the covariate dimension is relatively low (i.e., $\dim(X) = 5, 10$ or $15$), the empirical coverage of the minimax based approach is close to the $95\%$ nominal level in most of the specifications. For the higher dimensional regime $\dim(X) = 20$, the empirical coverage is slightly worse, which could be due to the curse of dimensionality, especially for the inner maximization over RKHS. }

{
When $q = 1$, the average bias of the parametric model based approach is consistently smaller than the minimax based approach. This is expected as the parametric model is  correctly specified in this case. In contrast, in the nonlinear settings with  $q = 1.5$ and $2$, the parametric models are misspecified, so the RMSE and the average bias of the parametric based approach are overall worse than the minimax based approach, especially for the more nonlinear scenario $q = 2$. This shows the benefit of using flexible function classes to model complex bridge functions. }

{
Interestingly, when $q = 1.5, 2$, the confidence intervals based on the parametric bridge function estimators do not under-cover the truth, even though the corresponding point estimators have larger bias and RMSE. Actually, they tend to over-cover the truth in many specifications. 
This is perhaps due to the fact that the asymptotic variance tends to be over-estimated, leading to excessively conservative confidence interval lengths.  As shown in \Cref{table: simulation}, the average confidence interval length produced by the parametric based approach is much larger than the minimax based approach; sometimes, it can even be 4 times larger. Consequently, even with higher average bias, the parametric based approach is still able to get high empirical coverage.} {In \Cref{sec: appendixsimulation}, we include additional results when we increase the dimension of $S_1, S_2, U$ from $5$ to $10$. The estimators and confidence interval coverage perform slightly worse in the higher dimensional setting.}

{
\begin{table}[t!]
    \centering
    \begin{tabular}{cccccccccc}
    \toprule
    \multirow{2}{*}{$q$} & & \multicolumn{2}{c}{$\textrm{dim}(X) = 5$}
     & \multicolumn{2}{c}{$\textrm{dim}(X) = 10$} & \multicolumn{2}{c}{$\textrm{dim}(X) = 15$} & \multicolumn{2}{c}{$\textrm{dim}(X) = 20$}\\ 
        & & MinMax & Param. &  MinMax & Param.  & MinMax & Param.  & MinMax & Param. \\ \cmidrule(lr){1-2}\cmidrule(lr){3-4}\cmidrule(lr){5-6}\cmidrule(lr){7-8}\cmidrule(lr){9-10} 
       \multirow{4}{*}{1} & CP & 90.0\% & 94.5\% &  96.0\% & 95.5\% & 94.0\% & 95.5\% & 90.5\% & 95.0\% \\
       & CI Len. & 0.541 & 0.576 & 0.541 & 0.578  & 0.548 & 0.579  & 0.546 & 0.579 \\
       & RMSE & 0.157 & 0.151 &  0.135 & 0.139 & 0.146 & 0.147  & 0.155 & 0.154 \\\vspace{0.25cm}
       & Bias & 0.058 & 0.017 & 0.033 & 0.015  & 0.036 & 0.002  & 0.054 & 0.001 \\ 
       \multirow{4}{*}{1.5} & CP & 96.5\% & 97.5\% & 95.5\% & 97.0\% & 95.0\% & 97.0\% & 92.5\% & 97.0\% \\
       & CI Len. &  0.619 & 0.823 & 0.576 & 0.828 & 0.578 & 0.819 & 0.574 & 0.823 \\
       & RMSE & 0.159 & 0.184 & 0.143 & 0.199 & 0.160 & 0.202  & 0.157 & 0.197 \\\vspace{0.25cm}
       & Bias & 0.033 & 0.031 & 0.036 & 0.075 & 0.034 & 0.052  & 0.020 & 0.056 \\ 
       \multirow{4}{*}{2} & CP & 95.5\% & 96.0\% & 94.0\% & 97.0\% & 93.0\% & 97.5\% & 93.0\% & 98.0\% \\
       & CI Len. & 0.698 & 2.229 & 0.611 &  2.023 & 0.612 & 1.842 & 0.595 & 2.028 \\
       & RMSE & 0.187 & 0.698 & 0.157 & 0.581 & 0.192 & 0.521  & 0.165 & 0.574 \\
       & Bias & 0.073 & 0.028 & 0.049 & 0.085  & 0.073 & 0.096 & 0.032 &  0.113 \\
       \bottomrule
    \end{tabular}
    \caption{The empirical coverage probability (CP) of the $95\%$-confidence interval and its average length (CI len.), and the root mean squared error (RMSE) and the average absolute bias (Bias) of the doubly robust estimators, with bridge functions estimated by the minimax approach (MinMax) and parametric approach (Param.) respectively.  
    The covariate dimension varies from $5$ to $20$ and the degree of nonlinearity varies from $q = 1$ to $2$. 
    Here $q = 1$ corresponds to linear (or exponential linear) bridge functions.}
    \label{table: simulation}
\end{table}}

\section{Conclusions}\label{sec: conclusion}
In this paper, we consider combining experimental and observational data for long-term causal inference. 
We are particularly interested in the challenge of persistent confounding, i.e., the presence of unobserved confounders that affect both the short-term and long-term outcomes. 
To overcome this challenge, we leverage the sequential structure of multiple  short-term outcomes and use part of them as proxy variables for the unobserved confounders. 
We propose several novel identification strategies for the average long-term treatment effect. Based on them, we design flexible treatment effect estimators and inference methods, for which we provide asymptotic guarantees. 
Our results show that the long-term treatment effect can be identified and estimated under much more general conditions  than before. 

Beyond these specific results, our work reveals an interesting role for the structure of short-term outcomes in long-term causal inference. 
To the best of our knowledge, the structure of repeated outcome measurements is largely unexplored in the long-term causal inference literature. 
We hope that our work will inspire other researchers to study other plausible structures for 
short-term outcomes and benefits these can have for
long-term causal inference.

\section*{Acknowledgments}
The authors thank
the associate editor and two anonymous reviewers for their insights and suggestions, which have led to significant improvement of this paper. \\

% \noindent\emph{Conflict of interest}: We have no conflict of interest to disclose.

\section*{Funding}

Guido Imbens thanks the Office of Naval Research for support under grant numbers N00014-17-1-2131 and N00014-19-1-2468 and Amazon for a gift.
Nathan Kallus acknowledges that this material is based upon work supported by the National Science Foundation under Grant No. 1846210. Xiaojie Mao is supported in part by National Natural Science Foundation of China (grant numbers 72201150, 72322001, and 72293561) and National Key
R\&D Program of China (grant number 2022ZD0116700). Yuhao Wang is supported in part by National Key R \& D Program of China (2022YFA1008100), the 2030 Innovation Megaprojects of China (Programme on New Generation Artificial Intelligence) Grant No. 2021AAA0150000, and the grant of National Natural Science Foundation of China (NSFC) 12201341.

\section*{Data availability}
The California GAIN dataset analyzed in \Cref{sec: real-data} contains sensitive individual data and cannot be shared publicly. It may be shared upon request. The data analyzed in \Cref{sec: simulation} are simulated according to the processes described in that section. The code script used to generate the simulated data is available at \url{https://github.com/CausalML/LongTermCausalInference}.

\bibliographystyle{plainnat}
\bibliography{semiparametric}

\begin{thebibliography}{92}
\providecommand{\natexlab}[1]{#1}
\providecommand{\url}[1]{\texttt{#1}}
\expandafter\ifx\csname urlstyle\endcsname\relax
  \providecommand{\doi}[1]{doi: #1}\else
  \providecommand{\doi}{doi: \begingroup \urlstyle{rm}\Url}\fi

\bibitem[Ai and Chen(2003)]{ai2003efficient}
Chunrong Ai and Xiaohong Chen.
\newblock Efficient estimation of models with conditional moment restrictions
  containing unknown functions.
\newblock \emph{Econometrica}, 71\penalty0 (6):\penalty0 1795--1843, 2003.

\bibitem[Angrist and Krueger(1992)]{angrist1992effect}
Joshua~D Angrist and Alan~B Krueger.
\newblock The effect of age at school entry on educational attainment: an
  application of instrumental variables with moments from two samples.
\newblock \emph{Journal of the American statistical Association}, 87\penalty0
  (418):\penalty0 328--336, 1992.

\bibitem[Athey et~al.(2019)Athey, Chetty, Imbens, and Kang]{athey2019surrogate}
Susan Athey, Raj Chetty, Guido~W Imbens, and Hyunseung Kang.
\newblock The surrogate index: Combining short-term proxies to estimate
  long-term treatment effects more rapidly and precisely.
\newblock \emph{NBER Working Paper}, \penalty0 (w26463), 2019.

\bibitem[Athey et~al.(2020)Athey, Chetty, and Imbens]{athey2020combining}
Susan Athey, Raj Chetty, and Guido Imbens.
\newblock Combining experimental and observational data to estimate treatment
  effects on long term outcomes, 2020.

\bibitem[Battocchi et~al.(2021)Battocchi, Dillon, Hei, Lewis, Oprescu, and
  Syrgkanis]{battocchi2021estimating}
Keith Battocchi, Eleanor Dillon, Maggie Hei, Greg Lewis, Miruna Oprescu, and
  Vasilis Syrgkanis.
\newblock Estimating the long-term effects of novel treatments.
\newblock \emph{Advances in Neural Information Processing Systems}, 34, 2021.

\bibitem[Bennett and Kallus(2020)]{bennett2020variational}
Andrew Bennett and Nathan Kallus.
\newblock The variational method of moments.
\newblock \emph{arXiv preprint arXiv:2012.09422}, 2020.

\bibitem[Bennett and Kallus(2021)]{bennett2021proximal}
Andrew Bennett and Nathan Kallus.
\newblock Proximal reinforcement learning: Efficient off-policy evaluation in
  partially observed markov decision processes.
\newblock \emph{arXiv preprint arXiv:2110.15332}, 2021.

\bibitem[Bennett et~al.(2019)Bennett, Kallus, and Schnabel]{NIPS2019_8615}
Andrew Bennett, Nathan Kallus, and Tobias Schnabel.
\newblock Deep generalized method of moments for instrumental variable
  analysis.
\newblock In \emph{Advances in Neural Information Processing Systems 32}, pages
  3564--3574. 2019.

\bibitem[Bennett et~al.(2022)Bennett, Kallus, Mao, Newey, Syrgkanis, and
  Uehara]{bennett2022inference}
Andrew Bennett, Nathan Kallus, Xiaojie Mao, Whitney Newey, Vasilis Syrgkanis,
  and Masatoshi Uehara.
\newblock Inference on strongly identified functionals of weakly identified
  functions.
\newblock \emph{arXiv e-prints}, pages arXiv--2208, 2022.

\bibitem[Bennett et~al.(2023{\natexlab{a}})Bennett, Kallus, Mao, Newey,
  Syrgkanis, and Uehara]{bennett2023minimax}
Andrew Bennett, Nathan Kallus, Xiaojie Mao, Whitney Newey, Vasilis Syrgkanis,
  and Masatoshi Uehara.
\newblock Minimax instrumental variable regression and $ l\_2 $ convergence
  guarantees without identification or closedness.
\newblock In \emph{The Thirty Sixth Annual Conference on Learning Theory},
  pages 2291--2318. PMLR, 2023{\natexlab{a}}.

\bibitem[Bennett et~al.(2023{\natexlab{b}})Bennett, Kallus, Mao, Newey,
  Syrgkanis, and Uehara]{bennett2023source}
Andrew Bennett, Nathan Kallus, Xiaojie Mao, Whitney Newey, Vasilis Syrgkanis,
  and Masatoshi Uehara.
\newblock Source condition double robust inference on functionals of inverse
  problems.
\newblock \emph{arXiv preprint arXiv:2307.13793}, 2023{\natexlab{b}}.

\bibitem[Blundell and Powell(2003)]{blundell2001endogeneity}
Richard Blundell and James~L Powell.
\newblock Endogeneity in nonparametric and semiparametric regression models.
\newblock 2003.

\bibitem[Cai et~al.(2021{\natexlab{a}})Cai, Lu, and Song]{cai2021coda}
Hengrui Cai, Wenbin Lu, and Rui Song.
\newblock Coda: Calibrated optimal decision making with multiple data sources
  and limited outcome.
\newblock \emph{arXiv preprint arXiv:2104.10554}, 2021{\natexlab{a}}.

\bibitem[Cai et~al.(2021{\natexlab{b}})Cai, Song, and Lu]{cai2021gear}
Hengrui Cai, Rui Song, and Wenbin Lu.
\newblock Gear: On optimal decision making with auxiliary data.
\newblock \emph{Stat}, 10\penalty0 (1):\penalty0 e399, 2021{\natexlab{b}}.

\bibitem[Carrasco et~al.(2007)Carrasco, Florens, and Renault]{Carrasco2007}
Marine Carrasco, jean-pierre Florens, and Eric Renault.
\newblock Chapter 77 linear inverse problems in structural econometrics
  estimation based on spectral decomposition and regularization.
\newblock \emph{Handbook of Econometrics}, 6:\penalty0 5633--5751, 12 2007.

\bibitem[Chen et~al.(2007)Chen, Geng, and Jia]{chen2007criteria}
Hua Chen, Zhi Geng, and Jinzhu Jia.
\newblock Criteria for surrogate end points.
\newblock \emph{Journal of the Royal Statistical Society: Series B (Statistical
  Methodology)}, 69\penalty0 (5):\penalty0 919--932, 2007.

\bibitem[Chen and Ritzwoller(2021)]{chen2021semiparametric}
Jiafeng Chen and David~M Ritzwoller.
\newblock Semiparametric estimation of long-term treatment effects.
\newblock \emph{arXiv preprint arXiv:2107.14405}, 2021.

\bibitem[Chen et~al.(2021)Chen, Zhang, and Ye]{chen2021minimax}
Shuxiao Chen, Bo~Zhang, and Ting Ye.
\newblock Minimax rates and adaptivity in combining experimental and
  observational data.
\newblock \emph{arXiv preprint arXiv:2109.10522}, 2021.

\bibitem[Chen and Pouzo(2012)]{chen2012estimation}
Xiaohong Chen and Demian Pouzo.
\newblock Estimation of nonparametric conditional moment models with possibly
  nonsmooth generalized residuals.
\newblock \emph{Econometrica}, 80\penalty0 (1):\penalty0 277--321, 2012.

\bibitem[Cheng and Cai(2021)]{cheng2021adaptive}
David Cheng and Tianxi Cai.
\newblock Adaptive combination of randomized and observational data.
\newblock \emph{arXiv preprint arXiv:2111.15012}, 2021.

\bibitem[Chernozhukov et~al.(2019)Chernozhukov, Newey, Robins, and
  Singh]{chernozhukov2019double}
V~Chernozhukov, W~Newey, J~Robins, and R~Singh.
\newblock Double/de-biased machine learning of global and local parameters
  using regularized riesz representers.
\newblock \emph{stat}, 1050:\penalty0 9, 2019.

\bibitem[Chernozhukov et~al.(2018)Chernozhukov, Chetverikov, Demirer, Duflo,
  Hansen, Newey, and Robins]{chernozhukov2018double}
Victor Chernozhukov, Denis Chetverikov, Mert Demirer, Esther Duflo, Christian
  Hansen, Whitney Newey, and James Robins.
\newblock Double/debiased machine learning for treatment and structural
  parameters, 2018.

\bibitem[Chernozhukov et~al.(2020)Chernozhukov, Fern{\'a}ndez-Val, Newey,
  Stouli, and Vella]{chernozhukov2020semiparametric}
Victor Chernozhukov, Iv{\'a}n Fern{\'a}ndez-Val, Whitney Newey, Sami Stouli,
  and Francis Vella.
\newblock Semiparametric estimation of structural functions in nonseparable
  triangular models.
\newblock \emph{Quantitative Economics}, 11\penalty0 (2):\penalty0 503--533,
  2020.

\bibitem[Chetty et~al.(2011)Chetty, Friedman, Hilger, Saez, Schanzenbach, and
  Yagan]{chetty2011does}
Raj Chetty, John~N Friedman, Nathaniel Hilger, Emmanuel Saez, Diane~Whitmore
  Schanzenbach, and Danny Yagan.
\newblock How does your kindergarten classroom affect your earnings? evidence
  from project star.
\newblock \emph{The Quarterly journal of economics}, 126\penalty0 (4):\penalty0
  1593--1660, 2011.

\bibitem[Colnet et~al.(2020)Colnet, Mayer, Chen, Dieng, Li, Varoquaux, Vert,
  Josse, and Yang]{colnet2020causal}
B{\'e}n{\'e}dicte Colnet, Imke Mayer, Guanhua Chen, Awa Dieng, Ruohong Li,
  Ga{\"e}l Varoquaux, Jean-Philippe Vert, Julie Josse, and Shu Yang.
\newblock Causal inference methods for combining randomized trials and
  observational studies: a review.
\newblock \emph{arXiv preprint arXiv:2011.08047}, 2020.

\bibitem[Cui et~al.(2020)Cui, Pu, Shi, Miao, and
  Tchetgen]{cui2020semiparametric}
Yifan Cui, Hongming Pu, Xu~Shi, Wang Miao, and Eric~Tchetgen Tchetgen.
\newblock Semiparametric proximal causal inference.
\newblock \emph{arXiv preprint arXiv:2011.08411}, 2020.

\bibitem[{Darolles} et~al.(2010){Darolles}, {Fan}, {Florens}, and
  {Renault}]{darolles2010nonparametric}
Serge {Darolles}, Yanqin {Fan}, Jean-Pierre {Florens}, and Eric {Renault}.
\newblock Nonparametric instrumental regression.
\newblock \emph{Econometrica}, 79\penalty0 (5):\penalty0 1541--1565, 2010.

\bibitem[Deaner(2021)]{deaner2021proxy}
Ben Deaner.
\newblock Proxy controls and panel data.
\newblock \emph{arXiv preprint arXiv:1810.00283}, 2021.

\bibitem[Dikkala et~al.(2020)Dikkala, Lewis, Mackey, and
  Syrgkanis]{DikkalaNishanth2020MEoC}
Nishanth Dikkala, Greg Lewis, Lester Mackey, and Vasilis Syrgkanis.
\newblock Minimax estimation of conditional moment models.
\newblock In \emph{Advances in Neural Information Processing Systems},
  volume~33, pages 12248--12262, 2020.

\bibitem[Dukes et~al.(2021)Dukes, Shpitser, and Tchetgen]{dukes2021proximal}
Oliver Dukes, Ilya Shpitser, and Eric J~Tchetgen Tchetgen.
\newblock Proximal mediation analysis.
\newblock \emph{arXiv preprint arXiv:2109.11904}, 2021.

\bibitem[Florens et~al.(2008)Florens, Heckman, Meghir, and
  Vytlacil]{florens2008identification}
Jean-Pierre Florens, James~J Heckman, Costas Meghir, and Edward Vytlacil.
\newblock Identification of treatment effects using control functions in models
  with continuous, endogenous treatment and heterogeneous effects.
\newblock \emph{Econometrica}, 76\penalty0 (5):\penalty0 1191--1206, 2008.

\bibitem[Frangakis and Rubin(2002)]{frangakis2002principal}
Constantine~E Frangakis and Donald~B Rubin.
\newblock Principal stratification in causal inference.
\newblock \emph{Biometrics}, 58\penalty0 (1):\penalty0 21--29, 2002.

\bibitem[Ghassami et~al.(2021)Ghassami, Shpitser, and
  Tchetgen]{ghassami2021proximal}
AmirEmad Ghassami, Ilya Shpitser, and Eric~Tchetgen Tchetgen.
\newblock Proximal causal inference with hidden mediators: Front-door and
  related mediation problems.
\newblock \emph{arXiv preprint arXiv:2111.02927}, 2021.

\bibitem[Ghassami et~al.(2022{\natexlab{a}})Ghassami, Shpitser, and
  Tchetgen]{ghassami2022combining}
AmirEmad Ghassami, Ilya Shpitser, and Eric~Tchetgen Tchetgen.
\newblock Combining experimental and observational data for identification of
  long-term causal effects.
\newblock \emph{arXiv preprint arXiv:2201.10743}, 2022{\natexlab{a}}.

\bibitem[Ghassami et~al.(2022{\natexlab{b}})Ghassami, Ying, Shpitser, and
  Tchetgen]{GhassamiAmirEmad2021MKML}
AmirEmad Ghassami, Andrew Ying, Ilya Shpitser, and Eric~Tchetgen Tchetgen.
\newblock Minimax kernel machine learning for a class of doubly robust
  functionals with application to proximal causal inference.
\newblock In \emph{International Conference on Artificial Intelligence and
  Statistics}, pages 7210--7239. PMLR, 2022{\natexlab{b}}.

\bibitem[Graham et~al.(2016)Graham, Pinto, and Egel]{graham2016efficient}
Bryan~S Graham, Cristine Campos de~Xavier Pinto, and Daniel Egel.
\newblock Efficient estimation of data combination models by the method of
  auxiliary-to-study tilting (ast).
\newblock \emph{Journal of Business \& Economic Statistics}, 34\penalty0
  (2):\penalty0 288--301, 2016.

\bibitem[Gupta et~al.(2019)Gupta, Kohavi, Tang, Xu, Andersen, Bakshy, Cardin,
  Chandran, Chen, Coey, et~al.]{gupta2019top}
Somit Gupta, Ronny Kohavi, Diane Tang, Ya~Xu, Reid Andersen, Eytan Bakshy,
  Niall Cardin, Sumita Chandran, Nanyu Chen, Dominic Coey, et~al.
\newblock Top challenges from the first practical online controlled experiments
  summit.
\newblock \emph{ACM SIGKDD Explorations Newsletter}, 21\penalty0 (1):\penalty0
  20--35, 2019.

\bibitem[Hall and Horowitz(2005)]{Hall05IV}
Peter Hall and Joel~L. Horowitz.
\newblock {Nonparametric methods for inference in the presence of instrumental
  variables}.
\newblock \emph{The Annals of Statistics}, 33\penalty0 (6):\penalty0 2904 --
  2929, 2005.
\newblock \doi{10.1214/009053605000000714}.
\newblock URL \url{https://doi.org/10.1214/009053605000000714}.

\bibitem[Hansen(1982)]{hansen1982large}
Lars~Peter Hansen.
\newblock Large sample properties of generalized method of moments estimators.
\newblock \emph{Econometrica: Journal of the Econometric Society}, pages
  1029--1054, 1982.

\bibitem[Hartford et~al.(2017)Hartford, Lewis, Leyton-Brown, and
  Taddy]{pmlr-v70-hartford17a}
Jason Hartford, Greg Lewis, Kevin Leyton-Brown, and Matt Taddy.
\newblock Deep {IV}: A flexible approach for counterfactual prediction.
\newblock In \emph{Proceedings of the 34th International Conference on Machine
  Learning}, volume~70, pages 1414--1423, 2017.

\bibitem[Hohnhold et~al.(2015)Hohnhold, O'Brien, and
  Tang]{hohnhold2015focusing}
Henning Hohnhold, Deirdre O'Brien, and Diane Tang.
\newblock Focusing on the long-term: It's good for users and business.
\newblock In \emph{Proceedings of the 21th ACM SIGKDD International Conference
  on Knowledge Discovery and Data Mining}, pages 1849--1858, 2015.

\bibitem[Hotz et~al.(2006)Hotz, Imbens, and Klerman]{hotz2006evaluating}
V~Joseph Hotz, Guido~W Imbens, and Jacob~A Klerman.
\newblock Evaluating the differential effects of alternative welfare-to-work
  training components: A reanalysis of the california gain program.
\newblock \emph{Journal of Labor Economics}, 24\penalty0 (3):\penalty0
  521--566, 2006.

\bibitem[Imbens and Athey(2006)]{Imbens2006}
Guido Imbens and Susan Athey.
\newblock Identification and inference in nonlinear difference-in-difference
  models.
\newblock \emph{Econometrica}, 74:\penalty0 431--497, 02 2006.
\newblock \doi{10.2139/ssrn.311920}.

\bibitem[Imbens et~al.(2021)Imbens, Kallus, and Mao]{imbens2021controlling}
Guido Imbens, Nathan Kallus, and Xiaojie Mao.
\newblock Controlling for unmeasured confounding in panel data using minimal
  bridge functions: From two-way fixed effects to factor models.
\newblock \emph{arXiv preprint arXiv:2108.03849}, 2021.

\bibitem[Imbens and Newey(2009)]{imbens2009identification}
Guido~W Imbens and Whitney~K Newey.
\newblock Identification and estimation of triangular simultaneous equations
  models without additivity.
\newblock \emph{Econometrica}, 77\penalty0 (5):\penalty0 1481--1512, 2009.

\bibitem[Joffe and Greene(2009)]{joffe2009related}
Marshall~M Joffe and Tom Greene.
\newblock Related causal frameworks for surrogate outcomes.
\newblock \emph{Biometrics}, 65\penalty0 (2):\penalty0 530--538, 2009.

\bibitem[Kallus and Mao(2020)]{kallus2020role}
Nathan Kallus and Xiaojie Mao.
\newblock On the role of surrogates in the efficient estimation of treatment
  effects with limited outcome data.
\newblock \emph{arXiv preprint arXiv:2003.12408}, 2020.

\bibitem[Kallus et~al.(2018)Kallus, Puli, and Shalit]{kallus2018removing}
Nathan Kallus, Aahlad~Manas Puli, and Uri Shalit.
\newblock Removing hidden confounding by experimental grounding.
\newblock \emph{Advances in neural information processing systems}, 31, 2018.

\bibitem[Kallus et~al.(2021)Kallus, Mao, and Uehara]{kallus2021causal}
Nathan Kallus, Xiaojie Mao, and Masatoshi Uehara.
\newblock Causal inference under unmeasured confounding with negative controls:
  A minimax learning approach, 2021.

\bibitem[Kay(1986)]{kay1986markov}
Richard Kay.
\newblock A markov model for analysing cancer markers and disease states in
  survival studies.
\newblock \emph{Biometrics}, pages 855--865, 1986.

\bibitem[Kohavi et~al.(2012)Kohavi, Deng, Frasca, Longbotham, Walker, and
  Xu]{kohavi2012trustworthy}
Ron Kohavi, Alex Deng, Brian Frasca, Roger Longbotham, Toby Walker, and Ya~Xu.
\newblock Trustworthy online controlled experiments: Five puzzling outcomes
  explained.
\newblock In \emph{Proceedings of the 18th ACM SIGKDD international conference
  on Knowledge discovery and data mining}, pages 786--794, 2012.

\bibitem[Kress et~al.(1989)Kress, Maz'ya, and Kozlov]{kress1989linear}
Rainer Kress, V~Maz'ya, and V~Kozlov.
\newblock \emph{Linear integral equations}, volume~82.
\newblock Springer, 1989.

\bibitem[Liu et~al.(2013)Liu, Wang, Morris, Doney, Leese, Pearson, and
  Palmer]{liu2013glycemic}
Yiyuan Liu, Minghui Wang, Andrew~D Morris, Alex~SF Doney, Graham~P Leese,
  Ewan~R Pearson, and Colin~NA Palmer.
\newblock Glycemic exposure and blood pressure influencing progression and
  remission of diabetic retinopathy: a longitudinal cohort study in godarts.
\newblock \emph{Diabetes Care}, 36\penalty0 (12):\penalty0 3979--3984, 2013.

\bibitem[Marshall and Jones(1995)]{marshall1995multi}
Guillermo Marshall and Richard~H Jones.
\newblock Multi-state models and diabetic retinopathy.
\newblock \emph{Statistics in medicine}, 14\penalty0 (18):\penalty0 1975--1983,
  1995.

\bibitem[Mastouri et~al.(2021)Mastouri, Zhu, Gultchin, Korba, Silva, Kusner,
  Gretton, and Muandet]{mastouri2021proximal}
Afsaneh Mastouri, Yuchen Zhu, Limor Gultchin, Anna Korba, Ricardo Silva, Matt
  Kusner, Arthur Gretton, and Krikamol Muandet.
\newblock Proximal causal learning with kernels: Two-stage estimation and
  moment restriction.
\newblock In \emph{International Conference on Machine Learning}, pages
  7512--7523. PMLR, 2021.

\bibitem[{Miao} and {Tchetgen}(2018)]{miao2018a}
Wang {Miao} and Eric~Tchetgen {Tchetgen}.
\newblock A confounding bridge approach for double negative control inference
  on causal effects (supplement and sample codes are included).
\newblock \emph{arXiv preprint arXiv:1808.04945}, 2018.

\bibitem[Miao et~al.(2016)Miao, Geng, and Tchetgen]{Miao2016}
Wang Miao, Zhi Geng, and Eric Tchetgen.
\newblock Identifying causal effects with proxy variables of an unmeasured
  confounder.
\newblock \emph{Biometrika}, 105, 09 2016.
\newblock \doi{10.1093/biomet/asy038}.

\bibitem[Mohapatra et~al.(2007)Mohapatra, Rozelle, and
  Goodhue]{mohapatra2007rise}
Sandeep Mohapatra, Scott Rozelle, and Rachael Goodhue.
\newblock The rise of self-employment in rural china: development or distress?
\newblock \emph{World Development}, 35\penalty0 (1):\penalty0 163--181, 2007.

\bibitem[Nagasawa(2018)]{nagasawa2018treatment}
Kenichi Nagasawa.
\newblock Treatment effect estimation with noisy conditioning variables.
\newblock \emph{arXiv preprint arXiv:1811.00667}, 2018.

\bibitem[Newey and Stouli(2021)]{newey2021control}
Whitney Newey and Sami Stouli.
\newblock Control variables, discrete instruments, and identification of
  structural functions.
\newblock \emph{Journal of Econometrics}, 222\penalty0 (1):\penalty0 73--88,
  2021.

\bibitem[{Newey} and {Powell}(2003)]{newey2003instrumental}
Whitney~K. {Newey} and James~L. {Powell}.
\newblock Instrumental variable estimation of nonparametric models.
\newblock \emph{Econometrica}, 71\penalty0 (5):\penalty0 1565--1578, 2003.

\bibitem[Poterba and Summers(1986)]{poterba1986reporting}
James~M Poterba and Lawrence~H Summers.
\newblock Reporting errors and labor market dynamics.
\newblock \emph{Econometrica: Journal of the Econometric Society}, pages
  1319--1338, 1986.

\bibitem[Prentice(1989)]{prentice1989surrogate}
Ross~L Prentice.
\newblock Surrogate endpoints in clinical trials: definition and operational
  criteria.
\newblock \emph{Statistics in medicine}, 8\penalty0 (4):\penalty0 431--440,
  1989.

\bibitem[Price et~al.(2018)Price, Gilbert, and van~der
  Laan]{price2018estimation}
Brenda~L Price, Peter~B Gilbert, and Mark~J van~der Laan.
\newblock Estimation of the optimal surrogate based on a randomized trial.
\newblock \emph{Biometrics}, 74\penalty0 (4):\penalty0 1271--1281, 2018.

\bibitem[Qi et~al.(2021)Qi, Miao, and Zhang]{qi2021proximal}
Zhengling Qi, Rui Miao, and Xiaoke Zhang.
\newblock Proximal learning for individualized treatment regimes under
  unmeasured confounding.
\newblock \emph{arXiv preprint arXiv:2105.01187}, 2021.

\bibitem[Richardson and Robins(2013)]{richardson2013single}
Thomas~S Richardson and James~M Robins.
\newblock Single world intervention graphs (swigs): A unification of the
  counterfactual and graphical approaches to causality.
\newblock \emph{Center for the Statistics and the Social Sciences, University
  of Washington Series. Working Paper}, 128\penalty0 (30):\penalty0 2013, 2013.

\bibitem[Rosenman et~al.(2020)Rosenman, Basse, Owen, and
  Baiocchi]{rosenman2020combining}
Evan Rosenman, Guillaume Basse, Art Owen, and Michael Baiocchi.
\newblock Combining observational and experimental datasets using shrinkage
  estimators.
\newblock \emph{arXiv preprint arXiv:2002.06708}, 2020.

\bibitem[Rosenman et~al.(2022)Rosenman, Owen, Baiocchi, and
  Banack]{rosenman2022propensity}
Evan~TR Rosenman, Art~B Owen, Mike Baiocchi, and Hailey~R Banack.
\newblock Propensity score methods for merging observational and experimental
  datasets.
\newblock \emph{Statistics in Medicine}, 41\penalty0 (1):\penalty0 65--86,
  2022.

\bibitem[Rubin(1974)]{rubin1974estimating}
Donald~B Rubin.
\newblock Estimating causal effects of treatments in randomized and
  nonrandomized studies.
\newblock \emph{Journal of educational Psychology}, 66\penalty0 (5):\penalty0
  688, 1974.

\bibitem[{Shi} et~al.(2020){Shi}, {Miao}, {Nelson}, and
  {Tchetgen}]{shi2020multiply}
Xu~{Shi}, Wang {Miao}, Jennifer~C. {Nelson}, and Eric J.~Tchetgen {Tchetgen}.
\newblock Multiply robust causal inference with double‐negative control
  adjustment for categorical unmeasured confounding.
\newblock \emph{Journal of The Royal Statistical Society Series B-statistical
  Methodology}, 82\penalty0 (2):\penalty0 521--540, 2020.

\bibitem[Shi et~al.(2021)Shi, Miao, Hu, and Tchetgen]{shi2021theory}
Xu~Shi, Wang Miao, Mengtong Hu, and Eric~Tchetgen Tchetgen.
\newblock Theory for identification and inference with synthetic controls: A
  proximal causal inference framework.
\newblock \emph{arXiv preprint arXiv:2108.13935}, 2021.

\bibitem[Simon et~al.(2011)Simon, Friedman, Hastie, and Tibshirani]{glmnet}
Noah Simon, Jerome Friedman, Trevor Hastie, and Rob Tibshirani.
\newblock Regularization paths for cox's proportional hazards model via
  coordinate descent.
\newblock \emph{Journal of Statistical Software}, 39\penalty0 (5):\penalty0
  1--13, 2011.
\newblock URL \url{https://www.jstatsoft.org/v39/i05/}.

\bibitem[Singh(2020)]{SinghRahul2020KMfU}
Rahul Singh.
\newblock Kernel methods for unobserved confounding: Negative controls,
  proxies, and instruments.
\newblock \emph{arXiv preprint arXiv:2012.10315}, 2020.

\bibitem[Singh(2021)]{singh2021finite}
Rahul Singh.
\newblock A finite sample theorem for longitudinal causal inference with
  machine learning: Long term, dynamic, and mediated effects.
\newblock \emph{arXiv preprint arXiv:2112.14249}, 2021.

\bibitem[Singh(2022)]{singh2022generalized}
Rahul Singh.
\newblock Generalized kernel ridge regression for long term causal inference:
  Treatment effects, dose responses, and counterfactual distributions.
\newblock \emph{arXiv preprint arXiv:2201.05139}, 2022.

\bibitem[Singh et~al.(2019)Singh, Sahani, and Gretton]{NEURIPS2019_17b3c706}
Rahul Singh, Maneesh Sahani, and Arthur Gretton.
\newblock Kernel instrumental variable regression.
\newblock In \emph{Advances in Neural Information Processing Systems},
  volume~32, 2019.

\bibitem[Spirtes et~al.(2000)Spirtes, Glymour, Scheines, and
  Heckerman]{spirtes2000causation}
Peter Spirtes, Clark~N Glymour, Richard Scheines, and David Heckerman.
\newblock \emph{Causation, prediction, and search}.
\newblock MIT press, 2000.

\bibitem[Tchetgen~Tchetgen et~al.(2020)Tchetgen~Tchetgen, Ying, Cui, Shi, and
  Miao]{tchetgen2020introduction}
Eric~J Tchetgen~Tchetgen, Andrew Ying, Yifan Cui, Xu~Shi, and Wang Miao.
\newblock An introduction to proximal causal learning.
\newblock \emph{arXiv e-prints}, pages arXiv--2009, 2020.

\bibitem[Tennenholtz et~al.(2020)Tennenholtz, Shalit, and
  Mannor]{tennenholtz2020off}
Guy Tennenholtz, Uri Shalit, and Shie Mannor.
\newblock Off-policy evaluation in partially observable environments.
\newblock In \emph{Proceedings of the AAAI Conference on Artificial
  Intelligence}, volume~34, pages 10276--10283, 2020.

\bibitem[Tsiatis(2007)]{tsiatis2007semiparametric}
Anastasios Tsiatis.
\newblock \emph{Semiparametric theory and missing data}.
\newblock Springer Science \& Business Media, 2007.

\bibitem[Van~der Vaart(2000)]{van2000asymptotic}
Aad~W Van~der Vaart.
\newblock \emph{Asymptotic statistics}, volume~3.
\newblock Cambridge university press, 2000.

\bibitem[VanderWeele(2013)]{vanderweele2013surrogate}
Tyler~J VanderWeele.
\newblock Surrogate measures and consistent surrogates.
\newblock \emph{Biometrics}, 69\penalty0 (3):\penalty0 561--565, 2013.

\bibitem[Wang et~al.(2020)Wang, Parast, Tian, and Cai]{wang2020model}
Xuan Wang, Layla Parast, Lu~Tian, and Tianxi Cai.
\newblock Model-free approach to quantifying the proportion of treatment effect
  explained by a surrogate marker.
\newblock \emph{Biometrika}, 107\penalty0 (1):\penalty0 107--122, 2020.

\bibitem[Weir and Walley(2006)]{weir2006statistical}
Christopher~J Weir and Rosalind~J Walley.
\newblock Statistical evaluation of biomarkers as surrogate endpoints: a
  literature review.
\newblock \emph{Statistics in medicine}, 25\penalty0 (2):\penalty0 183--203,
  2006.

\bibitem[Wooldridge(2015)]{wooldridge2015control}
Jeffrey~M Wooldridge.
\newblock Control function methods in applied econometrics.
\newblock \emph{Journal of Human Resources}, 50\penalty0 (2):\penalty0
  420--445, 2015.

\bibitem[Xu et~al.(2021)Xu, Kanagawa, and Gretton]{xu2021deep}
Liyuan Xu, Heishiro Kanagawa, and Arthur Gretton.
\newblock Deep proxy causal learning and its application to confounded bandit
  policy evaluation.
\newblock \emph{Advances in Neural Information Processing Systems}, 34, 2021.

\bibitem[Yang et~al.(2020{\natexlab{a}})Yang, Eckles, Dhillon, and
  Aral]{yang2020targeting}
Jeremy Yang, Dean Eckles, Paramveer Dhillon, and Sinan Aral.
\newblock Targeting for long-term outcomes.
\newblock \emph{arXiv preprint arXiv:2010.15835}, 2020{\natexlab{a}}.

\bibitem[Yang and Ding(2019)]{yang2019combining}
Shu Yang and Peng Ding.
\newblock Combining multiple observational data sources to estimate causal
  effects.
\newblock \emph{Journal of the American Statistical Association}, 2019.

\bibitem[Yang et~al.(2020{\natexlab{b}})Yang, Zeng, and Wang]{yang2020elastic}
Shu Yang, Donglin Zeng, and Xiaofei Wang.
\newblock Elastic integrative analysis of randomized trial and real-world data
  for treatment heterogeneity estimation.
\newblock \emph{arXiv preprint arXiv:2005.10579}, 2020{\natexlab{b}}.

\bibitem[Yang et~al.(2020{\natexlab{c}})Yang, Zeng, and Wang]{yang2020improved}
Shu Yang, Donglin Zeng, and Xiaofei Wang.
\newblock Improved inference for heterogeneous treatment effects using
  real-world data subject to hidden confounding.
\newblock \emph{arXiv preprint arXiv:2007.12922}, 2020{\natexlab{c}}.

\bibitem[Ying et~al.(2021)Ying, Miao, Shi, and Tchetgen]{ying2021proximal}
Andrew Ying, Wang Miao, Xu~Shi, and Eric J~Tchetgen Tchetgen.
\newblock Proximal causal inference for complex longitudinal studies.
\newblock \emph{arXiv preprint arXiv:2109.07030}, 2021.

\bibitem[Zheng and Laan(2011)]{zheng2011cross}
Wenjing Zheng and Mark~J Laan.
\newblock Cross-validated targeted minimum-loss-based estimation.
\newblock In \emph{Targeted Learning}, pages 459--474. Springer, 2011.

\end{thebibliography}

\newpage 
\appendix

This appendix is organized as follows. In \Cref{sec: validation}, we compare the identification assumptions in our paper with assumptions in some existing literature. In \Cref{sec: selection-example}, we provide examples for the selection bridge function in a discrete data setting and a linear model setting respectively. In \Cref{sec: existence}, we discuss sufficient conditions for the existence of bridge functions in general nonparametric models. \Cref{sec: covariate-adaptive} is supplementary to \Cref{sec: relax-assump}, discussing the identification and estimation of the long-term average treatment effect when \Cref{assump: unconfound-exp,assump: ext-valid} are relaxed. \Cref{sec: more-extension} presents several extensions of our framework, while \Cref{sec: appendixsimulation} provides more simulation results and implementation details. 
Finally, \Cref{sec: proof} collects all proofs. 

\appendix

{%
\section{Comparison to Other Identifying Conditions}\label{sec: validation}

To identify the average long-term treatment effect using data combination, restrictions must be imposed on unobserved confounders. In this paper, we crucially leverage an assumed sequential structure in the short-term outcomes and an assumption that these are sufficiently strong proxies (our \Cref{assump: CI,assump: completeness}). In this section, we compare to two other sets of assumptions that, in addition to \Cref{assump: unconfound-obs,assump: unconfound-exp,assump: ext-valid}, minimally ensure identification, and we discuss their relationship to persistent confounding. Each of the following provide an alternative setting that is \emph{just identified}, meaning dropping any one assumption breaks identification. Indeed, in our paper we needed \Cref{assump: CI,assump: completeness} for identification.

\subsection{Comparison to \citet{athey2020combining}}

\citet{athey2020combining} assume \emph{latent unconfoundedness}: $Y(a) \perp A \mid S(a), X, G = O$. The assumption, which makes no explicit reference to presence or absence of persistent confounding, states that, were it observed, controlling for $(S(a),X)$ would be sufficient. Along with \Cref{assump: unconfound-obs,assump: unconfound-exp,assump: ext-valid} (or, \Cref{assump: ext-valid2}), they show this assumption ensures identification.

There are many ways to potentially satisfy this abstract assumption. One possibility is if $S(a)=f_a(U)$ is an invertible transformation of $U$. Such a production-function approach calls to mind, for example, assumption 3.2 of \citet{Imbens2006}. This, however, precludes lossyness or noise in the relationship between short-term outcomes and confounders. Alternatively, we can consider restrictions encoded solely by causal diagrams that would ensure latent unconfoundedness holds. 
One such diagram is shown in \Cref{figure: DAG-obs-1}: here the unobserved confounders are only \emph{short-term confounders} ($U_s$) in they that can only affect the treatment and short-term outcomes, but not the long-term outcome. Another diagram is shown in \Cref{figure: DAG-obs-2}: here the unobserved confounders are only \emph{outcome confounders} ($U_o$) in that they simultaneously affect the short-term and long-term outcomes, but not the treatment.

Latent unconfoundedness generally may not hold in a diagram where confounders are persistent (\Cref{figure: DAG-a}), and in fact it does not whenever a distribution is ``well-represented'' by such a diagram. 
{In \Cref{figure: DAG-obs-app},
we duplicate the single world intervention graph in  
\Cref{figure: SWIG-obs} for the observational data}. 
This summarizes the statistical independences among the potential outcomes and other variables in the observational data. 
In this graph, the path $A \leftarrow U \rightarrow Y(a)$ is \emph{not} blocked by the nodes $S(a)$ and $X$, so $A$ and $Y(a)$ are \emph{not} $d$-separated by $S(a)$ and $X$. This means that the latent unconfoundedness assumption in \cite{athey2020combining}
is violated when the distribution of the random variables $(X, U, A, S(a), Y(a))$ is \emph{faithful} to the single world intervention graph in \Cref{figure: DAG-obs-app} \citep{spirtes2000causation}, roughly meaning that the graph is minimal for the distribution.\footnote{Formally, we say that a distribution $\mathbb{P}$ on the nodes of graph $\mathcal{G}$
 is faithful to the graph $\mathcal{G}$ when for any random variables $(A, B, C)$ in the graph, $A \perp B \mid C$ under the distribution $\mathbb{P}$ \emph{if and only if} $A$ and $B$ are $d$-separated
 by $C$ in  the graph $\mathcal{G}$ \citep{spirtes2000causation}.}

\begin{figure}[t]
\begin{subfigure}[b]{0.5\textwidth}
\centering 
\begin{tikzpicture}
\node[draw, circle, text centered] (A) {$A$};
\node[draw, circle, text centered, right=4cm of A] (Y) {$Y$};
\node[draw, circle, text centered, right=1.6cm of A] (S) {$S$};
\node[draw, circle, text centered, above right=1cm and 0.8cm of A] (X) {$X$};
\node[draw, circle, dashed, text centered, above right=1cm and 3.2cm of A] (U) {$U_{\op{s}}$};

\draw[->] (A) -- (S);
\draw[->] (A) to [bend right] (Y);
\draw[->] (S) -- (Y);
\draw[->] (X) -- (A);
\draw[->] (X) -- (S);
\draw[->] (X) -- (Y);
\draw[->] (U) -- (A);
\draw[->] (U) -- (S);
\draw[<->] (X) -- (U);
\end{tikzpicture}%
\caption{Short-term confounders $U_{\op{s}}$.}%
\label{figure: DAG-obs-1}%
\end{subfigure}
\begin{subfigure}[b]{0.5\textwidth}
\centering 
\begin{tikzpicture}
\node[draw, circle, text centered] (A) {$A$};
\node[draw, circle, text centered, right=4cm of A] (Y) {$Y$};
\node[draw, circle, text centered, right=1.6cm of A] (S) {$S$};
\node[draw, circle, text centered, above right=1cm and 0.3cm of A] (X) {$X$};
\node[draw, circle, dashed, text centered, above right=1cm and 3.2cm of A] (U) {$U_{\op{o}}$};

\draw[->] (A) -- (S);
\draw[->] (A) to [bend right] (Y);
\draw[->] (S) -- (Y);
\draw[->] (X) -- (A);
\draw[->] (X) -- (S);
\draw[->] (X) -- (Y);
\draw[->] (U) -- (S);
\draw[->] (U) -- (S);
\draw[->] (U) -- (Y);
\draw[<->] (X) -- (U);
\end{tikzpicture}%
\caption{Outcome confounders $U_{\op{o}}$.}%
\label{figure: DAG-obs-2}%
\end{subfigure}
\caption{Short-term confounders and outcome confounders in the observational data.}%
\end{figure}
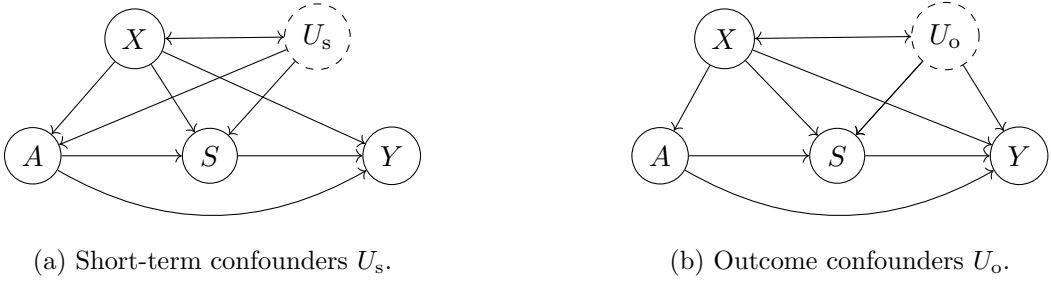

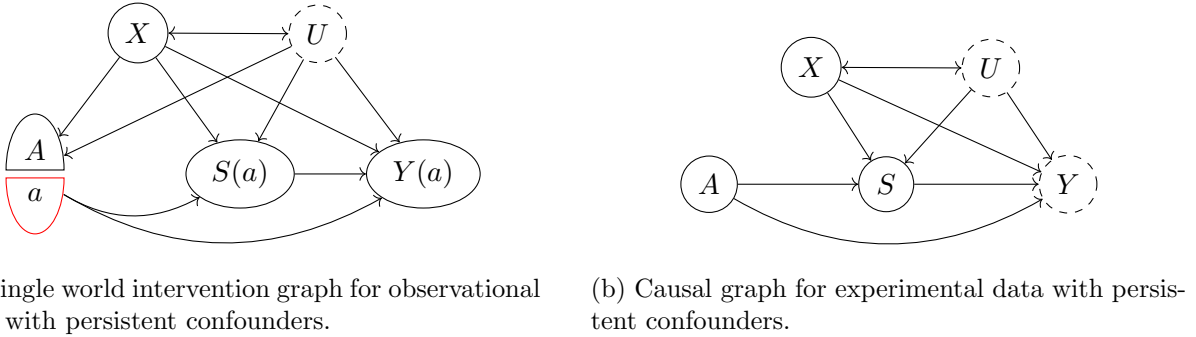
\begin{figure}[h]
\begin{subfigure}[b]{0.48\textwidth}
\centering 
\begin{tikzpicture}
\tikzset{swig hsplit={gap=3pt,
line color lower=red}}
\node[name=A, shape=swig hsplit]{
\nodepart{upper}{$A$}
\nodepart{lower}{{$a$}} };
\node[draw, ellipse, text centered, right=4cm of A] (Y) {$Y({a})$};
\node[draw, ellipse, text centered, right=1.6cm of A] (S) {$S({a})$};
\node[draw, circle, text centered, above right=1cm and 0.8cm of A] (X) {$X$};
\node[draw, circle, dashed, text centered, above right=1cm and 3.2cm of A] (U) {$U$};

\draw[->] (A) to [bend right] (S);
\draw[->] (A) to [bend right] (Y);
\draw[->] (S) -- (Y);
\draw[->] (X) -- (A);
\draw[->] (X) -- (S);
\draw[->] (X) -- (Y);
\draw[->] (U) -- (A);
\draw[->] (U) -- (S);
\draw[->] (U) -- (Y);
\draw[<->] (X) -- (U);
\end{tikzpicture}
\caption{Single world intervention graph for observational data with persistent confounders.}
\label{figure: DAG-obs-app}
\end{subfigure}
\hfill
\begin{subfigure}[b]{0.48\textwidth}
\centering 
\begin{tikzpicture}
\node[draw, circle, text centered] (A) {$A$};
\node[draw, dashed, circle, text centered, right=4cm of A] (Y) {$Y$};
\node[draw, circle, text centered, right=1.6cm of A] (S) {$S$};
\node[draw, circle, text centered, above right=1cm and 0.8cm of A] (X) {$X$};
\node[draw, circle, dashed, text centered, above right=1cm and 3.2cm of A] (U) {$U$};

\draw[->] (A) -- (S);
\draw[->] (A) to [bend right] (Y);
\draw[->] (S) -- (Y);
\draw[->] (X) -- (S);
\draw[->] (X) -- (Y);
\draw[->] (U) -- (S);
\draw[->] (U) -- (Y);
\draw[<->] (X) -- (U);
\end{tikzpicture}
\caption{Causal graph for experimental data with persistent confounders.}
\label{figure: DAG-exp-app}
\end{subfigure}
\caption{Graphs for persistent confounders.}
\label{figure: DAG-app}
\end{figure}

\subsection{Comparison to \citet{athey2019surrogate}}

\citet{athey2019surrogate} assume that the long-term outcome is independent of the treatment given the short-term outcomes, $A \perp Y \mid S, X, G = E$. This is based on the surrogate criterion proposed by \cite{prentice1989surrogate}. Crucially they show this condition enables identification even when $A$ is missing in the observational data, which can be extremely practical. This condition, however, rules out any direct effect of the treatment on the long-term outcome and any confounding between short-term and long-term outcomes, as might be induced by a persistent confounder. 

In \Cref{figure: DAG-exp-app}, we duplicate the causal diagram in \Cref{figure: DAG-exp-a}, which describes the causal relationship between variables in the experimental data. 
 In the setting of \Cref{figure: DAG-exp-app}, the surragacy condition is violated when the distribution of the random variables $(X,U,A,S,Y)$ is faithful. 
 (Note we do not use a single world intervention graph here as the assumption is made on factual variables, rather than on potential outcomes.)
 Indeed, in \Cref{figure: DAG-exp-app}, the short-term outcomes $S$ are colliders between the treatment $A$ and the persistent confounders $U$, so conditioning on $S$ induces dependence between the treatment $A$ and the long-term outcome $Y$. 
 Moreover, the treatment $A$ can also have direct causal effect on the long-term outcome $Y$. 
 Therefore, unless the dependence due to the direct causal effect of the treatment and the dependence due to conditioning on colliders happen to cancel with each other (which cannot happen if the distribution is faithful), the surrogacy condition in \cite{athey2019surrogate} is violated.
}

\section{Selection Bridge Functions in Special Examples}\label{sec: selection-example}
\subsection{Discrete Setting}
Recall that in \Cref{ex: discrete}, we consider $\Scal_1 = \Scal_2 = \Scal_3 = \braces{s_{(j)}: j = 1, \dots, M_s}$ and $\Ucal = \braces{u_{(k)}: k = 1, \dots, M_u}$. 
For any $s_2 \in \Scal_2, a \in \Acal, x \in \Xcal$, let $P(\mathbf{S}_1 \mid s_2, a, \mathbf{U}, x) \in \R{M_s\times M_u}$ denote the matrix whose $(j, k)$th element is 
$$\Prb{S_1 = s_{(j)} \mid S_2 = s_2, A = a, U = u_{(k)}, X = x, G = O},$$
and $r(s_2, \mathbf{U}, x; a) \in \R{M_u}$ denote the vector whose $k$th element is 
$${p\prns{s_2, u_{(k)}, x\mid a, G = E}}/{p\prns{s_2, u_{(k)}, x \mid a, G = O}}.$$
The existence of a selection bridge function is equivalent to the existence of a solution $z \in \R{M_s}$ to the following linear equation system for any $s_2 \in \Scal_2, a \in\Acal, x\in\Xcal$:
\begin{align*}
\bracks{P(\mathbf{S}_1 \mid s_2, a, \mathbf{U}, x)}^\top z = r(s_2, \mathbf{U}, x; a).
\end{align*}
One sufficient condition for the existence of solutions to the equation above is that 
the matrix $P(\mathbf{S}_1 \mid s_2, a, \mathbf{U}, x)$ has full column rank for any $s_2 \in \Scal_2, a \in\Acal, x\in\Xcal$. 
This full column rank condition means that $S_1$ is strongly informative for $U$.

\subsection{Linear Models}\label{sec: linear}
Recall that in \Cref{ex: linear}, $\prns{Y, S_3, S_2, S_1}$ are generated from the following linear structural equation system: 
\begin{align*}
&Y = \tau_y A + \alpha_y^\top S_3 + \beta_y^\top X + \gamma_y^\top U + \epsilon_y, \\ 
&S_j = \tau_j A + \alpha_j S_{j-1} + \beta_j X + \gamma_j U + \epsilon_j, ~ j \in \braces{3, 2}\\
&S_1 = \tau_1 A +  \beta_1 X + \gamma_1 U + \epsilon_1,
\end{align*}
where $\tau_y, (\tau_j, \alpha_y, \beta_y, \gamma_y), (\alpha_j, \beta_j, \gamma_j)$ are scalers, vectors, and matrices of conformable sizes respectively, and $\epsilon_y, \epsilon_j$ are independent mean-zero noise terms such that $\epsilon_y \perp (S, A, U, X)$ and $\epsilon_j \perp (S_{j-1}, \dots, S_1, A, U, X)$. 

We further assume 
$$\Prb{A = 1\mid U, X, G = E} = 1/2, ~~ \Prb{A = 1\mid U, X, G = O} = \prns{1+\exp(\kappa_1^\top U + \kappa_2^\top X)}^{-1}.$$
We also assume that $\prns{\epsilon_3, \epsilon_2, \epsilon_1}$ follows a joint Gaussian distribution with zero mean and a diagonal covariance matrix. 
Denote the covariance matrix for $\epsilon_j$ as $\sigma_j^2 I_j$ for $j = 1, \dots, 3$ where $I_j$ is an identity matrix of formable size. 

\begin{proposition}\label{prop: sel-linear}
Given the data generating process described above, $S_1 \mid S_2, A, U, X, G = O$ follows a Gaussian distribution with conditional expectation
\begin{align*}
\Eb{S_1 \mid S_2, A, U, X, G = O} = \lambda_1 S_2 + \lambda_2 A + \lambda_3 X + \lambda_4 U,
\end{align*} 
where 
\begin{align*}
&\lambda_1 = \sigma_1^2\alpha_2^\top\prns{\sigma_1^2 \alpha_2\alpha_2^\top + \sigma_2^2 I_2}^{-1}, \\
&\lambda_2 = \prns{I_1 - \sigma_1^2 \alpha_2^\top\prns{\sigma_1^2 \alpha_2\alpha_2^\top + \sigma_2^2 I_2}^{-1}\alpha_2}\tau_1 - \sigma_1^2\alpha_2^\top\prns{\sigma_1^2 \alpha_2\alpha_2^\top + \sigma_2^2 I_2}^{-1}\tau_2 \\
&\lambda_3 = \prns{I_1 - \sigma_1^2 \alpha_2^\top\prns{\sigma_1^2 \alpha_2\alpha_2^\top + \sigma_2^2 I_2}^{-1}\alpha_2}\beta_1 - \sigma_1^2\alpha_2^\top\prns{\sigma_1^2 \alpha_2\alpha_2^\top + \sigma_2^2 I_2}^{-1}\beta_2 \\
&\lambda_4 = \prns{I_1 - \sigma_1^2 \alpha_2^\top\prns{\sigma_1^2 \alpha_2\alpha_2^\top + \sigma_2^2 I_2}^{-1}\alpha_2}\gamma_1 - \sigma_1^2\alpha_2^\top\prns{\sigma_1^2 \alpha_2\alpha_2^\top + \sigma_2^2 I_2}^{-1}\gamma_2.
\end{align*}
When $\lambda_4$ has full column rank, then for any $\tilde\theta_1$ such that $\tilde\theta_1^\top\lambda_4 = \kappa_2^\top$ and $a \in \Acal$, there exists a selection bridge function of the  following form for some matrices $\tilde\theta_2, \tilde\theta_0$ of conformable sizes and some constants $c_{1, a}, c_{0, a}$: 
\begin{align*}
q_0\prns{S_2, S_1, a, X} = c_{1, a}\exp\prns{(-1)^a\prns{\tilde\theta_2^\top S_2 + \tilde\theta_1^\top S_1 + \tilde \theta_0^\top X}} + c_{0, a}.
\end{align*}
\end{proposition}

\begin{figure}
     \centering
     \begin{subfigure}[b]{0.48\textwidth}
         \centering
         \includegraphics[width=\textwidth]{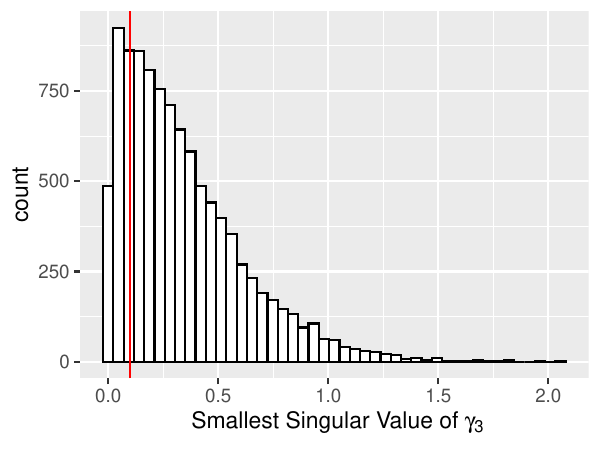}
     \end{subfigure}
     \hfill
     \begin{subfigure}[b]{0.48\textwidth}
         \centering
         \includegraphics[width=\textwidth]{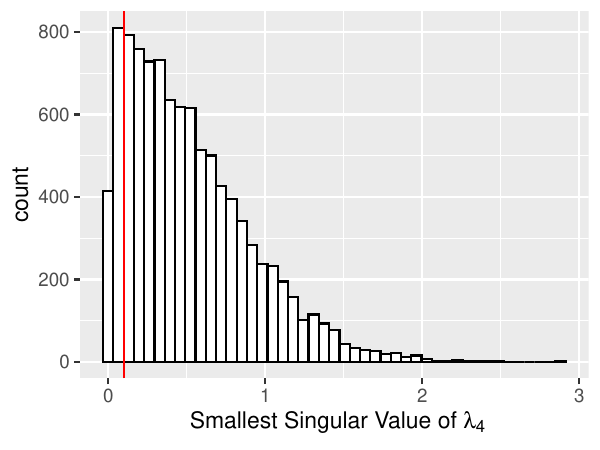}
     \end{subfigure}
        \caption{Distributions of the smallest singular values of the $\lambda_3$ and $\gamma_4$ matrices over $10000$ replications. The vertical bar corresponds to the $0.1$ singular value.}
        \label{fig: smallest-sv}
\end{figure}

{According to \Cref{prop: sel-linear,ex: linear}, the outcome bridge function and selection bridge function exist in this linear model setting if the matrix $\gamma_3$ and $\lambda_4$ have full column rank.  
To illustrate these existence conditions, we also run a simple simulation study. Specifically, we generate data according  to the linear model above. We set 
$\text{dim}(S_1) = \text{dim}(S_2) = \text{dim}(S_3) = 2$, draw all coefficients $\tau,\alpha,\beta,\gamma$'s from the standard normal distribution, and all noise terms from the mean-zero normal distribution with variance $0.5$. We generate $10000$ instances, compute the smallest singular values of the corresponding $\gamma_3$ and $\lambda_4$ matrices, and report their distributions in \Cref{fig: smallest-sv}.    
We observe that the smallest singular value of $\gamma_3$ is larger than $0.1$ around $88\%$ of time and the smallest singular value of $\lambda_4$ is larger than $0.1$ around $81\%$ of time. 
    These show that the existence of bridge functions in \Cref{assump: bridge,assump: bridge2}  may not  always hold but it does hold in quite many scenarios. }

\section{Completeness Conditions and Existence of Bridge Functions}\label{sec: existence}
The conditional moment equations in \Cref{eq: bridge-U,eq: bridge2-U} that define  outcome bridge functions and selection bridge functions are Fredholm integral equations of the first kind. 
Following \cite{Miao2016}, we characterize the existence of their solutions (\ie, the outcome and selection bridge functions) by singular value decomposition of compact operators \citep{Carrasco2007}. 

Let $L_2\prns{p(z)}$ denote the space of all square integrable functions of $z$ with respect to the distribution $p(z)$. It is a Hilbert space with inner product $\langle f_1, f_2\rangle = \int f_1(z)f_2(z)p(z)\diff z$.
Consider linear operators $\mathcal{T}_{s_2, a, x}: L_2\prns{p(s_3 \mid s_2, a, x)} \to L_2\prns{p(u \mid s_2, a, x)}$, $\mathcal{T}'_{s_2, a, x}: L_2\prns{p(s_1 \mid s_2, a, x)} \to L_2\prns{p(u \mid s_2, a, x)}$ defined as follows: 
\begin{align*}
&\bracks{\mathcal{T}_{s_2, a, x}h}\prns{s_2, a, u, x} = \Eb{h\prns{S_3, S_2, A, X} \mid S_2 = s_2, A = a, U = u, X = x, G = O}, \\
&\bracks{\mathcal{T}'_{s_2, a, x}q}\prns{s_2, a, u, x} = \Eb{q\prns{S_2, S_1, A, X} \mid S_2 = s_2, A = a, U = u, X = x, G = O}.
\end{align*}
\begin{assumption}\label{assump: compact-opt}
For any $s_2 \in \Scal_2, a \in \Acal, x \in \Acal$, 
\begin{enumerate}
\item $\iint p(s_3 \mid s_2, a, u, x)p(u \mid s_3, s_2, a, x)\diff s_3 \diff u < \infty$.
\item $\iint p(s_1 \mid s_2, a, u, x)p(u \mid s_1, s_2, a, x)\diff s_1 \diff u < \infty$.
\end{enumerate}
\end{assumption}
According to Example 2.3 in \cite{Carrasco2007}, \Cref{assump: compact-opt} ensures that for any $s_2 \in \Scal_2, a \in \Acal, x \in \Acal$, $\mathcal{T}_{s_2, a, x}, \mathcal{T}'_{s_2, a, x}$ are both compact operators. 
Then by Theorem 2.41 in \cite{Carrasco2007}, both of them admit singular value decomposition. Namely, there exist $(\lambda_{s_2, a, x, j}, \psi_{s_2, a, x, j}, \phi_{s_2, a, x, j})_{j=1}^{\infty}$ and $(\lambda'_{s_2, a, x, j}, \psi'_{s_2, a, x, j}, \phi'_{s_2, a, x, j})_{j=1}^{\infty}$ such that for any $j$, 
\begin{align*}
 &\mathcal{T}_{s_2, a, x} \psi_{s_2, a, x, j} = \lambda_{s_2, a, x, j} \phi_{s_2, a, x, j}  \\
 &\mathcal{T}'_{s_2, a, x} \psi'_{s_2, a, x, j} = \lambda'_{s_2, a, x, j} \phi'_{s_2, a, x, j} .
 \end{align*} 
 \begin{assumption}\label{assump: existence-reg}
 For any $s_2 \in \Scal_2, a \in \Acal, x \in \Acal$, 
\begin{enumerate}
\item $\Eb{Y \mid s_2, a, u, x, G = O}$ and $\frac{p(s_2, u, x \mid a, G = E)}{p(s_2, u, x \mid a, G = O)}$ both belong to $L_2\prns{p(u \mid s_2, a, x)}$.
\item $\sum_{j=1}^n  \lambda_{s_2, a, x, j}^{-2}\abs{\langle \Eb{Y \mid s_2, a, u, x, G = O}, \phi_{s_2, a, x, j} \rangle}^2 < \infty$.
\item $\sum_{j=1}^n  \lambda_{s_2, a, x, j}^{\prime-2}\abs{\langle \frac{p(s_2, u, x \mid a, G = E)}{p(s_2, u, x \mid a, G = O)}, \phi'_{s_2, a, x, j} \rangle}^2 < \infty$.
\end{enumerate}
 \end{assumption}
 Under regularity conditions in \Cref{assump: existence-reg,assump: compact-opt}, it can be shown that completeness conditions in \Cref{assump: completeness} guarantee the existence of bridge functions. 
 \begin{proposition}[Existence of Bridge Functions]\label{prop: existence}
 Suppose that \Cref{assump: existence-reg,assump: compact-opt} hold. 
\begin{enumerate}
\item If the completeness condition in \Cref{assump: completeness}  condition \ref{assump: completeness-1} holds, then there exists an outcome bridge function $h_0$ satisfying \Cref{eq: bridge-U}. 
\item If the completeness condition in \Cref{assump: completeness}  condition \ref{assump: completeness-2} holds, then there exists an outcome bridge function $q_0$ satisfying \Cref{eq: bridge2-U}. 
\end{enumerate}
 \end{proposition}
\Cref{prop: existence} can be proved by Picard's Theorem \citep[Theorem 15.18]{kress1989linear}. See Lemma 2 in \cite{Miao2016} or Lemma 13 and 14 in \cite{kallus2021causal} for details. 

\section{Relaxing \Cref{assump: unconfound-exp,assump: ext-valid}}\label{sec: covariate-adaptive}
In this section, we present additional identification results under \Cref{assump: unconfound-exp2,assump: ext-valid2} instead of the stronger conditions in \Cref{assump: unconfound-exp,assump: ext-valid}, and discuss their relations to the existing literature. 
We also discuss how to estimate the average long-term treatment effect in this case, based on the doubly robust identification strategy in \Cref{corollary: covariate-exp-dr}.

\subsection{Identification}\label{sec: covariate-adaptive-identify}
In \Cref{corollary: covariate-exp-dr}, we consider extending the doubly robust identification strategy in \Cref{thm: identification-DR}, which involves both outcome and selection bridge functions.
We now show that based on \Cref{corollary: covariate-exp-dr}, we can also extend \Cref{thm: identification1,thm: identification2}.

\begin{corollary}\label{corollary: covariate-exp}
Suppose \Cref{assump: CI,assump: ext-valid2,assump: unconfound-exp2,assump: unconfound-obs} hold. 
\begin{enumerate}
\item If further the completeness condition in \Cref{assump: completeness}  condition \ref{assump: completeness-2} and \Cref{assump: bridge} hold, then the average long-term treatment effect can be identified by any function $h_0$ that satisfies \Cref{eq: bridge-obs}:
\begin{align}\label{eq: identification-1-X}
\tau 
    &= \sum_{a \in \braces{0, 1}}\prns{-1}^{1-a} \Eb{\Eb{h_0\prns{\Sc, \Sb, A, X} \mid A = a, X, G = E} \mid G= O}.
\end{align}
\item If further the completeness condition in \Cref{assump: completeness} condition \ref{assump: completeness-1} and \Cref{assump: bridge2} hold, then the average long-term treatment effect can be identified by any function
 $q_0$ that satisfies \Cref{eq: bridge2-obs-1} or \Cref{eq: bridge2-obs-2}:
 \begin{align}\label{eq: identification-2-X}
 \tau 
    &= \sum_{a \in \braces{0, 1}}\prns{-1}^{1-a} \mathbb{E}\bigg[
    \frac{\Prb{G=E\mid A=a}\Prb{G=O\mid X}}{\Prb{G=O\mid A=a}\Prb{G=E\mid X}}\frac{\indic{A=a}}{\Prb{A=a\mid X, G=E}} \nonumber\\
    &\qquad\qquad\qquad\qquad\qquad\qquad\qquad\qquad\qquad\qquad \times q_0\left(S_{2}, S_{1}, A, X\right)Y\mid G=O\bigg]
 \end{align}
\end{enumerate} 
\end{corollary}
\begin{proof}[Proof for \Cref{corollary: covariate-exp}]
Obviously, \Cref{eq: identification-1-X} can be proved by setting $h = h_0, q = 0$ in \Cref{corollary: covariate-exp-dr} and \Cref{eq: identification-2-X} can be proved by setting $q = q_0, h = 0$ in \Cref{corollary: covariate-exp-dr}.
\end{proof}

We note that the two identification strategies in \Cref{corollary: covariate-exp} are closely related to those in \cite{athey2020combining,ghassami2022combining}. 
As we discussed in \Cref{remark: connection}, when there is no persistent confounder, we can let $S_1 = S_3 = \emptyset$ and $S = S_2$.
Then
$h_0\prns{S_2, A, X} = \Eb{Y \mid S, A, X, G=O}$ is the unique solution to \Cref{eq: bridge-obs}. As a result, the identification strategy in \Cref{eq: identification-1-X} exactly recovers the identification strategy in Theorem 1 in \cite{athey2020combining}.
Moreover, in the setup of the concurrent work  \cite{ghassami2022combining}, if we use $S_3$ as their short-term outcomes, $S_1$ as their auxiliary proxies, and condition on $S_2$ appropriately, then the identification strategies in their Theorems 9 and 10 coincide with ours in \Cref{eq: identification-1-X,eq: identification-2-X} respectively.
See also discussions in \Cref{sec: literature-data-comb} for additional comparisons.

Under the weaker conditions in \Cref{assump: unconfound-exp2,assump: ext-valid2}, \Cref{corollary: covariate-exp,corollary: covariate-exp-dr} shows that we need more complex identification strategies for the average long-term treatment effect over the observational data distribution. Actually, even in this case, the simpler identification strategies in \Cref{sec: identify-OBF,sec: identify-SBF,sec: identify-DR} are still useful. 
Below we show that under an additional assumption, they can identify average long-term treatment effect over the experimental data distribution.
\begin{corollary}\label{corollary: identify-exp}
Suppose \Cref{assump: CI,assump: ext-valid2,assump: unconfound-exp2,assump: unconfound-obs,assump: completeness,assump: bridge,assump: bridge2} hold and $Y(a) \perp G \mid S(a), U, X$. Then \Cref{eq: identification-1} in \Cref{thm: identification1}, \Cref{eq: identification-2} in \Cref{thm: identification2} and \Cref{eq: DR} in \Cref{thm: identification-DR} all identify the average long-term treatment effect over the experimental data distribution, \ie,
\begin{align*}
\tau_E = \Eb{Y(1)-Y(0) \mid G = E},
\end{align*}
\end{corollary}
In \Cref{corollary: identify-exp}, we still assume the weaker conditions in \Cref{assump: unconfound-exp2,assump: ext-valid2}. But we additionally require that the experimental and observational data share a common conditional distribution of the potential long-term outcome.
This additional assumption ensures that the bridge functions defined in terms of the observational data distribution can also be used to identify the 
average long-term treatment effect over the experimental data distribution.

{Finally, we note that the selection bridge functions can be used to identify more general parameters than the average treatment effects considered so far.}
{
\begin{corollary}\label{corollary: sel-bridge-dist}
Suppose \Cref{assump: CI,assump: ext-valid2,assump: unconfound-exp2,assump: unconfound-obs} and the assumptions in \Cref{corollary: covariate-exp} condition 2 hold. Then for any function
 $q_0$ that satisfies \Cref{eq: bridge2-obs-1} or \Cref{eq: bridge2-obs-2}, and any transformation $r: \Ycal\mapsto \R{}$, we have 
 \begin{align*}
 \Eb{r(Y(a)) \mid G = O} 
    &= \mathbb{E}\bigg[
    \frac{\Prb{G=E\mid A=a}\Prb{G=O\mid X}}{\Prb{G=O\mid A=a}\Prb{G=E\mid X}}\frac{\indic{A=a}}{\Prb{A=a\mid X, G=E}} \nonumber\\
    &\qquad\qquad\qquad\qquad\qquad\qquad\qquad\qquad\qquad\qquad \times q_0\left(S_{2}, S_{1}, A, X\right)r(Y)\mid G=O\bigg].
 \end{align*}
\end{corollary}
In particular, when applying \Cref{corollary: sel-bridge-dist} to the indicator function $r(\cdot) = \indic{\cdot \le y}$ for all $y \in \Ycal$, we can identify the entire distribution of the counterfactual long term outcome $Y(a)$. 
}

\subsection{Estimation}\label{sec: covariate-adaptive-estimate}
{
    We can again leverage the doubly robust identification strategy in \Cref{corollary: covariate-exp-dr} to estimate the average long-term treatment effect. This involves some nuisance functions/parameters $\eta^* = (h_0, \bar{h}_E, q_0, \alpha_0, \beta_0)$, 
    where $h_0, q_0$ are bridge functions given in \cref{eq: bridge-obs,eq: bridge2-obs-2}, $\bar{h}_{E}$ is given in \Cref{corollary: covariate-exp-dr}, and 
    \begin{align*}
    \alpha_0(A, X) 
        &= \frac{\Prb{G=E}\Prb{G=O\mid X}}{\Prb{G=O}\Prb{G=E\mid X}}\times \frac{1}{A\Prb{A=1\mid X, G=E} + (1-A)\Prb{A=0\mid X, G=E}} \\
    \beta_0(A, X) 
        &= \frac{\Prb{G=E\mid A}\Prb{G=O\mid X}}{\Prb{G=O\mid A}\Prb{G=E\mid X}} \times \frac{1}{A\Prb{A=1\mid X, G=E} + (1-A)\Prb{A=0\mid X, G=E}}, \\
    \end{align*}
    According to \Cref{corollary: covariate-exp-dr}, once we know these nuisance functions/parameters, we immediately have 
\begin{align}
\tau = \sum_{a\in\braces{0, 1}}\prns{-1}^{1-a}\big\{
   &\Eb{\phi_1\prns{Y, S, a, X; \eta^*} \mid G = E}\nonumber\\
    &+ \Eb{\phi_2\prns{Y, S, a, X; \eta^*} \mid G = O} + \Eb{\phi_3\prns{Y, S, a, X; \eta^*} \mid G = O}\big\},\label{eq: DR-est-eq}
\end{align}
where 
\begin{align*}
&\phi_1\prns{Y, S, a, X; \eta^*} = \indic{A=a}\alpha_0(A, X)\prns{h_0\prns{\Sc, \Sb, A, X} - \bar{h}_{E, 0}\prns{A, X} }, \\
&\phi_2\prns{Y, S, a, X; \eta^*} = \bar{h}_{E, 0}\prns{a, X} = \Eb{h_0\prns{\Sc, \Sb, a, X} \mid A = a, X, G = E}, \\
&\phi_3\prns{Y, S, a, X; \eta^*} = \indic{A=a}\beta_0(A, X)q_0\left(S_{2}, S_{1}, A, X\right)\left(Y-h_0\left(S_{3}, S_{2}, A, X\right)\right).
\end{align*}
}

{
    In the following lemma, we prove that the doubly robust equation in \Cref{eq: DR-est-eq} satisfies the so-called \emph{Neyman Orthogonality} property.
\begin{lemma}\label{lemma: neyman-orthogonality}
The estimating equation implied by \Cref{eq: DR-est-eq} satisfies the Neyman Orthogonality property, namely, the pathwise derivative of the following map at $\eta^*$ along any feasible direction is equal to $0$:
\begin{align}
\eta \mapsto \sum_{a\in\braces{0, 1}}\prns{-1}^{1-a}\big\{
   &\Eb{\phi_1\prns{Y, S, a, X; \eta} \mid G = O}\nonumber \\
   &+ \Eb{\phi_2\prns{Y, S, a, X; \eta} \mid G = E} + \Eb{\phi_3\prns{Y, S, a, X; \eta} \mid G = O}\big\}. \label{eq: orthogonality-map}
\end{align}
\end{lemma}
The Neyman orthogonality property plays a central role in the recent debiased machine learning literature \citep[\eg, ][]{chernozhukov2019double}. This property guarantees that the treatment effect estimators constructed from the doubly robust equation is insensitive to estimation errors of the nuisance functions/parameters. In particular, we  can construct a treatment effect estimator by adapting the cross-fitted estimator in \Cref{def: mean-est} to the doubly robust equation in \cref{eq: DR-est-eq}.
While \Cref{def: mean-est} only randomly splits the observational data, here we also need to split the experimental data. 
This is because the nuisance functions $\bar{h}_{E, 0}$ and $\alpha_0$ to be evaluated on the experimental data are themselves estimated from the experimental data. 
\begin{definition}\label{def: dr-cf-est}
Fix $a \in \mathcal{A}$ and an integer $K \ge 2$.
\begin{enumerate}
\item Randomly split the observational data $\mathcal{D}_O$ and experimental data $\mathcal{D}_E$ into $K$ (approximately) even folds  denoted as $\mathcal{D}_{O, 1}, \dots, \mathcal{D}_{O, K}$ and $\mathcal{D}_{E, 1}, \dots, \mathcal{D}_{E, K}$, respectively. 
\item For $k = 1, \dots, K$, use all  data other than the $k$th fold observational and experimental data, \ie, $\prns{\cup_{j \ne k} \mathcal{D}_{O, j}} \cup \prns{\cup_{j \ne k} \mathcal{D}_{E, j}}$,  to construct the bridge function estimators $\hat h_k, \hat q_k$, estimator of the function $\bar{h}_{E, 0}$ denoted as $\hat{\bar{h}}_k$, and estimators for $\alpha_0, \beta_0$ denoted by $\hat \alpha_k, \hat \beta_k$ respectively. 
\item Use the following estimator 
\begin{align*}
&\hat\tau = \frac{1}{K}\sum_{k=1}^K\sum_{a\in\{0, 1\}}(-1)^{(1-a)}\bigg[\frac{1}{n_{E, k}}\sum_{i \in \mathcal{D}_{E, k}}\indic{A_i = a}\hat\alpha_k(A_i, X_i)\prns{\hat h_k(S_{3, i}, S_{2, i}, A_i, X_i) - \hat{\bar{h}}_k(A_i, X_i)} \\
&\qquad + \frac{1}{n_{O, k}}\sum_{i\in \mathcal{D}_{O, k}}\hat{\bar{h}}_k(a, X_i) +  \indic{A_i = a}\hat\beta_k(A_i, X_i)\hat q_k(S_{2, i}, S_{1, i}, A_i, X_i)\prns{Y_i - \hat h_k(S_{3, i}, S_{2, i}, A_i, X_i)}\bigg],
\end{align*}
where $n_{E, k}$ and $n_{O, k}$ are the sample sizes of the $k$-th folds of experimental data $\mathcal{D}_{E, k}$ and observational data $\mathcal{D}_{O, k}$, respectively. 
\end{enumerate}
\end{definition}
}

{We note that the form of the estimator $\hat\tau$ in \Cref{def: dr-cf-est} is similar to the estimator in \cite{singh2021finite}, so we can similarly show that the estimator $\hat\tau$ is $\sqrt{n}$-consistent and asymptotically normal under high level conditions on the nuisance estimators.}  

\begin{theorem}\label{thm: DR-est-normality-X}
{
    Suppose \Cref{assump: CI,assump: ext-valid2,assump: unconfound-exp2,assump: unconfound-obs,assump: completeness,assump: bridge,assump: bridge2} hold. Assume that  for $k = 1, \dots, K$, the nuisance estimator $\hat\eta_k = (\hat h_k, \hat{\bar{h}}_{k}, \hat q_k, \hat\alpha_k, \hat\beta_k)$ is a  consistent estimator for $\eta^*$, and satisfies the following convergence rate conditions: 
    \begin{align*}
\|\hat\alpha_k - \alpha_0\|_{\mathcal{L}_2(\mathbb{P})}\|\hat{\bar{h}}_k - \bar{h}_{E, 0}\|_{\mathcal{L}_2(\mathbb{P})} &= o_{\mathbb{P}}(n^{-1/2}),  \\
({\|\hat\alpha_k - \alpha_0\|_{\mathcal{L}_2(\mathbb{P})} + \|\hat\beta_k - \beta_0\|_{\mathcal{L}_2(\mathbb{P})}})\|T(\hat h_k - h_0)\|_{\mathcal{L}_2(\mathbb{P})} &= o_{\mathbb{P}}(n^{-1/2}), \\
\min\left\{\|T(\hat h_k - h_0)\|_{\mathcal{L}_2(\mathbb{P})}\|\hat q_k - q_0\|_{\mathcal{L}_2(\mathbb{P})},\|\hat h_k - h_0\|_{\mathcal{L}_2(\mathbb{P})} \|T^\star(\hat q_k - q_0)\|_{\mathcal{L}_2(\mathbb{P})}\right\} &= o_{\mathbb{P}}(n^{-1/2}).
\end{align*}
}
Then 
\begin{align*}
\sqrt{n}\prns{\hat\tau - \tau} \rightsquigarrow  \mathcal{N}\prns{0, \sigma^2},
\end{align*}
where 
\begin{align*}
\sigma^2 
   &= \prns{1 + \lambda}\Eb{\prns{\phi_1\prns{Y, S, 1, X; \eta^*} - \phi_1\prns{Y, S, 0, X; \eta^*} - \tau}^2\mid G=O} \\
   &+ \frac{1 + \lambda}{\lambda}\Eb{\prns{\phi_2\prns{Y, S, 1, X; \eta^*} - \phi_2\prns{Y, S, 0, X; \eta^*}}^2 \mid G=E} \\
   &+ \prns{1 + \lambda}\Eb{\prns{\phi_3\prns{Y, S, 1, X; \eta^*} - \phi_3\prns{Y, S, 0, X; \eta^*}}^2\mid G=O}.
\end{align*}
\end{theorem}

\section{Additional Extensions}\label{sec: more-extension}
In this section, we extend our identification results to more settings. For simplicity, we focus on extending the first identification strategy in \Cref{thm: identification1}.

\subsection{Pre-treatment Outcomes}\label{sec: pretreat}
In the main text, the short-term outcomes $S = \prns{S_1, S_2, S_3}$ are all post-treatment outcomes. 
In this part, we let part of the outcomes be pre-treatment. 

We first consider the setting where $S_1$ is pre-treatment but $S_2, S_3$ are post-treatment. 
Below we modify \Cref{assump: unconfound-obs,assump: unconfound-exp,assump: CI,assump: ext-valid} accordingly. 

\begin{assumption}[Pre-treatment $S_1$]\label{assump: S1-pretreat}
Suppose the following hold for $a \in \braces{0, 1}$: 
\begin{enumerate}
\item On the observational data,  we have $\prns{Y\prns{a}, S_3\prns{a}, S_2\prns{a}} \perp A \mid S_1, U, X, G = O$ and $0 < \Prb{A = 1 \mid S_1, U, X, G = O} < 1$ almost surely. 
\item On the experimental data, we have $\prns{Y\prns{a}, S_3\prns{a}, S_2\prns{a}, U} \perp A \mid S_1, X, G = E$ and $0 < \Prb{A = 1 \mid S_1, X, G = O} < 1$ almost surely. 
\item The external validity $\prns{S_3\prns{a}, S_2\prns{a}, U} \perp G \mid S_1, X$ and overlap 
$$\frac{p\prns{S_1, U, X \mid A = a, G = E}}{p\prns{S_1, U, X\mid A= a, G = O}} < \infty, ~~ \text{almost surely.}$$
\item The sequential structure  $\prns{Y(a), S_3(a)} \perp S_1 
    \mid S_2(a), U, X, G = O$.
\end{enumerate}
\end{assumption}
Note that \Cref{assump: S1-pretreat} consider the most general setting: we allow the treatment assignments in the observational and experimental data to depend on pre-treatment outcomes $S_1$, and also allow the distribution of $S_1$ to be different on the two datasets. Now we extend our identification strategy
to this setting. 
\begin{corollary}\label{corollary: S1-pretreat}
Suppose conditions in \Cref{assump: S1-pretreat}, the completeness condition in \Cref{assump: completeness}  condition \ref{assump: completeness-2} and \Cref{assump: bridge} hold. Then the average long-term treatment effect is identifiable: for any function $h_0$ that satisfies \Cref{eq: bridge-obs}, 
\begin{align}\label{eq: identification-1-S1}
\tau 
    &= \Eb{\Eb{h_0\prns{\Sc, \Sb, A, X} \mid S_1, A = 1, X, G = E} \mid G = O} \nonumber \\
    &\qquad\qquad - \Eb{\Eb{h_0\prns{\Sc, \Sb, A, X} \mid S_1, A = 0, X, G = E} \mid G = O}.
\end{align}
\end{corollary}
The identification strategy in \Cref{eq: identification-1-S1} is very similar to \Cref{eq: identification-1-X}. 
\Cref{eq: identification-1-S1} essentially augments the covariates $X$ with the pre-treatment outcomes $S_1$. 

Similarly, we can also consider the setting where both $S_1$ and  $S_2$ are pre-treatment. 
\begin{assumption}[Pre-treatment $(S_1, S_2)$]\label{assump: S1-S2-pretreat}
Suppose the following hold for $a \in \braces{0, 1}$: 
\begin{enumerate}
\item On the observational data,  we have $\prns{Y\prns{a}, S_3\prns{a}} \perp A \mid S_2, S_1, U, X, G = O$ and $0 < \Prb{A = 1 \mid S_2, S_1, U, X, G = O} < 1$ almost surely. 
\item On the experimental data, we have $\prns{Y\prns{a}, S_3\prns{a}, U} \perp A \mid S_2, S_1, X, G = E$ and $0 < \Prb{A = 1 \mid S_2, S_1, X, G = O} < 1$ almost surely. 
\item The external validity $\prns{S_3\prns{a}, U} \perp G \mid S_2, S_1, X$ and overlap 
$$\frac{p\prns{S_2, S_1, U, X \mid A = a, G = E}}{p\prns{S_2, S_1, U, X\mid A= a, G = O}} < \infty, ~~ \text{almost surely.}$$
\item The sequential structure  $\prns{Y(a), S_3(a)} \perp S_1 
    \mid S_2, U, X, G = O$.
\end{enumerate}
\end{assumption}
We can analogously identify the long-term average treatment effect 
\begin{corollary}\label{corollary: S1-S2-pretreat}
Suppose conditions in \Cref{assump: S1-S2-pretreat}, the completeness condition in \Cref{assump: completeness}  condition \ref{assump: completeness-2} and \Cref{assump: bridge} hold. Then the average long-term treatment effect is identifiable: for any function $h_0$ that satisfies \Cref{eq: bridge-obs}, 
\begin{align}\label{eq: identification-1-S1-S2}
\tau 
    &= \Eb{\Eb{h_0\prns{\Sc, \Sb, A, X} \mid S_2, S_1, A = 1, X, G = E} \mid G = O} \nonumber \\
    &\qquad\qquad - \Eb{\Eb{h_0\prns{\Sc, \Sb, A, X} \mid S_2, S_1, A = 0, X, G = E} \mid G = O}.
\end{align}
\end{corollary}

\subsection{Partial Confounding Adjustments}\label{sec: partial}
\begin{figure}[t]
\centering 
\begin{subfigure}[b]{0.48\textwidth}
\begin{tikzpicture}
\node[draw, circle, text centered] (A) {$A$};
\node[draw, dashed, circle, text centered, right=5cm of A] (Y) {$Y$};
\node[draw, circle, text centered, right=0.6cm of A] (S) {$S_1$};
\node[draw, circle, text centered, right=2cm of A] (S2) {$S_2$};
\node[draw, circle, text centered, right=3.4cm of A] (S3) {$S_3$};
\node[draw, circle, dashed, text centered, above=1cm of S2] (U) {$U_\Diamond$};
\node[draw, circle, dashed, color= blue, text centered, left=0.7cm of U] (Ul) {$\color{blue}U_\dag$};
\node[draw, circle, dashed, color=magenta, text centered, right=0.7cm of U] (Ur) {$\color{magenta} U_\sharp$};

\draw[->] (A) -- (S);
\draw[->] (A) to [bend right] (S2);
\draw[->] (A) to [bend right] (S3);
\draw[->] (A) to [bend right] (Y);
\draw[->] (U) -- (S);
\draw[->] (U) -- (S2);
\draw[->] (U) -- (S3);
\draw[->] (S) -- (S2);
\draw[->] (S2) -- (S3);
\draw[->] (S3) -- (Y);
\draw[->] (U) -- (Y);
\draw[-] (U) -- (Ul);
\draw[-] (U) -- (Ur);
\draw[->,color=blue] (Ul) -- (S);
\draw[->,color=blue] (Ul) -- (S2);
\draw[->,color = magenta] (Ur) -- (S3);
\draw[->,color = magenta] (Ur) -- (Y);
\draw[->] (U) -- (A);
\draw[->,color=blue] (Ul) -- (A);
\end{tikzpicture}
\caption{Observational data.}
\label{figure: DAG-obs-c}
\end{subfigure}
\begin{subfigure}[b]{0.48\textwidth}
\begin{tikzpicture}
\node[draw, circle, text centered] (A) {$A$};
\node[draw, dashed, circle, text centered, right=5cm of A] (Y) {$Y$};
\node[draw, circle, text centered, right=0.6cm of A] (S) {$S_1$};
\node[draw, circle, text centered, right=2cm of A] (S2) {$S_2$};
\node[draw, circle, text centered, right=3.4cm of A] (S3) {$S_3$};
\node[draw, circle, dashed, text centered, above=1cm of S2] (U) {$U_\Diamond$};
\node[draw, circle, dashed, color= blue, text centered, left=0.7cm of U] (Ul) {$\color{blue}U_\dag$};
\node[draw, circle, dashed, color=magenta, text centered, right=0.7cm of U] (Ur) {$\color{magenta} U_\sharp$};

\draw[->] (A) -- (S);
\draw[->] (A) to [bend right] (S2);
\draw[->] (A) to [bend right] (S3);
\draw[->] (A) to [bend right] (Y);
\draw[->] (U) -- (S);
\draw[->] (U) -- (S2);
\draw[->] (U) -- (S3);
\draw[->] (S) -- (S2);
\draw[->] (S2) -- (S3);
\draw[->] (S3) -- (Y);
\draw[->] (U) -- (Y);
\draw[-] (U) -- (Ul);
\draw[-] (U) -- (Ur);
\draw[->,color=blue] (Ul) -- (S);
\draw[->,color=blue] (Ul) -- (S2);
\draw[->,color = magenta] (Ur) -- (S3);
\draw[->,color = magenta] (Ur) -- (Y);
\end{tikzpicture}
\caption{Experiment sample.}
\label{figure: DAG-exp-c}
\end{subfigure}
\caption{unobserved confounders $\prns{{\color{magenta}U_\sharp}, {\color{blue}U_\dag}}$ that can be ignored. Additional covariates $X$ can be present but we do not draw them to avoid cluttering the graphs. }
\label{figure: DAG-c}
\end{figure}
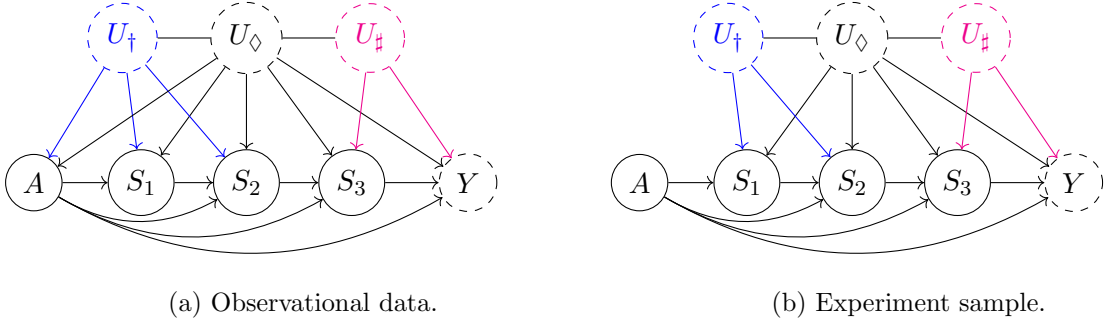

In the main text, the unobseved variables  $U$ stand for all unobserved confounders that can possibly affect the treatment, the short-term outcomes, the long-term outcome, or any subset of them (see \Cref{figure: DAG-b}). The identification strategies in \Cref{sec: identification} require the short-term outcomes $(S_1, S_3)$ to be sufficiently rich relative to all of the unobserved confounders. 
In this part, we show that actually we do not need to use the short-term outcomes to handle all such unobserved confounders. 
Instead, we can achieve identification under lower requirements for the short-term outcomes, still using the same identification strategies.

In \Cref{figure: DAG-c}, we plot three different types of unobserved confounders: confounders $U_{\diamond}$ can affect any of $(Y, S_3, S_2, S_1, A)$, confounders $U_{\dag}$ can affect $(S_2, S_1, A)$ but not $(S_3, Y)$, while confounders $U_{\sharp}$ can affect $(S_3, Y)$ but not $(S_2, S_1, A)$. 
Naively, one can view $U = \prns{U_{\diamond},U_{\dag}, U_{\sharp}}$ and argue identifiability following any of \Cref{thm: identification1,thm: identification2,thm: identification-DR}. 
This would require the short-term outcomes $(S_1, S_3)$ to be rich enough relative to all of 
$\prns{U_{\diamond},U_{\dag}, U_{\sharp}}$. 
Now we show that this is not necessary. Instead, we  need  $(S_1, S_3)$ to be rich enough relative to \emph{only} $U_{\diamond}$, but \emph{not} necessarily $\prns{U_{\dag}, U_{\sharp}}$. 

We first extend \Cref{assump: CI,assump: unconfound-obs,assump: unconfound-exp,assump: ext-valid} to the current setting, by substituting $U_{\diamond}$ for $U$ in these previous assumptions. 
\begin{assumption}\label{assump: partial}
Assume the following conditions hold for any $a \in \braces{0, 1}$ :
\begin{enumerate}
\item $\prns{Y(a), S_3\prns{a}} \perp A \mid S_2(a), U_{\diamond}, X, G = O$ and $0 < \Prb{A = 1 \mid U_{\diamond}, X, G = O} < 1$ almost surely. 
\item $\prns{S_2\prns{a}, U_{\diamond}} \perp A \mid X, G = E$ and $0 < \Prb{A = 1 \mid X, G = E} < 1$ almost surely. 
\item $\prns{S_3\prns{a}, S_2\prns{a}, U_{\diamond}} \perp G \mid X$, and 
\begin{align*}
\frac{p\prns{U_{\diamond}, X \mid A = a, G = E}}{p\prns{U_{\diamond}, X \mid A= a, G = O}} < \infty.
\end{align*}
\item $\prns{Y(a), S_3\prns{a}} \perp S_1\prns{a} \mid S_2\prns{a}, U_{\diamond}, X, G = O$.
\end{enumerate}
\end{assumption}

It is easy to verify that the current setting depicted in \Cref{figure: DAG-c} can satisfy \Cref{assump: partial}.
Moreover, below we modify the completeness condition in \Cref{assump: completeness}  condition \ref{assump: completeness-2} and the outcome bridge function assumption in \Cref{assump: bridge}.
\begin{assumption}\label{assump: partial-2}
\begin{enumerate}
\item For any $s_2 \in \Scal_2$, $a \in \braces{0, 1}$, $x \in \Xcal$, 
$$\text{if } \Eb{g\prns{U_{\diamond}} \mid \Sa, \Sb = s_2, A = a, X = x, G = O} = 0 ~ \text{holds almost surely},$$ 
then $g\prns{U_{\diamond}} = 0$ almost surely.
\item There exists an outcome bridge function $h_0: \Scal_3 \times \Scal_2 \times \Acal \times \Xcal \to \Rl$ such that 
\begin{align}\label{eq: bridge-U-partial}
\Eb{Y \mid \Sb, A, U_{\diamond}, X, G = O} = \Eb{h_0\prns{\Sc, \Sb, A, X} \mid \Sb, A, U_{\diamond}, X, G = O}.
\end{align}
\end{enumerate}
\end{assumption}

In \Cref{assump: partial-2}(a), we assume a partial completeness condition, which only require the short-term outcomes $S_1$ to be rich enough relative to $U_{\diamond}$. In \Cref{assump: partial-2}(b), we only require the bridge function to capture the unmeasured confounding due to $U_{\diamond}$. This is possible when the short-term outcomes $S_3$ are rich enough relative to $U_{\diamond}$. 
Importantly, we do not need $S_1, S_3$ to be rich enough relative to $\prns{U_{\diamond},U_{\dag}, U_{\sharp}}$ together. 

Then we show that the long-term average treatment effect can be identified according to the equation we derived in \Cref{corollary: covariate-exp}. 
This means actually the same identification strategy still works under lower requirements on the short-term outcomes. 
\begin{corollary}\label{corollary: partial}
Suppose \Cref{assump: partial,assump: partial-2} hold. Then the average long-term treatment effect is identifiable: for any function $h_0$ that satisfies \Cref{eq: bridge-obs}, \Cref{eq: identification-1-X} in \Cref{corollary: covariate-exp} holds.  
\end{corollary}

\subsection{Relaxing the External Validity Assumption}\label{sec: ext-validity}

{In \Cref{sec: extension} \Cref{assump: ext-valid2}, we assumed the external validity condition that the distributions of the unobserved confounders $U$ on the two datasets, conditional on the covariates $X$, are identical. 
In this section, we show that this assumption can be weakened, provided that $S_1$ and $S_2$ are both pre-treatment outcomes.
Specifically, we assume the following condition.
}

{
\begin{assumption}\label{assump: latent-conditional-ext-valid}
Suppose that for any $a \in \braces{0, 1}$,
    \begin{align*}
        G \perp S_3\prns{a} \mid S_2, S_1, U, X.
    \end{align*}
\end{assumption}
}

{
\Cref{assump: latent-conditional-ext-valid} imposes that the distributions of the potential short-term outcome $S_3\prns{a}$ are identical on the two datasets, conditional on the pre-treatment outcomes $S_2, S_1$, the unobserved confounders $U$, and the observed covariates $X$. 
Importantly, this assumption is weaker than the condition $G \perp (S_3(a), U) \mid S_2, S_1, X$, 
allowing for distribution shift of the unobserved confounders $U$. 
To handle the lack of external validity, we again view the short-term outcomes as proxies for the unobserved confounders. 
Specifically, we rely on the following external validity bridge function. 
}

{
\begin{assumption}[External validity bridge function]\label{assump: bridge3}
There exists an external validity bridge function $\tilde{q}: \Scal_2 \times \Scal_1 \times \Acal \times \Xcal \to \R{}$ defined as follows: 
\begin{align}\label{eq: ev-bridge}
\frac{p(S_2, U, X \mid A, G = O)}{p(S_2, U, X \mid A, G = E)} = \E[\tilde{q}\prns{S_2, S_1, A, X} \mid S_2, A, U, X, G = E]
\end{align}
\end{assumption}
The external validity bridge function in \Cref{assump: bridge3} \Cref{eq: ev-bridge} is very similar to the selection bridge function in \Cref{assump: bridge2} \Cref{eq: bridge2-U}. There are only two differences: one is that the left hand side of \Cref{eq: ev-bridge} is the reciprocal of the left hand side of \Cref{eq: bridge2-U}, and the other is that the right hand side of \Cref{eq: ev-bridge} involves a conditional expectation over the experimental data rather than the observational data. 
We will show that the external validity bridge function can adjust for the discrepancy in the distributions of unobserved confounders between the two datasets. 
Since the external validity bridge function in \Cref{eq: ev-bridge} is defined in terms of unobserved confounders, we cannot directly use this definition to learn an external validity bridge function. 
Instead, we give an alternative formulation that involves only the observed variables.
} 

{
\begin{lemma}\label{lemma: ev-bridge-obs}
Assume \Cref{assump: S1-S2-pretreat} conditions 1, 2, 4, \Cref{assump: latent-conditional-ext-valid}, and the completeness condition in \Cref{assump: completeness}  condition \ref{assump: completeness-2}. Then  
any function $\tilde{q}$ that satisfies 
\[
\frac{p(S_3, S_2, X \mid A, G = O)}{p(S_3, S_2, X \mid A, G = E)} = \E[\tilde{q}\prns{S_1, S_2, X, A} \mid S_2, X, S_3, A, G = E]
\]
is also a valid external validity bridge function in the sense of \Cref{eq: ev-bridge}. 
\end{lemma}
}

{
In the theorem below, we further show that the the average treatment effect can be identified by any external validity function and any outcome bridge function.  
}

{
\begin{theorem}\label{thm:ext-valid}
Assume the assumptions in \Cref{lemma: ev-bridge-obs}, \Cref{assump: bridge}, and $(S_2, S_1, U, X) \perp A \mid G = E$ hold. Let $\tilde q$ and $h$ be any functions that satisfy \Cref{eq: ev-bridge} and \Cref{eq: bridge-obs} respectively. Then for $a \in \{0, 1\}$, we have 
    \begin{align*}
        \E[Y(a) \mid S_2, X, G = O] =  m(S_2, a, X) \frac{p(S_2, X \mid G = E)}{p(S_2, X \mid G = O)},
    \end{align*}
where
\begin{align*}
     m(S_2, a, X) :=  &\E \bigg[ \E[h(S_3, S_2, X, A) \mid S_2, S_1,  X, G = E, A = a] \\
        &\quad\quad\quad\quad \sum_{a'}\Prb{A = a' \mid G = O}  \tilde{q}\prns{S_2, S_1, X, a'} \mid S_2, X, A = a, G = E\bigg].
\end{align*}
Moreover, we have 
    \begin{align*}
        \tau  = \E\left[ m(S_2, 1, X) - m(S_2, 0, X) \mid G = E\right].
    \end{align*}
\end{theorem}
}

\section{Additional Results for Numerical Studies}\label{sec: appendixsimulation}
\subsection{Additional Details for \Cref{sec: real-data}}
In the following proposition, we justify the sampling probability function described in \Cref{sec: data}.
\begin{proposition}\label{prop: experiment}
Let $(Z_1, Z_2, A)$ be a random vector with $(Z_1, Z_2) \perp A$ and $A \in \{0,1\}$. Let $G \in \{0,1\}$ be a binary random variable such that $G \perp Z_2 \mid Z_1$ and 
\[
\Prb{G = 1\mid Z_1, A = 1} \Prb{A = 1} + \Prb{G = 1 \mid Z_1, A = 0} \Prb{A = 0} \equiv C,
\]
where $C$ is a positive constant. 
Then the probability density of $(Z_1, Z_2)$ satisfies that
\[
p(z_1, z_2 \mid G = 1) \equiv p(z_1, z_2), ~~ \forall z_1, z_2.
\]
\end{proposition}
We can let $Z_1$ be the education level $U$, $Z_2$ be other covariates and the potential short-term outcomes, $A$ be the GAIN treatment assignment, and $G$ be the indicator for whether being selected into the observational dataset $\mathcal{D}_O$.
Then \Cref{prop: experiment} means that
the subsampling procedure does not change the distribution of latent confounders, covariates, and potential short-term outcomes. This explains why the subsampling is not against \Cref{assump: ext-valid}.

\subsection{Additional Results for \Cref{sec: real-data}}

\begin{table}[!t]
    \centering
    \begin{tabular}{c|c|c|c|c|c}
    \hline
      {\scriptsize $\dim(S_1)$} & $s_2$  & {\scriptsize $P (\mathbf{S}_1 \mid S_2 = s_2, A = 0, \mathbf{U})$} & {\scriptsize $P (\mathbf{S}_3 \mid S_2 = s_2, A = 0, \mathbf{U})$} & {\scriptsize $P (\mathbf{S}_1 \mid S_2 = s_2, A = 1, \mathbf{U})$} & {\scriptsize $P (\mathbf{S}_3 \mid S_2 = s_2, A = 1, \mathbf{U})$}\\ \hline
      \multirow{4}{*}{2} & $(0,0)$ &  $0.016$ & $0.012$ & 0.002 & 0.007\\ 
       & $(1,0)$ &  $0.015$ & $0.068$  & 0.059 & 0.018\\
       & $(0,1)$ &  $0.014$ & $0.019$ & 0.002 & 0.024\\
       & $(1,1)$ &  $0.026$ & $0.038$ & 0.010 &  0.020 \\\hline
       \multirow{4}{*}{4} 
       & $(0,0)$ &  $0.023$ & 0.023 & 0.018 & 0.009\\ 
       & $(1,0)$ &  $0.190$ & 0.125  & 0.072 & 0.055\\
       & $(0,1)$ &  $0.123$ & 0.185 & 0.067 & 0.032\\
       & $(1,1)$ &  $0.063$ & 0.084 & 0.036 &  0.028 \\\hline
       \multirow{4}{*}{6} 
       & $(0,0)$ &  0.028 & 0.020 & 0.020 & 0.014\\ 
       & $(1,0)$ & 0.134 & 0.145  & 0.068 & 0.053\\
       & $(0,1)$ &  0.151 & 0.165 & 0.070 & 0.083\\
       & $(1,1)$ &  0.084 & 0.072 & 0.040 &  0.037 \\\hline
    \end{tabular}
    \caption{List of smallest singular values of the empirical estimates of the conditional probability matrices $P (\mathbf{S}_1 \mid S_2 = s_2, A = a, \mathbf{U})$ and $P (\mathbf{S}_3 \mid S_2 = s_2, A = a, \mathbf{U})$ for $s_2 = (0, 0), \ldots, (1, 1), a = 0, 1$ with different dimension of $S_1$ and $S_3$. Here we keep throughout the dimension of $S_1$ and $S_3$ to be the same. The dimension of $S_1$ corresponds to the number of quarters included in the surrogate. Here $\dim(S_1) = 2$ means that $S_1, S_2$ and $S_3$ take the employment status of $1-2$-th quarters, $3-4$-th quarters and $5-6$-th quarters after the treatment respectively; $\dim(S_1) = 4$ means the three surrogates take $1-4$-th quarters, $5-6$-th quarters and $7-10$-th quarters respectively; and $\dim(S_1) = 6$ means the three surrogates take $1-6$-th quarters, $7-8$-th quarters and $9-14$-th quarters respectively. %
    }
    \label{table: singular}
\end{table}

{The results in \Cref{table: numerical} are shown in relative scale. \Cref{table: raw} further show the corresponding raw numbers. Moreover, we provide additional results on the RMSE and Bias of different methods across the replications of sampling. We find that our proposed estimators, when regularized appropriately, can achieve small bias and small variance. In contrast, the benchmark estimators that cannot handle general persistent confounding problem tend to have much higher bias. As a result, our estimators outperform the existing benchmarks by a large margin.}

\begin{sidewaystable}[!t]
\footnotesize
\centering 
 \begin{tabular}{cc|cccc|cccc|cccc|cc|c}
    \hline
    & & \multicolumn{4}{c}{$\hat\tau_{\out}$} & \multicolumn{4}{c}{$\hat\tau_{\sel}$} & \multicolumn{4}{c}{$\hat\tau_{\dr}$} &  \multicolumn{2}{c}{\scriptsize Athey et al.} & {\scriptsize Naive}  \\
    \hline
    $\eta$ & & 0 & .33 & .67 & 1 &  0 & .33 & .67 & 1 & 0 & .33 & .67 & 1 & NR & CV &  \\
         \hline
    \multirow{2}{*}{0} & MAE & 0.017 & 0.006 & 0.008 & 0.009 & 0.010 & 0.010 & 0.011 & 0.011 & 0.015 & 0.002 & 0.005 & 0.006 & 0.047 & 0.044 & 0.053 \\
& Med & 0.017 & 0.006 & 0.008 & 0.009 & 0.010 & 0.010 & 0.011 & 0.011 & 0.015 & 0.002 & 0.005 & 0.006 & 0.047 & 0.044 & 0.053\\
& RMSE & 0.017 & 0.006 & 0.008 & 0.009 & 0.010 & 0.010 & 0.011 & 0.011 & 0.015 & 0.002 & 0.005 & 0.006 & 0.047 & 0.044 & 0.053\\
& Bias & 0.017 & 0.006 & 0.008 & 0.009 & 0.010 & 0.010 & 0.011 & 0.011 & 0.015 & 0.002 & 0.005 & 0.006 & 0.047 & 0.044 & 0.053 \\
\hline
\multirow{2}{*}{0.2} & MAE & 0.048 & 0.010 & 0.012 & 0.013 & 0.012 & 0.013 & 0.013 & 0.013 & 0.045 & 0.007 & 0.008 & 0.009 & 0.048 & 0.050 & 0.059 \\
& Med & 0.023 & 0.010 & 0.012 & 0.013 & 0.012 & 0.013 & 0.013 & 0.013 & 0.021 & 0.006 & 0.008 & 0.009 & 0.048 & 0.050 & 0.059\\
& RMSE & 0.385 & 0.011 & 0.013 & 0.014 & 0.014 & 0.014 & 0.014 & 0.014 & 0.370 & 0.008 & 0.010 & 0.011 & 0.049 & 0.050 & 0.059\\
& Bias & 0.022 & 0.010 & 0.012 & 0.013 & 0.012 & 0.012 & 0.013 & 0.013 & 0.020 & 0.006 & 0.008 & 0.009 & 0.048 & 0.050 & 0.059\\
\hline
\multirow{2}{*}{0.4} & MAE & 0.055 & 0.014 & 0.016 & 0.017 & 0.016 & 0.016 & 0.016 & 0.016 & 0.051 & 0.011 & 0.012 & 0.013 & 0.050 & 0.058 & 0.067 \\
& Med & 0.025 & 0.014 & 0.016 & 0.017 & 0.016 & 0.016 & 0.016 & 0.016 & 0.023 & 0.010 & 0.011 & 0.012 & 0.050 & 0.058 & 0.067\\
& RMSE & 0.252 & 0.016 & 0.018 & 0.019 & 0.018 & 0.018 & 0.018 & 0.018 & 0.234 & 0.013 & 0.014 & 0.015 & 0.051 & 0.058 & 0.067\\
& Bias & 0.027 & 0.014 & 0.016 & 0.017 & 0.016 & 0.016 & 0.016 & 0.016 & 0.024 & 0.010 & 0.012 & 0.012 & 0.050 & 0.058 & 0.067\\
\hline
\multirow{2}{*}{0.6} & MAE & 0.069 & 0.021 & 0.023 & 0.024 & 0.021 & 0.021 & 0.021 & 0.021 & 0.063 & 0.016 & 0.018 & 0.018 & 0.053 & 0.068 & 0.076 \\
& Med & 0.030 & 0.020 & 0.022 & 0.024 & 0.021 & 0.021 & 0.021 & 0.021 & 0.028 & 0.015 & 0.017 & 0.018 & 0.053 & 0.068 & 0.076\\
& RMSE & 0.295 & 0.023 & 0.025 & 0.026 & 0.024 & 0.024 & 0.024 & 0.024 & 0.273 & 0.019 & 0.020 & 0.021 & 0.055 & 0.069 & 0.077\\
& Bias & 0.022 & 0.020 & 0.023 & 0.024 & 0.021 & 0.021 & 0.021 & 0.021 & 0.019 & 0.015 & 0.017 & 0.018 & 0.053 & 0.068 & 0.076\\
\hline
\multirow{2}{*}{0.8} & MAE & 0.111 & 0.030 & 0.032 & 0.033 & 0.029 & 0.029 & 0.029 & 0.029 & 0.098 & 0.024 & 0.026 & 0.026 & 0.059 & 0.082 & 0.088 \\
& Med & 0.038 & 0.030 & 0.032 & 0.034 & 0.029 & 0.029 & 0.029 & 0.030 & 0.036 & 0.024 & 0.025 & 0.026 & 0.060 & 0.082 & 0.088\\
& RMSE & 0.960 & 0.033 & 0.035 & 0.036 & 0.032 & 0.032 & 0.032 & 0.033 & 0.814 & 0.028 & 0.029 & 0.030 & 0.062 & 0.082 & 0.089\\
& Bias & 0.016 & 0.030 & 0.032 & 0.033 & 0.029 & 0.029 & 0.029 & 0.029 & 0.000 & 0.024 & 0.025 & 0.026 & 0.059 & 0.082 & 0.088\\
\hline
\multirow{2}{*}{1} & MAE & 0.072 & 0.033 & 0.035 & 0.036 & 0.030 & 0.031 & 0.031 & 0.031 & 0.065 & 0.026 & 0.028 & 0.028 & 0.061 & 0.089 & 0.095\\
& Med & 0.041 & 0.033 & 0.035 & 0.036 & 0.030 & 0.030 & 0.030 & 0.031 & 0.038 & 0.026 & 0.027 & 0.028 & 0.061 & 0.090 & 0.095\\
& RMSE & 0.214 & 0.037 & 0.038 & 0.040 & 0.035 & 0.035 & 0.035 & 0.035 & 0.179 & 0.031 & 0.032 & 0.033 & 0.065 & 0.090 & 0.095\\
& Bias & 0.036 & 0.033 & 0.035 & 0.036 & 0.030 & 0.030 & 0.030 & 0.031 & 0.036 & 0.025 & 0.027 & 0.027 & 0.061 & 0.089 & 0.095\\
\hline
\multirow{2}{*}{1.2} & MAE & 0.380 & 0.037 & 0.039 & 0.041 & 0.033 & 0.033 & 0.034 & 0.034 & 0.438 & 0.029 & 0.031 & 0.031 & 0.065 & 0.098 & 0.104\\
& Med & 0.045 & 0.036 & 0.039 & 0.041 & 0.032 & 0.032 & 0.033 & 0.034 & 0.043 & 0.027 & 0.029 & 0.030 & 0.064 & 0.099 & 0.104\\
& RMSE & 8.827 & 0.042 & 0.044 & 0.045 & 0.038 & 0.039 & 0.039 & 0.039 & 10.610 & 0.035 & 0.036 & 0.037 & 0.070 & 0.099 & 0.104\\
& Bias & 0.319 & 0.037 & 0.039 & 0.041 & 0.031 & 0.032 & 0.032 & 0.033 & 0.378 & 0.027 & 0.028 & 0.029 & 0.065 & 0.098 & 0.104\\
\hline 
\multirow{2}{*}{1.4} & MAE & 0.130 & 0.044 & 0.047 & 0.049 & 0.036 & 0.037 & 0.038 & 0.038 & 0.129 & 0.033 & 0.035 & 0.036 & 0.072 & 0.111 & 0.115\\
& Med & 0.056 & 0.043 & 0.046 & 0.048 & 0.033 & 0.033 & 0.034 & 0.035 & 0.054 & 0.029 & 0.030 & 0.031 & 0.072 & 0.111 & 0.115\\
& RMSE & 0.773 & 0.050 & 0.052 & 0.054 & 0.044 & 0.044 & 0.045 & 0.045 & 0.677 & 0.041 & 0.042 & 0.043 & 0.078 & 0.112 & 0.116 \\
& Bias & 0.022 & 0.043 & 0.047 & 0.049 & 0.032 & 0.034 & 0.035 & 0.036 & 0.019 & 0.029 & 0.031 & 0.032 & 0.072 & 0.111 & 0.115\\
\hline
\multirow{2}{*}{1.6} & MAE & 0.117 & 0.049 & 0.053 & 0.055 & 0.040 & 0.040 & 0.041 & 0.042 & 0.112 & 0.037 & 0.038 & 0.039 & 0.075 & 0.120 & 0.124\\
& Med & 0.064 & 0.048 & 0.052 & 0.055 & 0.035 & 0.036 & 0.037 & 0.040 & 0.060 & 0.032 & 0.034 & 0.035 & 0.075 & 0.120 & 0.124\\
& RMSE & 0.457 & 0.055 & 0.059 & 0.061 & 0.049 & 0.049 & 0.049 & 0.050 & 0.349 & 0.045 & 0.046 & 0.048 & 0.082 & 0.120 & 0.125\\
& Bias & 0.061 & 0.048 & 0.052 & 0.055 & 0.034 & 0.036 & 0.037 & 0.039 & 0.034 & 0.030 & 0.032 & 0.034 & 0.074 & 0.120 & 0.124\\
\hline
\end{tabular} 
\caption{{Same setting as Table 1, but MAE and Med are shown in raw numbers. Additional RMSE and Bias results are shown.}
}
\label{table: raw}
\end{sidewaystable}

{Moreover, we try using data to probe the plausibility of the \Cref{assump: bridge,assump: bridge2}  in our GAIN case study. According to Example 1, in a discrete setting, \Cref{assump: bridge,assump: bridge2} hold when certain conditional probability matrices have full column rank. We note that the outcomes in the GAIN dataset empirical example are all discrete, so we design some heurstic assessments here to shed some light on Assumptions 6 and 7 in the  empirical study. 
    Based on the GAIN dataset, we estimate the conditional probability matrices $P (\mathbf{S_1} \mid S_2 = s_2, A = a,\mathbf{U})$ and $P (\mathbf{S_3} \mid S_2 = s_2, A = a, \mathbf{U})$ by their empirical frequencies (we do not  condition on  $X$ since this is difficult noting that $X$ is multi-dimensional and some components are continuous), for $s_2 \in \{(0, 0), (1, 0), (0, 1), (1,1)\}$ and $a \in \{0, 1\}$. We vary the dimension of $S_1$ and $S_3$ (\ie, the number of employment status variables included in $S_1, S_3$ respectively) from $2$ to $6$ while fixing the dimension of $S_2$ as $2$ (the number we used in our original numerical study). The smallest singular values of the corresponding empirical probability matrices are calculated and shown in \Cref{table: singular}. 
    We can observe that the smallest singular value gets consistently larger as dimension of $S_1$ and $S_3$ increases, unless when the smallest singular value is already sufficiently large. This heuristically suggests that Assumptions 6 and 7 are more likely to hold if we incorporate more short-term outcomes in $S_1, S_3$ relative to $S_2$, validating our high level intuitions discussed above.}
    {Moreover, we also hope to validate that our assumptions are less plausible as the dimension of $S_2$ grows relative to $S_1, S_3$. We thus increase the dimension of $S_2$ from $2$ to $4$ while fixing the dimension of $S_1, S_3$ as $2$. The resulting  minimum singular value estimates are shown in \Cref{table: singular3,table: singular4}. We observe that as the dimension of $S_2$ grows, the minimum singular values of the corresponding empirical conditional probability matrices tend to be zero or close to zero, indicating violations of the completeness conditions. This means that our assumptions may be less plausible when the dimension of $S_2$ becomes larger relative to the dimension of $S_1,S_3$, thus validating our interpretations from an opposite perspective.
    }

    {Of course, {in \Cref{table: singular,table: singular3,table: singular4}}, some smallest singular values are indeed fairly small, posing threats to Assumptions 6 and 7.
    However, we find that  across all these settings, the performance of our proposed estimator is overall stable and it is significantly better than the existing state-of-art estimator in \citet{athey2020combining}.
    These results show potential benefit of using our method to account for general unobserved confounding, even if our assumptions may not necessarily hold exactly. {This is perhaps because our identification formula involves averaging  over the values of $S_2$, so that even if the completeness condition is violated at certain values of $S_2$, this partial violation, while incurring some bias, may have limited impact on the final averaging result.}} 
        
    {Specifically, we already show the performance of our estimator in \Cref{sec: real-data} \Cref{table: numerical} for the setting $\op{dim}(S_1) = \op{dim}(S_2) = \op{dim}(S_3) = 2$. 
    In \Cref{table: numericalq4}, we generate the data in the same way as in \Cref{table: numerical}, but with $(S_1, S_2, S_3)$ as the employment status in the $1-4$-th quarters, $5-6$-th quarters, and $7-10$-th quarters. In other words, we keep $\op{dim}(S_2) = 2$ but increase the dimension of both $S_1$ and $S_3$ to $\op{dim}(S_2) = \op{dim}(S_3) = 4$. In \Cref{table: numericalq6}, we set $(S_1, S_2, S_3)$ as the employment status in $1-6$-th quarters, $7-8$-th quarters and $9-14$-th quarters, \ie, we increase the dimesnion of $S_1$ and $S_3$ to  $\op{dim}(S_2) = \op{dim}(S_3) = 6$. 
    {In \Cref{table: numericals3,table: numericals4}, we keep the dimension of $S_1$ and $S_3$ as $2$, and increase $\op{dim}(S_2)$ to $3$ and $4$, respectively. In other words,  in \Cref{table: numericals3}, we set $S_1, S_2$ and $S_3$ as the employment status in the $1-2$-th quarters, $3-5$-th quarters and $6-7$-th quarters, respectively; and in \Cref{table: numericals4} we instead take the $1-2$-th quarters, $3-6$-th quarters and $8-9$-th quarters, respectively.}
    Apparently, with ridge regularization (namely the ``.33'', ``.67'' and ``1'' columns), our estimator is still consistently better than~\citet{athey2020combining} by a large margin in all settings, showing that the performance of our estimator is stable with respect to the number of quarters in surrogate construction. Interestingly, with the existence of ridge regularization, our estimator can perform slightly worse as we increase the dimension of surrogates, which may be 
due to the non-uniqueness of bridge functions. 
When the ridge regularization does not exist (namely the ``0'' column), our estimator can be quite unstable, sometimes even worse than the naive estimator. Such phenomenon has also been observed \Cref{table: numerical}.}

{Since the singular values in \Cref{table: singular,table: singular3,table: singular4} are calculated from the probability matrix estimates without conditioning on the covariates $X$, they cannot directly validate the assumptions underlying \Cref{table: numericalq4,table: numericalq6,table: numericals3,table: numericals4} as they all control for the covariates $X$.
We thus further rerun our numerical experiment in the setting of \Cref{table: numerical}, 
without including the covariates $X$. Our previous heuristic diagnostics in  \Cref{table: singular} are directly relevant to this simpler setting. This setting is also reasonable because the treatment in the observational data is constructed to be confounded only by the omitted education variable $U$ so controlling for $X$ is not required. The results are shown in \Cref{table: numericalnoX}. Importantly, our proposed estimator still outperform the benchmark estimators. This provides a setting where our previous diagnostics are relevant and the qualitative findings in our previous experiments also continue to hold.}

\begin{table}[!t]
    \centering
    \begin{tabular}{c|c|c|c|c|c}
    \hline
      {\scriptsize $\dim(S_1)$} & $s_2$  & {\scriptsize $P (\mathbf{S}_1 \mid S_2 = s_2, A = 0, \mathbf{U})$} & {\scriptsize $P (\mathbf{S}_3 \mid S_2 = s_2, A = 0, \mathbf{U})$} & {\scriptsize $P (\mathbf{S}_1 \mid S_2 = s_2, A = 1, \mathbf{U})$} & {\scriptsize $P (\mathbf{S}_3 \mid S_2 = s_2, A = 1, \mathbf{U})$}\\ \hline
      \multirow{8}{*}{2} & $(0,0,0)$ &  0.010 & 0.005 & 0.002 & 0.003\\ 
       & $(1,0,0)$ &  0.142 & 0.319  & 0.058 & 0.017\\
       & $(0,1,0)$ &  0.340 & 0.000 & 0.031 & 0.033\\
       & $(1,1,0)$ &  0.085 & 0.065 & 0.017 &  0.052 \\
       & $(0,0,1)$ &  0.020 & 0.009 & 0.001 & 0.011\\ 
       & $(1,0,1)$ &  0.104 & 0.092  & 0.071 & 0.025\\
       & $(0,1,1)$ &  0.000 & 0.019 & 0.010 & 0.012\\
       & $(1,1,1)$ &  0.011 & 0.007 & 0.003 &  0.012 \\\hline
    \end{tabular}
    \caption{{Same as \Cref{table: singular}, but with $S_1, S_2, S_3$ taking the quarters $1-2$, $3-5$ and $6-7$, respectively.}
    }
    \label{table: singular3}
\end{table}

\begin{table}[!t]
    \centering
    \begin{tabular}{c|c|c|c|c|c}
    \hline
      {\scriptsize $\dim(S_1)$} & $s_2$  & {\scriptsize $P (\mathbf{S}_1 \mid S_2 = s_2, A = 0, \mathbf{U})$} & {\scriptsize $P (\mathbf{S}_3 \mid S_2 = s_2, A = 0, \mathbf{U})$} & {\scriptsize $P (\mathbf{S}_1 \mid S_2 = s_2, A = 1, \mathbf{U})$} & {\scriptsize $P (\mathbf{S}_3 \mid S_2 = s_2, A = 1, \mathbf{U})$}\\ \hline
      \multirow{16}{*}{2} & $(0,0,0,0)$ &  0.003 & 0.010 & 0.000 & 0.005\\ 
       & $(1,0,0,0)$ &  0.120 & 0.083 & 0.051 & 0.018\\
       & $(0,1,0,0)$ &  0.387 & 0.129 & 0.000 & 0.025\\
       & $(1,1,0,0)$ &  0.088 & 0.118 & 0.005 & 0.024\\
       & $(0,0,1,0)$ &  0.145 & 0.280 & 0.047 & 0.051\\
       & $(1,0,1,0)$ &  0.614 & 0.300 & 0.084 & 0.068\\
       & $(0,1,1,0)$ &  0.210 & 0.396 & 0.015 & 0.013\\
       & $(1,1,1,0)$ &  0.176 & 0.194 & 0.004 & 0.030\\
       & $(0,0,0,1)$ &  0.039 & 0.044 & 0.023 & 0.028\\
       & $(1,0,0,1)$ &  0.707 & 1.000 & 0.025 & 0.084\\
       & $(0,1,0,1)$ &  1.414 & 1.414 & 0.242 & 0.375\\
       & $(1,1,0,1)$ &  0.282 & 0.282 & 0.033 & 0.049\\
       & $(0,0,1,1)$ &  0.054 & 0.103 & 0.029 & 0.066\\
       & $(1,0,1,1)$ &  0.297 & 0.171 & 0.105 & 0.024\\
       & $(0,1,1,1)$ &  0.050 & 0.014 & 0.039 & 0.033\\
       & $(1,1,1,1)$ &  0.033 & 0.062 & 0.005 & 0.004\\\hline
    \end{tabular}
    \caption{{Same as \Cref{table: singular}, but with $S_1, S_2, S_3$ taking the quarters $1-2$, $3-6$ and $7-8$, respectively.}
    }
    \label{table: singular4}
\end{table}

\begin{table}[!t]
\centering 
 \begin{tabular}{cc|cccc|cccc|cccc|cc|c}
    \hline
    & & \multicolumn{4}{c}{$\hat\tau_{\out}$} & \multicolumn{4}{c}{$\hat\tau_{\sel}$} & \multicolumn{4}{c}{$\hat\tau_{\dr}$} &  \multicolumn{2}{c}{\scriptsize Athey et al.} & {\scriptsize Naive}  \\
    \hline
    $\eta$ & & 0 & .33 & .67 & 1 &  0 & .33 & .67 & 1 & 0 & .33 & .67 & 1 & NR & CV &  \\
         \hline
    \multirow{2}{*}{0} & MAE & -560 & 74 & 76 & 77 & 81 & 78 & 83 & 85 & -465 & 78 & 81 & 81 & 13 & 33 & 0.053 \\
& Med & -560 & 74 & 76 & 77 & 81 & 78 & 83 & 85 & -465 & 78 & 81 & 81 & 13 & 33 & 0.053\\
\hline
\multirow{2}{*}{0.2} & MAE & -416 & 73 & 74 & 74 & 78 & 78 & 78 & 78 & -363 & 77 & 79 & 79 & 21 & 31 & 0.059 \\
& Med & 33 & 72 & 74 & 74 & 78 & 77 & 77 & 78 & 41 & 77 & 78 & 79 & 21 & 31 & 0.059\\
\hline
\multirow{2}{*}{0.4} & MAE & -150 & 71 & 72 & 72 & 75 & 75 & 75 & 75 & -151 & 76 & 77 & 77 & 27 & 28 & 0.067 \\
& Med & 43 & 71 & 72 & 72 & 75 & 76 & 75 & 76 & 45 & 76 & 77 & 77 & 28 & 29 & 0.067\\
\hline
\multirow{2}{*}{0.6} & MAE & -577 & 68 & 68 & 68 & 72 & 72 & 72 & 72 & -574 & 73 & 74 & 74 & 34 & 25 & 0.076 \\
& Med & 44 & 68 & 68 & 68 & 73 & 72 & 73 & 72 & 47 & 74 & 74 & 74 & 34 & 25 & 0.076 \\
\hline
\multirow{2}{*}{0.8} & MAE & -685 & 63 & 63 & 63 & 67 & 67 & 67 & 67 & -652 & 69 & 70 & 69 & 37 & 21 & 0.088 \\
& Med & 534 & 62 & 62 & 62 & 67 & 67 & 67 & 67 & 37 & 69 & 69 & 69 & 37 & 21 & 0.088\\
\hline
\multirow{2}{*}{1} & MAE & -135 & 63 & 63 & 63 & 68 & 68 & 68 & 68 & -139 & 70 & 70 & 70 & 38 & 18 & 0.095\\
& Med & 35 & 63 & 63 & 62 & 69 & 69 & 69 & 68 & 36 & 70 & 70 & 71 & 39 & 18 & 0.095 \\
\hline
\multirow{2}{*}{1.2} & MAE & -237 & 62 & 62 & 62 & 68 & 68 & 68 & 68 & -213 & 70 & 70 & 70 & 38 & 15 & 0.104 \\
& Med & 28 & 62 & 62 & 62 & 69 & 69 & 69 & 68 & 31 & 71 & 71 & 71 & 39 & 15 & 0.104 \\
\hline 
\multirow{2}{*}{1.4} & MAE & -241 & 61 & 60 & 59 & 69 & 69 & 68 & 68 & -254 & 70 & 70 & 70 & 37 & 11 & 0.115 \\
& Med & 11 & 61 & 61 & 60 & 71 & 71 & 71 & 70 & 14 & 73 & 73 & 72 & 36 & 11 & 0.115 \\
\hline
\multirow{2}{*}{1.6} & MAE & -271 & 59 & 58 & 57 & 68 & 68 & 68 & 68 & -292 & 70 & 70 & 69 & 37 & 10 & 0.124 \\
& Med & 4 & 60 & 59 & 58 & 71 & 71 & 71 & 70 & 5 & 73 & 72 & 72 & 36 & 10 & 0.124\\
\hline
\end{tabular} 
\caption{Same as \Cref{table: numerical}, but with $S_1, S_2, S_3$ taking the quarters $1 - 4$, $5 - 6$ and $7-10$.
}
\label{table: numericalq4}
\end{table}

\begin{table}[!t]
\centering 
 \begin{tabular}{cc|cccc|cccc|cccc|cc|c}
    \hline
    & & \multicolumn{4}{c}{$\hat\tau_{\out}$} & \multicolumn{4}{c}{$\hat\tau_{\sel}$} & \multicolumn{4}{c}{$\hat\tau_{\dr}$} &  \multicolumn{2}{c}{\scriptsize Athey et al.} & {\scriptsize Naive}  \\
    \hline
    $\eta$ & & 0 & .33 & .67 & 1 &  0 & .33 & .67 & 1 & 0 & .33 & .67 & 1 & NR & CV &  \\
         \hline
    \multirow{2}{*}{0} & MAE & -1030 & 74 & 72 & 71 & 75 & 73 & 75 & 72 & -1010 & 77 & 76 & 75 & 17 & 39 & 0.053 \\
& Med & -1030 & 74 & 72 & 71 & 75 & 73 & 75 & 72 & -1010 & 77 & 76 & 75 & 17 & 39 & 0.053\\
\hline
\multirow{2}{*}{0.2} & MAE & -9723 & 72 & 70 & 69 & 75 & 75 & 74 & 74 & -9789 & 76 & 75 & 74 & 23 & 39 & 0.059 \\
& Med & -184 & 72 & 70 & 69 & 75 & 75 & 74 & 74 & -184 & 76 & 75 & 74 & 23 & 38 & 0.059\\
\hline
\multirow{2}{*}{0.4} & MAE & -1102 & 70 & 68 & 67 & 74 & 73 & 73 & 73 & -1113 & 74 & 73 & 73 & 29 & 38 & 0.067 \\
& Med & -152 & 70 & 68 & 67 & 74 & 73 & 73 & 73 & -152 & 75 & 74 & 73 & 29 & 38 & 0.067\\
\hline
\multirow{2}{*}{0.6} & MAE & -861 & 67 & 65 & 64 & 72 & 71 & 71 & 71 & -876 & 72 & 71 & 70 & 35 & 36 & 0.076 \\
& Med & -149 & 67 & 65 & 64 & 72 & 72 & 71 & 71 & -152 & 72 & 72 & 71 & 35 & 36 & 0.076 \\
\hline
\multirow{2}{*}{0.8} & MAE & -8496 & 62 & 60 & 59 & 68 & 68 & 67 & 67 & -8513 & 69 & 67 & 67 & 37 & 32 & 0.088 \\
& Med & -94 & 61 & 59 & 58 & 68 & 68 & 67 & 66 & -92 & 69 & 67 & 67 & 37 & 32 & 0.088\\
\hline
\multirow{2}{*}{1} & MAE & -640 & 59 & 58 & 57 & 68 & 68 & 67 & 67 & -645 & 67 & 66 & 66 & 38 & 30 & 0.095\\
& Med & -111 & 59 & 58 & 57 & 69 & 68 & 67 & 67 & -112 & 67 & 66 & 66 & 38 & 29 & 0.095\\
\hline
\multirow{2}{*}{1.2} & MAE & -459 & 57 & 55 & 55 & 68 & 67 & 67 & 66 & -467 & 65 & 65 & 64 & 37 & 26 & 0.104 \\
& Med & -96 & 57 & 56 & 55 & 69 & 68 & 68 & 68 & -102 & 66 & 66 & 65 & 37 & 26 & 0.104 \\
\hline 
\multirow{2}{*}{1.4} & MAE & -2157 & 53 & 51 & 51 & 66 & 66 & 66 & 65 & -2210 & 63 & 63 & 63 & 33 & 21 & 0.115 \\
& Med & -96 & 53 & 52 & 51 & 69 & 68 & 67 & 67 & -101 & 65 & 64 & 64 & 33 & 21 & 0.115 \\
\hline
\multirow{2}{*}{1.6} & MAE & -683 & 50 & 49 & 48 & 66 & 65 & 65 & 64 & -714 & 62 & 62 & 62 & 31 & 17 & 0.124 \\
& Med & -73 & 50 & 49 & 48 & 68 & 67 & 66 & 66 & -78 & 64 & 63 & 63 & 31 & 17 & 0.124\\
\hline
\end{tabular} 
\caption{Same as \Cref{table: numerical}, but with $S_1, S_2, S_3$ taking the quarters $1 - 6$, $7 - 8$ and $9 - 14$.
}
\label{table: numericalq6}
\end{table}

\begin{table}[!t]
\centering 
 \begin{tabular}{cc|cccc|cccc|cccc|cc|c}
    \hline
    & & \multicolumn{4}{c}{$\hat\tau_{\out}$} & \multicolumn{4}{c}{$\hat\tau_{\sel}$} & \multicolumn{4}{c}{$\hat\tau_{\dr}$} &  \multicolumn{2}{c}{\scriptsize Athey et al.} & {\scriptsize Naive}  \\
    \hline
    $\eta$ & & 0 & .33 & .67 & 1 &  0 & .33 & .67 & 1 & 0 & .33 & .67 & 1 & NR & CV &  \\
         \hline
    \multirow{2}{*}{0} & MAE & 47 & 80 & 79 & 79 & 82 & 81 & 81 & 78 & 45 & 84 & 84 & 84 & 14 & 23 & 0.053 \\
& Med & 47 & 80 & 79 & 79 & 82 & 81 & 81 & 78 & 45 & 84 & 84 & 84 & 14 & 23 & 0.053\\
\hline
\multirow{2}{*}{0.2} & MAE & 2 & 76 & 76 & 75 & 79 & 79 & 79 & 79 & 3 & 81 & 81 & 81 & 21 & 21 & 0.059 \\
& Med & 57 & 76 & 76 & 75 & 79 & 79 & 79 & 79 & 57 & 81 & 81 & 81 & 21 & 21 & 0.059\\
\hline
\multirow{2}{*}{0.4} & MAE & -122 & 73 & 72 & 72 & 76 & 76 & 76 & 76 & -109 & 78 & 78 & 77 & 27 & 19 & 0.067 \\
& Med & 52 & 73 & 73 & 72 & 77 & 76 & 77 & 76 & 54 & 78 & 78 & 78 & 27 & 19 & 0.067\\
\hline
\multirow{2}{*}{0.6} & MAE & -164 & 68 & 68 & 67 & 72 & 72 & 72 & 72 & -147 & 74 & 74 & 74 & 32 & 16 & 0.076 \\
& Med & 40 & 69 & 68 & 68 & 73 & 73 & 73 & 73 & 46 & 75 & 74 & 74 & 32 & 16 & 0.076 \\
\hline
\multirow{2}{*}{0.8} & MAE & -240 & 62 & 62 & 61 & 67 & 67 & 67 & 67 & -191 & 68 & 68 & 68 & 33 & 12 & 0.088 \\
& Med & 36 & 62 & 61 & 61 & 67 & 67 & 67 & 67 & 46 & 69 & 69 & 69 & 33 & 12 & 0.088\\
\hline
\multirow{2}{*}{1} & MAE & -204 & 62 & 61 & 61 & 68 & 68 & 67 & 67 & -161 & 69 & 69 & 68 & 35 & 10 & 0.095\\
& Med & 35 & 62 & 61 & 61 & 68 & 68 & 68 & 68 & 45 & 69 & 69 & 69 & 36 & 10 & 0.095 \\
\hline
\multirow{2}{*}{1.2} & MAE & -205 & 61 & 60 & 59 & 68 & 68 & 67 & 67 & -163 & 69 & 69 & 68 & 36 & 8 & 0.104 \\
& Med & 31 & 61 & 60 & 60 & 69 & 69 & 68 & 68 & 43 & 70 & 70 & 70 & 37 & 8 & 0.104 \\
\hline 
\multirow{2}{*}{1.4} & MAE & -214 & 58 & 57 & 57 & 68 & 68 & 67 & 67 & -151 & 69 & 68 & 68 & 36 & 6 & 0.115 \\
& Med & 29 & 59 & 58 & 57 & 71 & 71 & 70 & 70 & 43 & 72 & 72 & 71 & 36 & 6 & 0.115 \\
\hline
\multirow{2}{*}{1.6} & MAE & -314 & 57 & 56 & 54 & 67 & 67 & 67 & 66 & -253 & 68 & 68 & 67 & 37 & 6 & 0.124 \\
& Med & 30 & 58 & 56 & 54 & 71 & 70 & 70 & 69 & 44 & 71 & 71 & 71 & 38 & 6 & 0.124\\
\hline
\end{tabular} 
\caption{{Same as \Cref{table: numerical}, but with $S_1, S_2, S_3$ taking the quarters $1-2$, $3-5$ and $6-7$.}
}
\label{table: numericals3}
\end{table}

\begin{table}[!t]
\centering 
 \begin{tabular}{cc|cccc|cccc|cccc|cc|c}
    \hline
    & & \multicolumn{4}{c}{$\hat\tau_{\out}$} & \multicolumn{4}{c}{$\hat\tau_{\sel}$} & \multicolumn{4}{c}{$\hat\tau_{\dr}$} &  \multicolumn{2}{c}{\scriptsize Athey et al.} & {\scriptsize Naive}  \\
    \hline
    $\eta$ & & 0 & .33 & .67 & 1 &  0 & .33 & .67 & 1 & 0 & .33 & .67 & 1 & NR & CV &  \\
         \hline
    \multirow{2}{*}{0} & MAE & 27 & 80 & 80 & 80 & 82 & 84 & 83 & 85 & 35 & 85 & 85 & 85 & 15 & 28 & 0.053 \\
& Med & 27 & 80 & 80 & 80 & 82 & 84 & 83 & 85 & 35 & 85 & 85 & 85 & 15 & 28 & 0.053\\
\hline
\multirow{2}{*}{0.2} & MAE & 21 & 77 & 77 & 76 & 80 & 80 & 80 & 80 & 29 & 82 & 82 & 82 & 21 & 26 & 0.059 \\
& Med & 35 & 77 & 77 & 76 & 80 & 80 & 80 & 80 & 43 & 82 & 82 & 82 & 22 & 26 & 0.059 \\
\hline
\multirow{2}{*}{0.4} & MAE & -262 & 74 & 73 & 73 & 77 & 77 & 77 & 77 & -245 & 79 & 79 & 78 & 28 & 24 & 0.067 \\
& Med & 42 & 74 & 73 & 73 & 78 & 78 & 78 & 78 & 50 & 80 & 79 & 79 & 28 & 24 & 0.067 \\
\hline
\multirow{2}{*}{0.6} & MAE & -90 & 69 & 69 & 68 & 73 & 73 & 73 & 73 & -78 & 75 & 75 & 74 & 34 & 21 & 0.076 \\
& Med & 47 & 70 & 69 & 69 & 74 & 74 & 73 & 74 & 53 & 76 & 75 & 75 & 34 & 21 & 0.076 \\
\hline
\multirow{2}{*}{0.8} & MAE & -1 & 63 & 63 & 62 & 68 & 68 & 68 & 68 & 7 & 70 & 69 & 69 & 35 & 16 & 0.088 \\
& Med & 45 & 63 & 62 & 62 & 68 & 68 & 68 & 68 & 52 & 70 & 69 & 69 & 35 & 17 & 0.088 \\
\hline
\multirow{2}{*}{1} & MAE & -26 & 63 & 62 & 62 & 68 & 68 & 68 & 68 & -20 & 70 & 69 & 69 & 38 & 14 & 0.095\\
& Med & 43 & 63 & 62 & 61 & 69 & 69 & 69 & 69 & 50 & 71 & 70 & 70 & 39 & 14 & 0.095 \\
\hline
\multirow{2}{*}{1.2} & MAE & -60 & 62 & 61 & 61 & 68 & 68 & 68 & 68 & -49 & 70 & 69 & 69 & 39 & 11 & 0.104 \\
& Med & 40 & 62 & 62 & 61 & 70 & 69 & 69 & 69 & 48 & 72 & 71 & 71 & 40 & 11 & 0.104 \\
\hline 
\multirow{2}{*}{1.4} & MAE & -37 & 60 & 59 & 58 & 68 & 69 & 68 & 68 & -30 & 70 & 69 & 69 & 39 & 8 & 0.115 \\
& Med & 35 & 61 & 59 & 58 & 72 & 72 & 72 & 70 & 46 & 73 & 73 & 72 & 39 & 8 & 0.115 \\
\hline
\multirow{2}{*}{1.6} & MAE & -51 & 59 & 57 & 56 & 68 & 68 & 68 & 67 & -18 & 69 & 69 & 68 & 40 & 7 & 0.124 \\
& Med & 35 & 59 & 57 & 56 & 72 & 71 & 71 & 70 & 45 & 72 & 72 & 71 & 40 & 7 & 0.124\\
\hline
\end{tabular} 
\caption{{Same as \Cref{table: numerical}, but with $S_1, S_2, S_3$ taking the quarters $1-2$, $3-6$ and $7-8$.}
}
\label{table: numericals4}
\end{table}

\begin{table}[!t]
\centering 
 \begin{tabular}{cc|cccc|cccc|cccc|cc|c}
    \hline
    & & \multicolumn{4}{c}{$\hat\tau_{\out}$} & \multicolumn{4}{c}{$\hat\tau_{\sel}$} & \multicolumn{4}{c}{$\hat\tau_{\dr}$} &  \multicolumn{2}{c}{\scriptsize Athey et al.} & {\scriptsize Naive}  \\
    \hline
    $\eta$ & & 0 & .33 & .67 & 1 &  0 & .33 & .67 & 1 & 0 & .33 & .67 & 1 & NR & CV &  \\
         \hline
    \multirow{2}{*}{0} & MAE & 93 & 72 & 67 & 65 & 84 & 76 & 71 & 68 & 94 & 81 & 76 & 73 & 61 & 29 & 0.053 \\
& Med & 93 & 72 & 67 & 65 & 84 & 76 & 71 & 68 & 94 & 81 & 76 & 73 & 61 & 29 & 0.053\\
\hline
\multirow{2}{*}{0.2} & MAE & 60 & 64 & 60 & 58 & 75 & 68 & 64 & 61 & 59 & 73 & 68 & 65 & 55 & 25 & 0.059 \\
& Med & 86 & 64 & 60 & 58 & 75 & 68 & 63 & 61 & 85 & 73 & 68 & 65 & 55 & 25 & 0.059\\
\hline
\multirow{2}{*}{0.4} & MAE & 43 & 57 & 54 & 52 & 67 & 61 & 57 & 54 & 43 & 65 & 61 & 58 & 49 & 21 & 0.067 \\
& Med & 76 & 58 & 54 & 52 & 68 & 61 & 57 & 55 & 75 & 65 & 61 & 58 & 50 & 21 & 0.067\\
\hline
\multirow{2}{*}{0.6} & MAE & 23 & 50 & 47 & 46 & 59 & 53 & 50 & 48 & 26 & 57 & 53 & 51 & 43 & 16 & 0.076 \\
& Med & 65 & 51 & 48 & 46 & 60 & 54 & 50 & 48 & 65 & 57 & 54 & 51 & 44 & 16 & 0.076 \\
\hline
\multirow{2}{*}{0.8} & MAE & 22 & 43 & 41 & 39 & 50 & 46 & 43 & 41 & 24 & 49 & 45 & 44 & 37 & 12 & 0.088 \\
& Med & 56 & 43 & 40 & 39 & 50 & 46 & 43 & 41 & 54 & 49 & 46 & 44 & 37 & 13 & 0.088\\
\hline
\multirow{2}{*}{1} & MAE & 27 & 40 & 38 & 37 & 48 & 42 & 39 & 38 & 31 & 45 & 42 & 40 & 34 & 11 & 0.095\\
& Med & 51 & 40 & 38 & 37 & 47 & 43 & 40 & 38 & 50 & 46 & 42 & 41 & 35 & 10 & 0.095 \\
\hline
\multirow{2}{*}{1.2} & MAE & 18 & 36 & 34 & 33 & 45 & 38 & 36 & 34 & 23 & 41 & 38 & 36 & 31 & 9 & 0.104 \\
& Med & 47 & 36 & 34 & 33 & 44 & 38 & 36 & 34 & 45 & 41 & 38 & 37 & 31 & 9 & 0.104 \\
\hline 
\multirow{2}{*}{1.4} & MAE & -6 & 32 & 30 & 29 & 41 & 33 & 31 & 30 & 16 & 36 & 33 & 32 & 28 & 8 & 0.115 \\
& Med & 40 & 32 & 30 & 29 & 40 & 33 & 31 & 30 & 41 & 36 & 33 & 32 & 28 & 7 & 0.115 \\
\hline
\multirow{2}{*}{1.6} & MAE & -3 & 29 & 27 & 27 & 39 & 30 & 28 & 27 & -7 & 32 & 30 & 29 & 25 & 7 & 0.124 \\
& Med & 37 & 29 & 27 & 26 & 38 & 30 & 28 & 27 & 38 & 32 & 30 & 29 & 26 & 7 & 0.124\\
\hline
\end{tabular} 
\caption{{Same as \Cref{table: numerical}, but without including the covariates $X$.}
}
\label{table: numericalnoX}
\end{table}

\subsection{Implementation details of the minimax approach in \Cref{sec: simulation}}

{
In this section, we provide implementation details of the minimax approach in~\Cref{sec: simulation}. To construct the outcome bridge function, we set the outer minimization function class as a neural network class with four layers. For $\dim(X) = 10, 15, 20$, we choose the number of neurons in first and second hidden layers to be $50$ and $10$, respectively; for $\dim(X) = 5$, we change the number of neurons in the first hidden layer to $30$. We set the momentum, learning rate, number of epochs of the neural network optimizer to be 0.95, 0.0002 and 40, respectively; and set the size of each batch to be 1 / 10 of the total sample size. We use the ReLU activation function for the first three layers and set the activation function for the last layer as a linear function. For the inner maximization function class, we set it a RKHS class with a product radial basis function
kernel. To construct the selection bridge function, we use a similar neural architecture as in the outcome bridge function construction, except that we set the activation function in the last layer as a softplus activation function; for the inner maximization function class, we use a RKHS with a linear kernel.}

{The simulation results in \Cref{table: simulation} correspond to $U, S_1, S_3$ equal to $5$. We further increase their dimension to $10$ and report the corresponding results in \Cref{table: simulationls}. 
To accommodate the increased dimensionality, we raise the number of neurons in the first layer from $50$ to $70$ for $\dim(X) = 10, 15, 20$. We observe that the confidence interval coverage becomes slightly worse in this higher dimensional setting.}

\begin{table}[t!]
    \centering
    \begin{tabular}{cc|cc|cc|cc|cc}
    \hline
    & & \multicolumn{2}{c}{$\textrm{dim}(X) = 5$}
     & \multicolumn{2}{c}{$\textrm{dim}(X) = 10$} & \multicolumn{2}{c}{$\textrm{dim}(X) = 15$} & \multicolumn{2}{c}{$\textrm{dim}(X) = 20$}\\ \hline
       $q$ & & MinMax & Param. &  MinMax & Param.  & MinMax & Param.  & MinMax & Param. \\ \hline 
       \multirow{4}{*}{1} & CP & 96.0\% & 94.5\% & 91.5\% & 92.5\% & 94.0\% & 94.0\% & 93.0\% & 94.0\% \\
       & CI Len. & 0.793 & 0.820 & 0.793 & 0.817  & 0.787 & 0.817  & 0.781 & 0.810 \\
       & RMSE & 0.212 & 0.216 &  0.211 & 0.212 & 0.205 & 0.213  & 0.205 & 0.209 \\
       & Bias & 0.016 & 0.001 & 0.036 & 0.001  & 0.012 & 0.011  & 0.025 & 0.005 \\ \hline 
       \multirow{4}{*}{1.5} & CP & 92.0\% & 94.0\% & 90.5\% & 93.0\% & 92.0\% & 93.0\% & 93.5\% & 95.0\% \\
       & CI Len. &  0.883 & 1.196 & 0.848 & 1.168 & 0.849 & 1.163 & 0.839 & 1.163 \\
       & RMSE & 0.261 & 0.319 & 0.235 & 0.310 & 0.261 & 0.312  & 0.239 & 0.294 \\
       & Bias & 0.104 & 0.108 & 0.070 & 0.069 & 0.121 & 0.098  & 0.091 & 0.078 \\ \hline 
       \multirow{4}{*}{2} & CP & 90.0\% & 95.5\% & 90.5\% & 94.0\% & 91.5\% & 94.0\% & 91.5\% & 92.5\% \\
       & CI Len. & 0.949 & 3.682 & 0.882 & 3.414 & 0.895 & 3.392 & 0.882 & 3.373 \\
       & RMSE & 0.314 & 1.068 & 0.246 & 0.974 & 0.265 & 0.981  & 0.253 & 1.012 \\
       & Bias & 0.165 & 0.219 & 0.062 & 0.068  & 0.111 & 0.214 & 0.093 &  0.008 \\
       \hline
    \end{tabular}
    \caption{{Same setting as \Cref{table: simulation}, but with the dimension of $S_1, S_3$ and $U$ equal to $10$}.}
    \label{table: simulationls}
\end{table}

\section{Proofs}\label{sec: proof}
\subsection{Supporting Lemmas}\label{sec: support}
\begin{lemma}\label{cor: implication1}
Under \Cref{assump: unconfound-obs,assump: unconfound-exp,assump: ext-valid}, we have 
\begin{align}\label{eq: ext-valid-obs}
\prns{S_3, S_2} \perp G \mid   A, U, X.
\end{align}
\end{lemma}

\begin{proof}
For any $a \in \mathcal{A}$, $s_3\in\Scal_3, s_2 \in \Scal_2$ and $g \in \braces{E, O}$, we have 
\begin{align*}
p_{S_3, S_2}\prns{s_3, s_2 \mid U, X, A = a, G = g}         &=  p_{S_3\prns{a}, S_2\prns{a}}\prns{s_3, s_2 \mid U, X, A = a, G = g} \\
      &=  p_{S_3\prns{a}, S_2\prns{a}}\prns{s_3, s_2 \mid U, X, G = g} \\
      &= p_{S_3\prns{a}, S_2\prns{a}}\prns{s_3, s_2 \mid U, X} \\
      &= p_{S_3, S_2}\prns{s_3, s_2 \mid U, X, A = a},
\end{align*}
where the second equation follows from \Cref{assump: unconfound-exp,assump: unconfound-obs} and the third equation follows from \Cref{assump: ext-valid}.
\end{proof}

\begin{lemma}\label{cor: implication2}
Under \Cref{assump: unconfound-obs,assump: CI}, we have 
\begin{align}
\prns{Y, S_3} \perp S_1 \mid S_2, A, U, X,  G = O. \label{eq: CI-1}
\end{align}  
\end{lemma}

\begin{proof}
For any $a \in \mathcal{A}, s \in \Scal_2$ and any bounded continous functions $f: \mathcal{Y} \times \Scal_3 \to \Rl$ and $g: \mathcal{S}_1 \to \Rl$, we have 
\begin{align*}
&\Eb{f\prns{Y, S_3}g\prns{S_1} \mid S_2 = s, U, X, A = a, G = O} \\
    =& \Eb{f\prns{Y\prns{a}, S_3\prns{a}}g\prns{S_1\prns{a}} \mid S_2\prns{a} = s, U, X, A = a, G = O}  \\
    =& \Eb{f\prns{Y\prns{a}, S_3\prns{a}}g\prns{S_1\prns{a}} \mid S_2\prns{a} = s, U, X, G = O} \\
    =& \Eb{f\prns{Y\prns{a}, S_3\prns{a}} \mid S_2\prns{a} = s, U, X, G = O}\Eb{g\prns{S_1\prns{a}} \mid S_2\prns{a} = s, U, X, G = O} \\
    =& \Eb{f\prns{Y, S_3} \mid S_2 = s, U, X, A = a, G = O} \Eb{g\prns{S_1} \mid S_2 = s, U, X, A = a, G = O},
\end{align*}
where the second equation follows from \Cref{assump: unconfound-obs}, the third equation follows from \Cref{eq: CI-potential-1} in \Cref{assump: CI} , and the fourth equation again follows from \Cref{assump: unconfound-obs}.  
\end{proof}

\begin{lemma}\label{lemma: data-overlap}
Under \Cref{assump: ext-valid}, for any $a \in \Acal$, the following holds almost surely:
\begin{align}
\frac{p\prns{\Sb, U, X \mid A=a, G = E}}{p\prns{\Sb, U, X \mid A=a, G = O}} = \frac{p\prns{U, X \mid A=a, G = E}}{p\prns{U, X \mid A=a, G = O}} < \infty \label{eq: data-overlap2}
\end{align}
\end{lemma}
\begin{proof}
This is proved by noting that 
\begin{align*}
\frac{p\prns{\Sb, U, X \mid A=a, G = E}}{p\prns{\Sb, U, X \mid A=a, G = O}} 
    &= \frac{p\prns{\Sb(a), U, X \mid A=a, G = E}}{p\prns{\Sb(a), U, X \mid A=a, G = O}} \\
    &= \frac{p\prns{\Sb(a) \mid U, X,  A=a, G = E}}{p\prns{\Sb(a) \mid U, X , A=a, G = O}} \frac{p\prns{ U, X \mid A=a, G = E}}{p\prns{U, X  \mid A=a, G = O}}  \\
    &= \frac{p\prns{ U, X \mid A=a, G = E}}{p\prns{U, X  \mid A=a, G = O}} < \infty.
\end{align*}
where the last equation follows from \Cref{eq: external-validity-0} in \Cref{assump: ext-valid}.
\end{proof}

\subsection{Proofs for \Cref{sec: identify-OBF}}

\begin{proof}[Proof for \cref{lemma: bridge-obs}]
In \cref{cor: implication2}, we already proved that \Cref{assump: unconfound-obs,assump: CI} imply 
\begin{align*}
\prns{Y, S_3} \perp S_1 \mid S_2, A, U, X,  G = O. 
\end{align*}  
Therefore, for any function $h_0\prns{S_3, S_2, A, X}$, we have 
\begin{align*}
\Eb{Y \mid S_2, S_1, A, G = O} 
    &= \Eb{\Eb{Y \mid S_2, S_1, A, U, X, G = O} \mid S_2, S_1, A, X, G = O} \\
    &= \Eb{\Eb{Y \mid S_2, A, U, X, G = O} \mid S_2, S_1, A, X, G = O},
\end{align*}
and 
\begin{align*}
&\Eb{h_0\prns{S_3, S_2, A, X} \mid S_2, S_1, A, X, G = O} \\  
    =& \Eb{\Eb{h_0\prns{S_3, S_2, A, X} \mid S_2, S_1, A, U, X, G = O} \mid S_2, S_1, A, X, G = O} \\ 
    =& \Eb{\Eb{h_0\prns{S_3, S_2, A, X} \mid S_2, A, U, X, G = O} \mid S_2, S_1, A, X, G = O}.
\end{align*}
Therefore, for any $h_0\prns{S_3, S_2, A, X}$ that satisfies \cref{eq: bridge-U}, we have 
\begin{align*}
0   &= \Eb{Y - h_0\prns{S_3, S_2, A, X} \mid S_2, S_1, A, X, G = O} \\
    &= \Eb{\Eb{Y - h_0\prns{S_3, S_2, A, X} \mid S_2, A, U, X, G = O} \mid S_2, S_1, A, X, G = O}.
\end{align*}
It follows from the completeness condition in \Cref{assump: completeness}  condition \ref{assump: completeness-2} that 
\begin{align*}
\Eb{Y - h_0\prns{S_3, S_2, A, X} \mid S_2, A, U, X, G = O}  = 0,
\end{align*}
Namely, any function $h_0\prns{S_3, S_2, A, X}$ that satisfies \cref{eq: bridge-obs} is a valid outcome bridge function satisfying \cref{eq: bridge-U}. 
\end{proof}

\begin{proof}[Proof for \cref{thm: identification1}]
According to \Cref{lemma: bridge-obs}, any function $h_0$ that solves \Cref{eq: bridge-obs} also satisfies  \Cref{eq: bridge-U}. Thus we only need to show that  for any function $h_0$ that solves \Cref{eq: bridge-U}, we have $\mu(a) = \Eb{h_0\prns{\Sc, \Sb, A, X} \mid A = a, G = E} $. This is proved as follows: 
\begin{align*}
&\Eb{h_0\prns{\Sc, \Sb, A, X} \mid A = a, G = E} \\
    =& \Eb{\Eb{h_0\prns{\Sc, \Sb, a, X} \mid \Sb, A = a, U, X, G = E} \mid A = a, G = E} \\
    =& \Eb{\Eb{h_0\prns{\Sc, \Sb, a, X} \mid \Sb, A = a, U, X, G = O} \mid A = a, G = E} \\
    =& \Eb{\Eb{Y \mid \Sb, A = a, U, X, G = O} \mid A = a, G = E} \\
    =& \Eb{\Eb{Y(a) \mid \Sb(a), U, X, G = O} \mid A = a, G = E} \\
    =& \Eb{\Eb{Y(a) \mid \Sb(a), U, X, G = O} \mid G = E} \\
    =& \Eb{\Eb{Y(a) \mid \Sb(a), U, X, G = O} \mid G = O} = \Eb{Y\prns{a} \mid G = O} = \mu(a),
\end{align*}
where the second equation uses the fact that $G \perp S_3 \mid S_2, A = a, U, X$ (see \Cref{eq: ext-valid-obs} in \Cref{cor: implication1}) and \Cref{eq: data-overlap2} in \Cref{lemma: data-overlap}, the third equation uses the definition of the outcome bridge function, the fourth equation uses the fact that $Y\prns{a} \perp A \mid S_2\prns{a}, U, X, G = O$ according to \Cref{assump: unconfound-obs}, the fifth uses the fact that $\prns{S_2\prns{a}, U, X} \perp A \mid G = E$ according to \Cref{assump: unconfound-exp}, and the sixth equation holds because $G \perp \prns{S_2\prns{a}, U, X}$ in \Cref{assump: ext-valid} and \Cref{eq: data-overlap2} in \Cref{lemma: data-overlap}. 
\end{proof}

\subsection{Proofs for \Cref{sec: identify-SBF}}
\begin{proof}[Proof for \Cref{lemma: bridge2-obs}]
First note that 
\begin{align*}
&p\prns{S_3, S_2, X \mid A, G = E} \\
    =& \int p\prns{S_3 \mid S_2, A, U = u, X, G = E}p\prns{S_2, u, X \mid A, G = E}\diff u \\
    =& \int p\prns{S_3 \mid S_2, A, U = u, X, G = O}p\prns{S_2, u, X \mid A, G = E}\diff u, 
\end{align*}
where the second equation follows  from $S_3\perp G \mid S_2, A, U, X$ that we prove in \Cref{cor: implication1}. 

Next, note that 
\begin{align*}
&p\prns{S_3, S_2, X \mid A, G = O} \Eb{q_0\prns{S_2, S_1, A, X} \mid S_3, S_2, A, X, G = O} \\
=& p\prns{S_3, S_2, X \mid A, G = O} \int p\prns{u \mid  S_3, S_2, A, X, G = O}\Eb{q_0\prns{S_2, S_1, A, X} \mid S_3, S_2, A, U = u, X, G = O} \diff u \\
=& p\prns{S_3, S_2, X \mid A, G = O} \int p\prns{u \mid  S_3, S_2, A, X, G = O}\Eb{q_0\prns{S_2, S_1, A, X} \mid S_2, A, U = u, X, G = O} \diff u  \\
=& \int p\prns{S_3, S_2, u, X \mid  A, G = O}\Eb{q_0\prns{S_2, S_1, A, X} \mid S_2, A, U = u, X, G = O} \diff u \\
=& \int p\prns{S_3 \mid S_2, A, U = u, X, G = O}p\prns{S_2, u, X \mid A, G = O}\Eb{q_0\prns{S_2, S_1, A, X} \mid S_2, A, U = u, X, G = O} \diff u,
\end{align*}
where the second equation follows from $S_1 \perp S_3 \mid S_2, A, U, X, G = O$ that we prove in \Cref{cor: implication2}. 

Therefore, any function $q_0$ that satisfies \Cref{eq: bridge2-obs-1} must satisfy 
\begin{align*}
\int p\prns{S_3 \mid S_2, A, U = u, X, G = O} \Delta\prns{S_2, A, u, X} \diff u = 0,
\end{align*}
where 
\begin{align*}
\Delta\prns{S_2, A, U, X} = p\prns{S_2, U, X \mid A, G = E} - p\prns{S_2, U, X \mid A, G = O}\Eb{q_0\prns{S_2, S_1, A, X} \mid S_2, A, U, X, G = O}.
\end{align*}
By Bayes rule, this is equivalent to 
\begin{align*}
P\prns{S_3 \mid S_2, A, X, G = O}\Eb{\frac{\Delta\prns{S_2, A, U, X}}{p\prns{U \mid S_2, A, X, G = O}} \mid S_3, S_2, A, X, G = O} = 0.
\end{align*}
According to \cref{assump: completeness}  condition \ref{assump: completeness-1}, we have $\Delta\prns{S_2, A, U, X} = 0$ almost surely. 
In other words, if  $q_0$ satisfies \Cref{eq: bridge2-obs-1}, then it must also satisfy \Cref{eq: bridge2-U}. 
\end{proof}

\begin{lemma}\label{lemma: bridge2-equivalence}
Under assumptions in \Cref{lemma: bridge2-obs}, \Cref{eq: bridge2-obs-1} is equivalent to \Cref{eq: bridge2-obs-2}.
\end{lemma}
\begin{proof}
Note that \Cref{eq: bridge2-obs-1} is equivalent to 
\begin{align*}
\Eb{q_0\prns{S_2, S_1, A, X} \mid S_3, S_2, A, X, G = O}  &= \frac{p\prns{\Sc, \Sb, X\mid A, G = E}}{p\prns{\Sc, \Sb, X \mid A, G = O}} \\
    &= \frac{\Prb{G = E \mid S_3, S_2, A, X}\Prb{G = O \mid A}}{\Prb{G = O \mid S_3, S_2, A, X}\Prb{G = E \mid A}} \\
  &= \frac{\prns{1 - \Prb{G = O \mid S_3, S_2, A, X}}\Prb{G = O \mid A}}{\Prb{G = O \mid S_3, S_2, A, X}\Prb{G = E \mid A}}.
\end{align*}
It is equivalent to 
\begin{align*}
&\Prb{G = O \mid S_3, S_2, A, X}\Eb{\frac{\Prb{G = E \mid A}}{\Prb{G = O \mid A}} q_0\prns{S_2, S_1, A, X} \mid S_3, S_2, A, X, G = O}\\
 =& 1 - \Prb{G = O \mid S_3, S_2, A, X},
\end{align*}
or 
\begin{align*}
\Prb{G = O \mid S_3, S_2, A, X}\Eb{\frac{\Prb{G = E \mid A}}{\Prb{G = O \mid A}} q_0\prns{S_2, S_1, A, X} + 1 \mid S_3, S_2, A, G = O}  = 1.
\end{align*}
The conclusion then  follows straightforwardly. 
\end{proof}

\begin{proof}[Proof for \Cref{thm: identification2}]
According to \Cref{lemma: bridge2-obs}, any function $q_0$ that solves \Cref{eq: bridge2-obs-1} or \Cref{eq: bridge2-obs-2} must also satisfy \Cref{eq: bridge2-U}. 
Thus we only need to show that for any $q_0$ that solves \Cref{eq: bridge2-U}, we have $$\mu(a) = \Eb{q\prns{S_2, S_1, A, X}Y \mid A = a, G = O}.$$ 
This is proved as follows: 
\begin{align*}
&\Eb{q_0\prns{S_2, S_1, A, X}Y \mid A = a, G = O} \\ 
    =& \Eb{\Eb{q_0\prns{S_2, S_1, A, X}Y \mid S_2, A, U, X, G= O} \mid A = a, G = O} \\
    =& \Eb{\Eb{q_0\prns{S_2, S_1, A, X} \mid S_2,  A, U, X, G= O}\Eb{Y \mid S_2, A, U, X, G = O} \mid A = a, G = O} \\
    =& \mathbb{E}\bigg[\Eb{q_0\prns{S_2(a), S_1(a), A, X}\mid S_2(a), A = a, U, X, G= O} \\
    &\qquad\qquad\qquad \times \Eb{Y(a) \mid S_2(a), A = a, U, X, G= O} \mid A = a, G = O\bigg] \\
    =& \Eb{\Eb{q_0\prns{S_2, S_1, A, X}\mid S_2, A, U, X , G= O}\Eb{Y(a) \mid S_2(a), U, X, G= O} \mid A = a, G = O} \\
    =& \Eb{\frac{p\prns{S_2, U, X \mid A, G = E}}{p\prns{S_2, U, X \mid A, G = O}}\Eb{Y(a) \mid S_2(a), U, X, G= O} \mid A = a, G = O } \\
    =& \Eb{\frac{p\prns{S_2(a), U, X \mid A = a, G = E}}{p\prns{S_2(a), U, X \mid A = a, G = O}}\Eb{Y(a) \mid S_2(a), U, X, G= O} \mid A = a, G = O }  \\
    =& \Eb{\Eb{Y(a) \mid S_2(a), U, X, G= O} \mid A = a, G = E} \\
    =& \Eb{\Eb{Y(a) \mid S_2(a), U, X, G= O} \mid G = E}  \\
    =& \Eb{\Eb{Y(a) \mid S_2(a), U, X, G= O} \mid G = O} \\
    =& \Eb{Y(a) \mid G = O} = \mu(a). 
\end{align*}
Here the second equation uses the fact that $Y \perp S_1 \mid S_2, A, U, X, G = O$ that we prove in \Cref{cor: implication2}, the fourth equation uses the fact that $Y\prns{a} \perp A \mid S_2\prns{a}, U, X, G = O$ according to \Cref{assump: unconfound-obs}, the fifth equation uses the definition of the selection bridge function $q_0\prns{S_2, S_1, A, X}$, the seventh equation uses change of measure, the eighth equation uses the fact that $A \perp \prns{S_2\prns{a}, U, X} \mid G = E$ according to \Cref{assump: unconfound-exp}, and the ninth equation uses the fact that $G \perp \prns{S_2\prns{a}, U, X}$ according to \Cref{assump: ext-valid}.
\end{proof}

\subsection{Proofs for \Cref{sec: identify-DR}}
\begin{proof}[Proof for \Cref{thm: identification-DR}]
If conditions in \Cref{thm: identification1} hold and $h = h_0$ satisfies \Cref{eq: bridge-obs}, then 
\begin{align*}
&\Eb{h\prns{\Sc, \Sb, A, X} \mid A = a, G = E} + \Eb{q\prns{\Sb, \Sa, A, X}\prns{Y - h\prns{\Sc, \Sb, A, X}} \mid A = a, G = O} \\
=& \Eb{h\prns{\Sc, \Sb, A, X} \mid A = a, G = E} \\
 &\qquad +\Eb{q\prns{\Sb, \Sa, A, X}\Eb{{Y - h\prns{\Sc, \Sb, A, X}} \mid S_2, S_1, A, X, G = O} \mid A = a, G = O} \\
=& \Eb{h\prns{\Sc, \Sb, A, X} \mid A = a, G = E} \\
=& \mu(a),
\end{align*}
where the second equation follows from \Cref{eq: bridge-obs} and the third equation follows from \Cref{thm: identification1}.

If conditions in \Cref{thm: identification2} hold and $q = q_0$ satisfies \Cref{eq: bridge2-obs-1} or \Cref{eq: bridge2-obs-2}, then 
\begin{align*}
&\Eb{h\prns{\Sc, \Sb, A, X} \mid A = a, G = E} + \Eb{q\prns{\Sb, \Sa, A, X}\prns{Y - h\prns{\Sc, \Sb, A, X}} \mid A = a, G = O} \\
=& \Eb{q\prns{\Sb, \Sa, A, X}Y\mid A=a, G=O} \\
&- \Eb{h\prns{S_3, S_2, A, X}\Eb{{q\prns{\Sb, \Sa, A, X} - \frac{p\prns{\Sc, \Sb, X\mid A, G = E}}{p\prns{\Sc, \Sb, X \mid A, G = O}}}\mid S_3, S_2, A, X, G = O} \mid A  = a, G= O} \\
=& \Eb{q\prns{\Sb, \Sa, A, X}Y\mid A=a, G=O} \\
=& \mu(a),
\end{align*}
where the second equation follows from \Cref{eq: bridge2-obs-1} and the third equation follows from \Cref{thm: identification2}.
\end{proof}

\subsection{Proofs for \Cref{sec: estimation}}
\begin{proof}[Proof for \Cref{thm: consistency}]
We first prove statement \eqref{thm: consistency-q}. We define 
\begin{align*}
\tilde \mu_{\sel}\prns{a} = \frac{1}{K}\sum_{k=1}^K \bracks{\frac{1}{\nOka}\sum_{i \in \mathcal{D}_{O, k}}\indic{A_i = a}{\tilde q\prns{S_{2, i}, S_{1, i}, A_i, X_i}Y_i}}.
\end{align*}
Since we assume $\tilde q = q_0$, as $n\to\infty$, it follows from Law of Large Number and \Cref{thm: identification2} that 
\begin{align*}
\tilde \mu_{\sel}\prns{a}  \to \Eb{ q_0\prns{\Sb, \Sa, A, X}Y \mid A= a, G = O} = \mu(a).
\end{align*}
Now we only need to show that $\hat\mu_{\sel}\prns{a} - \tilde \mu_{\sel}\prns{a} = o_{\mathbb{P}}\prns{1}$, as this would imply that $\hat\mu_{\sel}\prns{a} = \mu(a) + o_{\mathbb{P}}\prns{1}$, so that $\hat\tau_{\sel}$ is a consistent estimator for $\tau$. To prove this, note that 
\begin{align*}
\hat\mu_{\sel}\prns{a} - \tilde \mu_{\sel}\prns{a}  = \frac{1}{K}\sum_{k=1}^K {\frac{n_{O, k}}{\nOka} \Delta_{\sel, k}}.
\end{align*}
where 
\begin{align*}
\Delta_{\sel, k} = \frac{1}{n_{O, k}}\sum_{i \in \mathcal{D}_{O, k}}\indic{A_i = a}{\prns{\hat q_k\prns{S_{2, i}, S_{1, i}, A_i, X_i} -  q_0\prns{S_{2, i}, S_{1, i}, A_i, X_i}}Y_i}.
\end{align*}
Then by Cauchy-Schwartz inequality, for any $k\in\braces{1, \dots, K}$, we have 
\begin{align*}
\op{Var}\prns{\Delta_{\sel, k} \mid \mathcal{D}_{O, -k}} 
    &= \frac{1}{n_{O, k}}\op{Var}\prns{\indic{A = a}{\prns{\hat q_k\prns{S_2, S_1, A, X} -  q_0\prns{\Sb, \Sa, A, X}}Y} \mid \mathcal{D}_{O, -k}} \\
    &\le \frac{1}{n_{O, k}}\Eb{\prns{\indic{A = a}{\prns{\hat q_k\prns{S_2, S_1, A, X} -  q_0\prns{\Sb, \Sa, A, X}}Y}}^2 \mid \mathcal{D}_{O, -k}} \\
    &\lesssim  \frac{1}{n_{O, k}} \|\qk - q_0\|^2_{\mathcal{L}_2\prns{\mathbb{P}}} \le \frac{\rho^2_{q, n}}{n_{O, k}}.
\end{align*}
By Markov inequality, we then have that 
$$\abs{\Delta_{\sel, k}} = \Eb{\abs{\Delta_{\sel, k}} \mid  \mathcal{D}_{O, -k} } + O_{\mathbb{P}}\prns{\frac{\rho_{q, n}}{\sqrt{n_{O, k}}}}. $$
Here 
\begin{align*}
\Eb{\abs{\Delta_{\sel, k}} \mid  \mathcal{D}_{O, -k} } \lesssim \|\qk - q_0\|_{\mathcal{L}_2\prns{\mathbb{P}}} = \rho_{q, n}. 
\end{align*}
Therefore 
\begin{align*}
\hat\mu_{\sel}\prns{a} - \tilde \mu_{\sel}\prns{a} 
    &= \frac{1}{K}\sum_{k=1}^K {\frac{n_{O, k}}{\nOka} \Delta_{\sel, k}} = \frac{1}{K}\sum_{k=1}^K \frac{1}{\Prb{A=a\mid G=O}}O_{\mathbb{P}}\prns{\rho_{q, n} + \frac{\rho_{q, n}}{\sqrt{n_{O, k}}}} \\
    &= o_{\mathbb{P}}\prns{1}.
\end{align*} 

Similarly, we can prove that $\hat\mu_{\out}\prns{a} = \mu\prns{a} + o_{\mathbb{P}}\prns{1}$ so that $\hat\tau_{\out}$ is a consistent estimator for $\tau$, \ie, statement \eqref{thm: consistency-h} is true. 

Finally, we can similarly prove that $\hat\mu_{\dr}\prns{a} - \tilde\mu_{\dr}\prns{a} = o_{\mathbb{P}}\prns{1}$, where 
\begin{align*}
 \tilde\mu_{\dr}\prns{a} &=\frac{1}{K}\sum_{k=1}^K \bracks{\frac{1}{\nEa}\sum_{i \in \mathcal{D}_E} \indic{A_i = a}{\tilde h\prns{S_{3,i}, S_{2, i}, A_i, X_i}}}\\
    &\phantom{\hat\mu_\dr(a)} + 
    \frac{1}{K}\sum_{k=1}^K \bracks{\frac{1}{\nOka}\sum_{i \in \mathcal{D}_{O, k}}\indic{A_i = a}{\tilde q\prns{S_{2, i}, S_{1, i}, A_i, X_i}\prns{Y_i - \tilde h\prns{S_{3, i}, S_{2, i}, A_i, X_i}}}}.
\end{align*}
By Law of Large Number, the limit of $ \tilde\mu_{\dr}\prns{a}$ is 
\begin{align*}
\Eb{\tilde h\prns{\Sc, \Sb, A, X} \mid A = a, G = E} + \Eb{\tilde q\prns{\Sb, \Sa, A, X}\prns{Y - \tilde h\prns{\Sc, \Sb, A, X}} \mid A = a, G = O}.
\end{align*}
According to \Cref{thm: identification-DR}, this is equal to $\mu(a)$ if either $\tilde q = q_0$ or $\tilde h = h_0$. Thus if either $\tilde q = q_0$ or $\tilde h = h_0$, $\hat\mu_{\dr}\prns{a} - \mu\prns{a} = o_{\mathbb{P}}\prns{1}$ so that $\hat\tau_{\dr}$ is a consistent estimator for $\tau$. This proves statement \eqref{thm: consistency-dr}.
\end{proof}

\begin{proof}[Proof for \Cref{thm: dist-dr}]
By simple algebra, we can show that 
\begin{align*}
\hat\mu_{\dr}\prns{a} - \tilde \mu_{\dr}\prns{a}  
    &= \frac{1}{K}\sum_{k=1}^K {\frac{n_{O, k}}{\nOka} \Delta^O_{\dr, k}} + \frac{n_E}{\nEa}\Delta^E_{\dr, k} \\
    &= \frac{1}{K}\sum_{k=1}^K \frac{1}{\Prb{A=a \mid G = O} + o_{\mathbb{P}}\prns{1}} \Delta^O_{\dr, k} + \frac{1}{\Prb{A=a\mid G=E} + o_{\mathbb{P}}\prns{1}}\Delta^E_{\dr, k}
\end{align*}
where 
\begin{align*}
\Delta_{\dr, k}^O  
    = \frac{1}{n_{O, k}}\sum_{i \in \mathcal{D}_{O, k}}\bigg[&\indic{A_i = a}{\qk\prns{S_{2, i}, S_{1, i}, A_i, X_i}\prns{Y_i - \hk\prns{S_{3, i}, S_{2, i}, A_i, X_i}}} \\
    - &\indic{A_i = a}{q_0\prns{S_{2, i}, S_{1, i}, A_i, X_i}\prns{Y_i - h_0\prns{S_{3, i}, S_{2, i}, A_i, X_i}}}\bigg], 
\end{align*}
and 
\begin{align*}
\Delta^E_{\dr, k} = \frac{1}{n_E}\sum_{i \in \mathcal{D}_E} \indic{A_i = a}\prns{\hk\prns{S_{3,i}, S_{2, i}, A_i, X_i} - h_0\prns{S_{3,i}, S_{2, i}, A_i, X_i}}. 
\end{align*}
By following the proof for \Cref{thm: consistency}, we can show that 
\begin{align*}
\Delta_{\dr, k}^O = \Eb{\Delta^O_{\dr, k} \mid  \mathcal{D}_{O, -k} } +  O_{\mathbb{P}}\prns{\frac{\max\{\rho_{q, n}, \rho_{h, n}\}}{\sqrt{n_{O, k}}}}
 = \Eb{\Delta^O_{\dr, k} \mid  \mathcal{D}_{O, -k} } + o_{\mathbb{P}}\prns{n^{-1/2}}
\end{align*}
and 
\begin{align*}
\Delta^E_{\dr, k} = \Eb{\Delta^E_{\dr, k} \mid  \mathcal{D}_{O, -k} } +  O_{\mathbb{P}}\prns{\frac{\rho_{h, n}}{\sqrt{n_{E}}}} = \Eb{\Delta^E_{\dr, k} \mid  \mathcal{D}_{O, -k} } + o_{\mathbb{P}}\prns{n^{-1/2}}.
\end{align*}
Moreover, we have   
\begin{align}
&\abs{\frac{1}{\Prb{A=a \mid G = O}}\Eb{\Delta^O_{\dr, k} \mid  \mathcal{D}_{O, -k} } + \frac{1}{\Prb{A=a \mid G = E}}\Eb{\Delta^E_{\dr, k} \mid  \mathcal{D}_{O, -k} }} \nonumber \\
=& \bigg|\Eb{\hk\prns{S_{3}, S_{2}, A, X} - h_0\prns{S_{3}, S_{2}, A, X} \mid A = a, G = E, \mathcal{D}_{O, -k}} \nonumber \\
+& \Eb{{\qk\prns{S_{2}, S_{1}, A, X}\prns{Y - \hk\prns{S_{3}, S_{2}, A, X}}} \mid A = a, G = O, \mathcal{D}_{O, -k}} \nonumber\\
-& \Eb{{q_0\prns{S_{2}, S_{1}, A, X}\prns{Y - h_0\prns{S_{3}, S_{2}, A, X}}} \mid A = a, G = O, \mathcal{D}_{O, -k}} \bigg| \nonumber \\
=& \abs{\mathcal{R}_{k, 1}  +\mathcal{R}_{k, 2} + \mathcal{R}_{k, 3}}. \label{eq: thm-dist-1}
\end{align}
Here 
\begin{align*}
&{\mathcal{R}_{k, 1}} = \Eb{q_0\prns{S_2, S_1, A, X}\prns{\hk\prns{S_{3}, S_{2}, A, X} - h_0\prns{S_{3}, S_{2}, A, X}} \mid A = a, G = O, \mathcal{D}_{O, -k}}  \\
&{\mathcal{R}_{k, 2}}= \Eb{\qk\prns{S_2, S_1, A, X}\prns{h_0\prns{S_{3}, S_{2}, A, X} - \hk\prns{S_{3}, S_{2}, A, X}} \mid A = a, G = O, \mathcal{D}_{O, -k}}  \\
&{\mathcal{R}_{k, 3}} = 0. 
\end{align*}
{
Thus
\begin{align*}
\text{\Cref{eq: thm-dist-1}} 
    &= \abs{{\mathcal{R}_{k, 1}}+{\mathcal{R}_{k, 2}}} \\
    &=\abs{\Eb{\prns{q_0 - \qk}({\hk - h_0}) \mid A = a, G = O, \mathcal{D}_{O, -k}}} \\
    &= \abs{\Eb{\frac{\indic{A=a}}{\Prb{A=a \mid G=O}}\prns{q_0 - \qk}({\hk - h_0}) \mid G = O, \mathcal{D}_{O, -k}}} 
\end{align*}
It follows that 
\begin{align*}
\text{\Cref{eq: thm-dist-1}} 
    &= \abs{\Eb{\frac{\indic{A=a}}{\Prb{A=a \mid G=O}}\Eb{{q_0 - \qk} \mid S_3, S_2, A, X, G = O}({\hk - h_0}) \mid \mathcal{D}_{O, -k}}} \\
    &\le \|P^\star\prns{\qk - q_0}\|_{\mathcal{L}_2\prns{\mathbb{P}}}\|\hk - h_0\|_{\mathcal{L}_2\prns{\mathbb{P}}},
\end{align*}
and 
\begin{align*}
\text{\Cref{eq: thm-dist-1}} 
    &= \abs{\Eb{\frac{\indic{A=a}}{\Prb{A=a \mid G=O}}\prns{q_0 - \qk}\Eb{{\hk - h_0} \mid S_2, S_1, A, X, G = O} \mid \mathcal{D}_{O, -k}}} \\
    &\le \|q_0 - \qk\|_{\mathcal{L}_2\prns{\mathbb{P}}}\|P({\hk - h_0})\|_{\mathcal{L}_2\prns{\mathbb{P}}}.
\end{align*}
This means that 
\begin{align*}
\text{\Cref{eq: thm-dist-1}}  \le \min\braces{\|P^\star\prns{\qk - q_0}\|_{\mathcal{L}_2\prns{\mathbb{P}}}\|\hk - h_0\|_{\mathcal{L}_2\prns{\mathbb{P}}}, \|q_0 - \qk\|_{\mathcal{L}_2\prns{\mathbb{P}}}\|P({\hk - h_0})\|_{\mathcal{L}_2\prns{\mathbb{P}}}} = o_{\mathbb{P}}\prns{n^{-1/2}}.
\end{align*}
Therefore, we have $\hat\mu_{\dr}\prns{a} - \tilde \mu_{\dr}\prns{a}  = o_{\mathbb{P}}\prns{n^{-1/2}}$.
} 

Furthermore, 
\begin{align*}
 &\tilde\mu_{\dr}\prns{a} - \mu(a) \\
    &\qquad =\frac{1}{K}\sum_{k=1}^K \bracks{\frac{1}{\nEa}\sum_{i \in \mathcal{D}_E} \indic{A_i = a}\prns{h_0\prns{S_{3,i}, S_{2, i}, A_i, X_i}-\mu(a)}}\\
    &\qquad + 
    \frac{1}{K}\sum_{k=1}^K \bracks{\frac{1}{\nOka}\sum_{i \in \mathcal{D}_{O, k}}\indic{A_i = a}{ q_0\prns{S_{2, i}, S_{1, i}, A_i, X_i}\prns{Y_i -  h_0\prns{S_{3, i}, S_{2, i}, A_i, X_i}}}} \\
    &\qquad =  \frac{1}{\Prb{A=a\mid G= E}n_E}\sum_{i\in\mathcal{D}_E} {\indic{A_i = a}\prns{h_0\prns{S_{3,i}, S_{2, i}, A_i, X_i} - \mu(a)}}\\
    &\qquad +  \frac{1}{\Prb{A=a\mid G=O}n_O}\sum_{i\in\mathcal{D}_O}\indic{A_i = a}{ q_0\prns{S_{2, i}, S_{1, i}, A_i, X_i}\prns{Y_i -  h_0\prns{S_{3, i}, S_{2, i}, A_i, X_i}}} + o_{\mathbb{P}}\prns{n^{-1/2}} 
\end{align*}

Combine the results above, we have 
\begin{align*}
\hat\tau_{\dr} - \tau 
    &= \frac{1}{n_E}\sum_{i \in \mathcal{D}_E}\bracks{\frac{A_i - \Prb{A_i = 1\mid G_i = E}}{\Prb{A_i = 1 \mid G_i = E}\prns{1-\Prb{A_i = 1\mid G_i = E}}}\prns{h_0\prns{S_{3, i}, S_{2, i}, A_i, X_i} - \mu(A_i)}} \\
    &+ \frac{1}{n_O}\sum_{i \in \mathcal{D}_O}\bracks{\frac{A_i - \Prb{A_i = 1\mid G_i = O}}{\Prb{A_i = 1 \mid G_i = O}}q_0\prns{S_{2, i}, S_{1, i}, A_i, X_i}\prns{Y_i -  h_0\prns{S_{3, i}, S_{2, i}, A_i, X_i}}} + o_{\mathbb{P}}\prns{n^{-1/2}}.
\end{align*}
Then the asserted conclusion follows from Central Limit Theorem. 
\end{proof}

\begin{proof}[Proof for \Cref{thm: CI}]
We only need to prove that $\hat\sigma^2$ is a consistent estimator for $\sigma^2$, since then we can apply Slutsky's theorem to show that  as $n\to\infty$,
\begin{align*}
\frac{\sqrt{n}\prns{\hat\tau_{\dr}-\tau}}{\hat\sigma}  \rightsquigarrow  \mathcal{N}\prns{0, 1}. 
\end{align*}
This in turn implies the desired asymptotic coverage conclusion.

To prove the consistency of $\hat\sigma^2$, we first consider the following (infeasible) estimator:
\begin{align*}
    \tilde\sigma^2 
        &=  \frac{n}{n_EK}\sum_{k=1}^K \braces{\frac{1}{n_E}\sum_{i \in \mathcal{D}_E} \bracks{\frac{A_i - \hat\pi_E}{\hat\pi_{E}}\prns{h_0\prns{S_{3,i}, S_{2, i}, A_i, X_i}-\hat\mu_\dr(A_i)}}^2} \nonumber \\
        &+ \frac{n}{n_OK}\sum_{k=1}^K \braces{\frac{1}{\nOka}\sum_{i \in \mathcal{D}_{O, k}}\bracks{\frac{A_i - \hat\pi_O}{\hat\pi_O}q_0(S_{2, i}, S_{1, i}, A_i, X_i)\prns{Y_i - h_0\prns{S_{3, i}, S_{2, i}, A_i, X_i}}}^2} \\
        &=  \frac{n}{n_E} \braces{\frac{1}{n_E}\sum_{i \in \mathcal{D}_E} \bracks{\frac{A_i - \hat\pi_E}{\hat\pi_{E}}\prns{h_0\prns{S_{3,i}, S_{2, i}, A_i, X_i}-\hat\mu_\dr(A_i)}}^2} \nonumber \\
        &+ \frac{n}{n_O} \braces{\frac{1}{n_O}\sum_{i \in \mathcal{D}_{O}}\bracks{\frac{A_i - \hat\pi_O}{\hat\pi_O}q_0(S_{2, i}, S_{1, i}, A_i, X_i)\prns{Y_i - h_0\prns{S_{3, i}, S_{2, i}, A_i, X_i}}}^2}.
\end{align*}
Since $n/n_E \to (1+\lambda)/\lambda$, $n/n_O \to 1+\lambda$, $\hat\pi_E \to \Prb{A=1\mid G=E}$, and  $\hat\pi_O \to \Prb{A=1\mid G=O}$, we can apply Law of Large Number and  Slutsky's theorem to show that $\tilde\sigma^2$ is a consistent estimator for $\sigma^2$. Therefore, as long as we can prove that $\hat\sigma^2 - \tilde\sigma^2 \to 0$ as $n \to \infty$, we have $\hat\sigma^2 \to \sigma^2$ as $n \to \infty$, which finishes our proof. 

To prove $\hat\sigma^2 - \tilde\sigma^2 \to 0$, we define that 
\begin{align*}
    \psi_{1, i}(h) 
        &= \frac{A_i - \hat\pi_E}{\hat\pi_{E}}\prns{h\prns{S_{3,i}, S_{2, i}, A_i, X_i}-\hat\mu_\dr(A_i)}, \\
    \psi_{2, i}(h, q) 
        &= \frac{A_i - \hat\pi_O}{\hat\pi_O}q(S_{2, i}, S_{1, i}, A_i, X_i)\prns{Y_i - h\prns{S_{3, i}, S_{2, i}, A_i, X_i}}.
\end{align*}
It follows that 
\begin{align*}
    \abs{\hat\sigma^2 - \tilde\sigma^2} 
        &= \frac{n}{n_EK}\sum_{k=1}^K\underbrace{\abs{ {\frac{1}{n_E}\sum_{i \in \mathcal{D}_E} \bracks{\psi^2_{1, i}(\hk) - \psi_{1, i}^2(h_0)}}}}_{\Delta_{1, k}} \\
        &+ \frac{n}{n_OK}\sum_{k=1}^K\underbrace{\abs{ {\frac{1}{\nOka}\sum_{i \in \mathcal{D}_{O, k}}\bracks{\psi^2_{2, i}(\hk, \qk) - \psi^2_{2, i}(h_0, q_0)}}}}_{\Delta_{2, k}}.
\end{align*}
We now analyze $\Delta_{1, k}$:
\begin{align*}
    \Delta_{1, k} 
        &\le \abs{\frac{1}{n_E}\sum_{i \in \mathcal{D}_E}\prns{\psi_{1, i}(\hk) - \psi_{1, i}(h_0)}\prns{2\psi_{1, i}(h_0) + \psi_{1, i}(\hk) - \psi_{1, i}(h_0)}} \\
        &\le \bracks{\frac{1}{n_E}\sum_{i \in \mathcal{D}_E}\prns{\psi_{1, i}(\hk) - \psi_{1, i}(h_0)}^2}^{1/2}\braces{\bracks{\frac{1}{n_E}\sum_{i \in \mathcal{D}_E}\prns{\psi_{1, i}(\hk) - \psi_{1, i}(h_0)}^2}^{1/2} + 2\bracks{\frac{1}{n_E}\sum_{i \in \mathcal{D}_E}{\psi^2_{1, i}(h_0)}}^{1/2}}.
\end{align*}
Moreover, since $\Prb{A=1\mid G=E}$ is  strictly positive according to \Cref{assump: unconfound-exp}, we have that for large enough $n$, $\hat\pi_E \ge \Prb{A=1\mid G=E}/2 > 0$ with high probability. It follows that 
\begin{align*}
    \frac{1}{n_E}\sum_{i \in \mathcal{D}_E}\prns{\psi_{1, i}(\hk) - \psi_{1, i}(h_0)}^2 
        &\lesssim \frac{1}{n_E}\sum_{i \in \mathcal{D}_E}\prns{\hk\prns{S_{3,i}, S_{2, i}, A_i, X_i} - h_0\prns{S_{3,i}, S_{2, i}, A_i, X_i}}^2 \\
        &= \|\hk -  h_0 \|_{\mathcal{L}_2\prns{\mathbb{P}}} + o_{\mathbb{P}}(1) = O_{\mathbb{P}}\prns{\rho_{h, n}} + o_{\mathbb{P}}(1) = o_{\mathbb{P}}(1).
\end{align*}
It follows that $\Delta_{1, k} = o_{\mathbb{P}}(1)$. Similarly, we can show that $\Delta_{2, k} = o_{\mathbb{P}}(1)$. These together ensure that as $n\to\infty$,
\begin{align*}
    \hat\sigma^2 - \tilde\sigma^2 \to 0.
\end{align*}
\end{proof}

\begin{proof}[Proof for \cref{thm: efficiency}]
{We consider a semiparametric model $\mathcal{M}_{\op{sp}}$ that places no restrictions on the data distribution except the existence of a bridge function $h_0$ in \Cref{assump: bridge}. 
Consider a regular parametric submodel indexed by a parameter $t$: $\mathcal{P}_t = \braces{p_t\prns{y, s, a, x, g}: t \in \Rl}$  
where $p_0(y, s, a, x, g)$ equals the true density $p(y, s, a, x, g)$. The associated score function is denoted as $\score(y,s, a, x, g) = \partial_t \log p_t(y, s, a, x, g) \vert_{t = 0}$. The expectation w.r.t the distribution $p_t(y, s, a, x, g)$ is denoted by $\E_t$.}

{By following the proof for Theorem 11 in \cite{kallus2021causal}, under the condition that bridge functions $h_0, q_0$ uniquely exist and the linear operator $T$ is bijective, the tangent space corresponding to $\mathcal{M}_{\op{sp}}$ is given by 
\begin{align}\label{eq: tangent}
\mathcal{S} = \bigg\{
    &\score\prns{Y, S, A, X, G} 
        = \score\prns{S_2, S_1, A, X, G} + \score\prns{Y, S_3 \mid S_2, S_1, A, X, G}: \\
    &\score\prns{S_2, S_1, A, X, G} \in  L_2\prns{S_2, S_1, A, X, G}, ~ \score\prns{Y, S_3 \mid S_2, S_1, A, X, G} \in  L_2\prns{Y, S_3 \mid S_2, S_1, A, X, G}, \nonumber \\
    &\Eb{\score\prns{S_2, S_1, A, X, G} } = 0, \Eb{\score\prns{Y, S_3 \mid S_2, S_1, A, X, G} \mid S_2, S_1, A, X, G} = 0, \nonumber \\
    &\Eb{\prns{Y - h_0(S_3, S_2, A, X)\score\prns{Y, S_3 \mid S_2, S_1, A, X, G}}\mid S_2, S_1, A, X, G=O} \in \text{Range}(T)\bigg\}. \nonumber
\end{align}}

We now analyze the path differentiability of the counterfactual mean parameter $\mu_t\prns{a}$ under a submodel distribution with parameter value $t$. According to \Cref{thm: identification1}, we have 
\begin{align*}
 \mu_t\prns{a} = \E_t\bracks{h_t\prns{S_3, S_2, A, X} \mid A = a, G = E},
 \end{align*} 
 where $h_t\prns{S_3, S_2, A, X}$ is the corresponding outcome bridge function defined by 
\begin{align*}
\mathbb{E}_t\bracks{Y - h_t\prns{S_3, S_2, A, X} \mid S_2, S_1, A, X, G = O} = 0. 
\end{align*}

Note that we have 
\begin{align}
 \frac{\partial}{\partial t}\mu_t\prns{a}\vert_{t = 0} 
    &=  \frac{\partial}{\partial t} \mathbb{E}_t\bracks{h_t\prns{S_3, S_2, A, X} \mid A = a, G = E}\vert_{t = 0} \nonumber \\
    &= \mathbb{E}\bracks{h_0\prns{S_3, S_2, A, X}\score\prns{S_3, S_2 , X\mid A, G} \mid A = a, G = E} \label{eq: diff-1}\\
    &+ \frac{\partial}{\partial t}\Eb{h_t\prns{S_3, S_2, A, X} \mid A = a, G = E}\vert_{t = 0}.\label{eq: diff-2}
 \end{align} 

 We first analyze the term in \Cref{eq: diff-1}.
\begin{align}
 &\mathbb{E}\bracks{h_0\prns{S_3, S_2, A, X}\score\prns{S_3, S_2, X \mid A, G} \mid A = a, G = E}  \label{eq: efficiency-term1} \\
    =& \mathbb{E}\bracks{\prns{h_0\prns{S_3, S_2, A, X} - \mu\prns{a}}\score\prns{S_3, S_2, X\mid A, G} \mid A = a, G = E} \nonumber  \\
    =& \mathbb{E}\bracks{\prns{h_0\prns{S_3, S_2, A, X} - \mu\prns{a}}\score\prns{S_3, S_2, A, X, G} \mid A = a, G = E} \nonumber \\
    =& \mathbb{E}\bracks{\prns{h_0\prns{S_3, S_2, A, X} - \mu\prns{a}}\score\prns{Y, S_3, S_2, S_1, A, X, G} \mid A = a, G = E} \nonumber \\
    =& \mathbb{E}\bracks{\frac{\indic{A = a, G = E}}{\Prb{A= a, G = E}}\prns{h_0\prns{S_3, S_2, A, X} - \mu\prns{a}}\score\prns{Y, S_3, S_2, S_1, A, X, G}} \nonumber
 \end{align}
 where the second equation holds because 
 \begin{align*}
 &\Eb{\prns{h_0\prns{S_3, S_2, A, X} - \mu\prns{a}}\score\prns{A, G} \mid A = a, G = E} \\
 =&   \Eb{\prns{h_0\prns{S_3, S_2, A, X} - \mu\prns{a}} \mid A = a, G = E} \score\prns{A=a, G=E} = 0,
 \end{align*}
 and the third equation holds because 
 \begin{align*}
 &\Eb{\prns{h_0\prns{S_3, S_2, A, X} - \mu\prns{a}}\score\prns{Y, S_1 \mid S_3, S_2, A, X, G} \mid A = a, G = E} \\
 =&  \Eb{\prns{h_0\prns{S_3, S_2, A, X} - \mu\prns{a}}\Eb{\score\prns{Y, S_1 \mid S_3, S_2, A, X, G} \mid S_3, S_2, A, X, G} \mid A = a, G = E} = 0.
 \end{align*}

Next we analyze the term in \Cref{eq: diff-2}.
\begin{align*}
 &\frac{\partial}{\partial t}\Eb{h_t\prns{S_3, S_2, A, X} \mid A = a, G = E}\vert_{t = 0} \\
    &= \frac{\partial}{\partial t}\Eb{\frac{p\prns{S_3, S_2, X \mid A, G = E}}{p\prns{S_3, S_2, X \mid A, G = O}}  h_t\prns{S_3, S_2, A, X} \mid A = a, G = O}\vert_{t = 0} \\
    &= \frac{\partial}{\partial t}\Eb{q_0\prns{S_2, S_1, A, X}h_t\prns{S_3, S_2, A, X} \mid A = a, G = O}\vert_{t = 0} \\
    &= \Eb{q_0\prns{S_2, S_1, A, X}\frac{\partial}{\partial t} \Eb{h_t\prns{S_3, S_2, A, X} \mid S_2, S_1, A, X, G = O}\vert_{t = 0} \mid A = a, G = O},
 \end{align*}
 where the second equation holds because of \Cref{eq: bridge2-obs-1}. 

 Furthermore, by taking the derivative of the left hand side w.r.t $t$ at $t = 0$, we have 
\begin{align}\label{eq: path-diff-1}
&\frac{\partial}{\partial t}\Eb{h_t\prns{S_3, S_2, A, X} \mid S_2, S_1, A, X, G = O}\vert_{t = 0} \nonumber \\
=& \Eb{\prns{Y - h_0\prns{S_3, S_2, A, X}}\score\prns{Y, S_3 \mid S_2, S_1, A,X,  G} \mid S_2, S_1, A,X, G = O} = 0.
\end{align} 
It follows that 
 \begin{align}
 &\frac{\partial}{\partial t}\Eb{h_t\prns{S_3, S_2, A, X} \mid A = a, G = E}\vert_{t = 0} \label{eq: efficiency-term2} \\
    &= \Eb{q_0\prns{S_2, S_1, A, X}{\prns{Y - h_0\prns{S_3, S_2, A, X}}}\score\prns{Y, S_3 \mid S_2, S_1, A, X, G}\mid A = a, G = O} \nonumber \\
    &= \Eb{q_0\prns{S_2, S_1, A, X}{\prns{Y - h_0\prns{S_3, S_2, A, X}}}\score\prns{Y, S_3, S_2, S_1, A, X, G}\mid A = a, G = O} \nonumber \\
    &= \Eb{\frac{\indic{A = a, G = O}}{\Prb{A = a, G = O}}q_0\prns{S_2, S_1, A, X}{\prns{Y - h_0\prns{S_3, S_2, A, X}}}\score\prns{Y, S_3, S_2, S_1, A, X, G}}, \nonumber 
 \end{align}
 where the second equation holds because 
  \begin{align*}
 &\Eb{q_0\prns{S_2, S_1, A, X}{\prns{Y - h_0\prns{S_3, S_2, A, X}}}\score\prns{S_2, S_1, A, X, G}\mid A = a, G = O} \\
 =& \E\big[q_0\prns{S_2, S_1, A, X}\Eb{Y - h_0\prns{S_3, S_2, A, X} \mid S_2, S_1, A, X, G = O} \\
 &\qquad\qquad\qquad\qquad\qquad\qquad  \times \score\prns{S_2, S_1, A,X,  G = O}\mid A = a, G = O\big] = 0.
 \end{align*}

 Combining \Cref{eq: efficiency-term1,eq: efficiency-term2}, we have 
  \begin{align*}
   \frac{\partial}{\partial t}\mu_t\prns{a}\vert_{t = 0} = \Eb{\psi_a\prns{Y, S_3, S_2, S_1, A, X, G}\score\prns{Y, S_3, S_2, S_1, A, X, G}},
  \end{align*} 
where 
\begin{align*}
\psi_a\prns{Y, S_3, S_2, S_1, A, X, G} 
    &= \frac{\indic{A = a, G = E}}{\Prb{A = a, G = E}}\prns{h_0\prns{S_3, S_2, A, X} - \mu\prns{a}} \nonumber \\
    & + \frac{\indic{A = a, G = O}}{\Prb{A = a, G = O}}q_0\prns{S_2, S_1, A, X}\prns{Y - h_0\prns{S_3, S_2, A, X}}.  
\end{align*}
Therefore, 
\begin{align*}
\frac{\partial}{\partial t}\tau_t\vert_{t = 0}  &= \frac{\partial}{\partial t}\mu_t\prns{1}\vert_{t = 0} - \frac{\partial}{\partial t}\mu_t\prns{0}\vert_{t = 0} \\
  &= \Eb{\psi\prns{Y, S_3, S_2, S_1, A, X, G}\score\prns{Y, S_3, S_2, S_1, A, X, G}},
\end{align*}
where 
\begin{align*}
&\psi\prns{Y, S_3, S_2, S_1, A, X, G}  \\
    =& \psi_1\prns{Y, S_3, S_2, S_1, A, X, G} - \psi_0\prns{Y, S_3, S_2, S_1, A, X, G} - \tau \\
    =& \frac{\indic{G = E}}{\Prb{G = E}}\frac{A - \Prb{A = 1 \mid G = E}}{\Prb{A = 1 \mid G = E}}\prns{h_0\prns{S_3, S_2, A, X} - \mu(A)} \\
    +& \frac{\indic{G = O}}{\Prb{G = O}}\frac{A - \Prb{A=1\mid G=O}}{\Prb{A=1\mid G=O}}q_0\prns{S_2, S_1, A, X}\prns{Y - h_0\prns{S_3, S_2, A, X}} - \tau.
\end{align*}

{
We can easily decompose  $\psi\prns{Y, S_3, S_2, S_1, A, X, G}$ into two terms:
\begin{align*}
\psi\prns{Y, S_3, S_2, S_1, A, X, G} 
&= \Eb{\psi\prns{Y, S_3, S_2, S_1, A, X, G} \mid S_2, S_1, A, X, G} - \tau \\
&+ \psi\prns{Y, S_3, S_2, S_1, A, X, G} - \Eb{\psi\prns{Y, S_3, S_2, S_1, A, X, G} \mid S_2, S_1, A, X, G},
\end{align*}
where 
\begin{align*}
&\Eb{\psi\prns{Y, S_3, S_2, S_1, A, X, G} \mid S_2, S_1, A, X, G} - \tau \in L_2(S_2, S_1, A, X, G) \\
&\Eb{\Eb{\psi\prns{Y, S_3, S_2, S_1, A, X, G} \mid S_2, S_1, A, X, G} - \tau} = 0 \\
& \psi\prns{Y, S_3, S_2, S_1, A, X, G} - \Eb{\psi\prns{Y, S_3, S_2, S_1, A, X, G} \mid S_2, S_1, A, X, G} \in L_2(Y, S_3 \mid S_2, S_1, A, X, G) \\
& \Eb{\psi\prns{Y, S_3, S_2, S_1, A, X, G} - \Eb{\psi\prns{Y, S_3, S_2, S_1, A, X, G} \mid S_2, S_1, A, X, G} \mid S_2, S_1, A, X, G} = 0.
\end{align*}
Moreover, since $T$ is surjective, its range space $\text{Range}(T)$ is the whole $L_2(S_2, S_1, A, X)$ space so we automatically have 
\begin{align*}
 &\mathbb{E}\big[\prns{\psi\prns{Y, S_3, S_2, S_1, A, X, G} - \Eb{\psi\prns{Y, S_3, S_2, S_1, A, X, G} \mid S_2, S_1, A, X, G}}\\
 &\qquad\qquad\qquad\qquad \times \prns{Y - h_0(S_3, S_2, A, X)} \mid S_2, S_1, A, X, G = O\big] \in \text{Range}(T). 
 \end{align*} 
 This means that $\psi\prns{Y, S_3, S_2, S_1, A, X, G}$ belongs to the tangent space $\mathcal{S}$. 
Thus $\psi\prns{Y, S_3, S_2, S_1, A, X, G}$ is the efficient influence function for $\tau$, and its variance, which is equal to $\sigma^2$ in \Cref{thm: efficiency}, is the semiparametric efficiency lower bound for $\tau$ relative to the tangent space $\mathcal{S}$ in \Cref{eq: tangent}.
} 
\end{proof}

\subsection{Proofs for \Cref{sec: extension}}

\begin{proof}[Proof for \Cref{corollary: covariate-exp-dr}]
Before proving the theorem, we note that by Bayes rule, we can easily verify that 
\begin{align*}
\frac{\Prb{G=E\mid A=a}\Prb{G=O\mid X}}{\Prb{G=O\mid A=a}\Prb{G=E\mid X}}\frac{\indic{A=a}}{\Prb{A=a\mid X, G=E}} = \frac{\indic{A=a}}{\Prb{A=a \mid X, G=O}} \frac{p(X \mid A=a, G=O)}{p(X \mid A=a, G=E)}.
\end{align*}

\paragraph*{Assume that condition 1 holds so we have $h=h_0$ satisfying \Cref{eq: bridge-obs}.} In this case, for any function $q$, we have 
\begin{align}\label{eq: X-dr-2a}
& \Eb{\frac{\Prb{G=E\mid A=a}\Prb{G=O\mid X}}{\Prb{G=O\mid A=a}\Prb{G=E\mid X}}\frac{\indic{A=a}}{\Prb{A=a\mid X, G=E}}q\left(S_{2}, S_{1}, A, X\right)\left(Y-h\left(S_{3}, S_{2}, A, X\right)\right)\mid G=O} \nonumber\\
=& \mathbb{E}\bigg[\frac{\Prb{G=E\mid A=a}\Prb{G=O\mid X}}{\Prb{G=O\mid A=a}\Prb{G=E\mid X}}\frac{\indic{A=a}}{\Prb{A=a\mid X, G=E}}q\left(S_{2}, S_{1}, A, X\right)\nonumber \\
&\qquad\qquad\qquad\qquad\qquad\qquad \times \Eb{Y-h_0\left(S_{3}, S_{2}, A, X\right) \mid S_2, S_1, A, X, G = O}\mid G=O\bigg] = 0,
\end{align}
where the last equation uses the conditional moment equation in \Cref{eq: bridge-obs}.

Moreover, for function $h = h_0$, 
\begin{align}\label{eq: X-dr-1a}
&\Eb{\frac{\Prb{G=E}\Prb{G=O\mid X}}{\Prb{G=O}\Prb{G=E\mid X}}\frac{\indic{A=a}}{\Prb{A=a\mid X, G=E}}\prns{h\prns{\Sc, \Sb, A, X} - \bar{h}_E\prns{A, X} }\mid G=E} \nonumber \\
=& \Eb{\frac{\Prb{G=E}\Prb{G=O\mid X}}{\Prb{G=O}\Prb{G=E\mid X}}\Eb{h\prns{\Sc, \Sb, a, X} - \bar{h}_E\prns{a, X} \mid A= a, X, G=E }\mid G=E} = 0.
\end{align}

Finally, we only need to prove that 
\begin{align}\label{eq: mu-out-bridge}
\mu(a) = \Eb{\Eb{h_0\prns{\Sc, \Sb, A, X} \mid A = a, X = x, G = E} \mid G= O}
\end{align}
According to \cref{lemma: bridge-obs}, we already know that any function $h_0\prns{S_3, S_2, A, X}$ that satisfies \Cref{eq: bridge-obs} must be a valid bridge function in the sense of \Cref{eq: bridge-U}. Thus we only need to prove \Cref{eq: mu-out-bridge} for $h_0\prns{S_3, S_2, A, X}$ that satisfies \Cref{eq: bridge-U}. By following the proof in \Cref{thm: identification1}, we can show that 
\begin{align*}
\Eb{h_0\prns{S_3, S_2, A, X} \mid A = a, X, G = E} 
    &= \Eb{\Eb{Y(a) \mid S_2(a), U, X, G = O} \mid A = a, X, G = E}.
\end{align*}
Therefore, 
\begin{align*}
&\Eb{\Eb{h_0\prns{S_3, S_2, A, X} \mid A = a, X, G = E}\mid G = O} \\
    =&  \Eb{\Eb{\Eb{Y(a) \mid S_2(a), U, X, G = O} \mid X, G = E} \mid G = O} \\
    =& \Eb{\Eb{\Eb{Y(a) \mid S_2(a), U, X, G = O} \mid X, G = O} \mid G = O}\\
    =& \Eb{\Eb{Y\prns{a} \mid X, G = O} \mid G = O} = \Eb{Y\prns{a} \mid G = O}.
\end{align*}
Here the first equation follows from the fact that $A \perp \prns{S(a), U} \mid X, G = E$ in \Cref{assump: unconfound-exp2}, the second equation follows from  \Cref{eq: data-overlap2} in \Cref{lemma: data-overlap}, and the third equation follows from the fact that $G \perp \prns{S(a), U} \mid X$ in \Cref{assump: ext-valid2}. 

Combining \Cref{eq: X-dr-2a,eq: X-dr-1a,eq: mu-out-bridge} proves the conclusion. 

\paragraph*{Assume that condition 2 holds so we have $q = q_0$ satisfying \Cref{eq: bridge2-obs-1} or \Cref{eq: bridge2-obs-2}.} We first prove that 
\begin{align}\label{eq: mu-sel-bridge}
\mu(a) 
   &= \Eb{\frac{\indic{A=a}}{\Prb{A=a \mid X, G=O}} \frac{p(X \mid A=a, G=O)}{p(X \mid A=a, G=E)}q_0\left(S_{2}, S_{1}, A, X\right)Y \mid G = O} \nonumber \\
   &= \Eb{\frac{\Prb{G=E\mid A=a}\Prb{G=O\mid X}}{\Prb{G=O\mid A=a}\Prb{G=E\mid X}}\frac{\indic{A=a}}{\Prb{A=a\mid X, G=E}}q\left(S_{2}, S_{1}, A, X\right)Y\mid G=O}.
\end{align}
To prove this, note that according to \Cref{lemma: bridge2-obs}, any function $q_0\prns{S_2, S_1, A, X}$ that satisfies \Cref{eq: bridge2-obs-1} or \Cref{eq: bridge2-obs-2} is a valid selection bridge function in the sense of \Cref{eq: bridge2-U}. Thus we only need to prove \Cref{eq: mu-sel-bridge} for any $q_0$ that satisfies \Cref{eq: bridge2-U}. We further note that the right hand side of \Cref{eq: mu-sel-bridge} is equal to the following:
\begin{align*}
&\Eb{\Eb{\frac{p(X \mid A=a, G=O)}{p(X \mid A=a, G=E)}q_0\left(S_{2}, S_{1}, A, X\right)Y \mid A = a, X, G = O}\mid G = O} \nonumber \\
=& \mathbb{E}\bigg[\mathbb{E}\bigg[\mathbb{E}\bigg[\frac{p(X \mid A=a, G=O)}{p(X \mid A=a, G=E)}q_0\left(S_{2}, S_{1}, A, X\right)\mid S_2, A =a, U, X, G = O]  \nonumber\\
&\qquad\qquad\qquad \times \Eb{Y \mid S_2, A =a, U, X, G = O} \mid A = a, X, G = O\bigg]\mid G = O\bigg] \nonumber\\
&= \mathbb{E}\bigg[\mathbb{E}\bigg[\frac{p\prns{\Sb, U \mid A, X, G = E}}{p\prns{\Sb, U \mid A, X, G = O}}\Eb{Y(a) \mid S_2(a), U, X, G = O} \mid A = a, X, G = O\bigg]\mid G = O\bigg] \nonumber\\
&= \mathbb{E}\bigg[\mathbb{E}\bigg[\Eb{Y(a) \mid S_2(a), U, X, G = O} \mid A = a, X, G = E\bigg]\mid G = O\bigg] \nonumber\\
&= \mathbb{E}\bigg[\mathbb{E}\bigg[\Eb{Y(a) \mid S_2(a), U, X, G = O} \mid X, G = E\bigg]\mid G = O\bigg] \nonumber\\
&= \mathbb{E}\bigg[\mathbb{E}\bigg[\Eb{Y(a) \mid S_2(a), U, X, G = O} \mid X, G = O\bigg]\mid G = O\bigg] \nonumber \\
&= \Eb{Y(a) \mid G = O} = \mu(a).
\end{align*}
Here the first equation uses $Y \perp S_1 \mid S_2, A, U, X, G = O$ which we prove in \Cref{cor: implication2}, the second equation uses the fact that $q_0$ satisfies \Cref{eq: bridge2-U} and $Y(a) \perp A \mid S_2(a), U, X, G = O$ according to \Cref{assump: unconfound-obs}, the fourth equation uses that $S_2(a) \perp A \mid X, G = E$ according to \Cref{assump: unconfound-exp2}, the fifth equation uses the fact that $S_2(a) \perp G \mid X$ according to \Cref{assump: ext-valid2}.

Next, we can follow the proof above to show that for any $h$, 
\begin{align}
&\Eb{\frac{\indic{A=a}}{\Prb{A=a \mid X, G=O}} \frac{p(X \mid A=a, G=O)}{p(X \mid A=a, G=E)}q_0\left(S_{2}, S_{1}, A, X\right)h\left(S_{3}, S_{2}, A, X\right)\mid G=O} \nonumber \\
=& \Eb{h\prns{S_3\prns{a}, S_2\prns{a}, a, X} \mid G = O} \label{eq: eliminate-1}
\end{align}
And by change of measure, we can also verify that 
\begin{align}
&\Eb{\frac{\Prb{G=E}\Prb{G=O\mid X}}{\Prb{G=O}\Prb{G=E\mid X}}\frac{\indic{A=a}}{\Prb{A=a\mid X, G=E}}{h\prns{\Sc, \Sb, A, X}}\mid G=E} \nonumber \\
=& \Eb{\Eb{h\prns{\Sc, \Sb, A, X} \mid A = a, X, G = E} \mid G = O} = \Eb{h\prns{S_3\prns{a}, S_2\prns{a}, a, X} \mid G = O}, \label{eq: eliminate-2}
\end{align}
and 
\begin{align*}
&\Eb{\frac{\Prb{G=E}\Prb{G=O\mid X}}{\Prb{G=O}\Prb{G=E\mid X}}\frac{\indic{A=a}}{\Prb{A=a\mid X, G=E}} \bar{h}_E\prns{A, X}\mid G=E} \\
=&\Eb{\Eb{\bar{h}_E\prns{A, X} \mid A = a, X, G = E} \mid G = O} = \Eb{\bar{h}_E\prns{A, X}\mid G=O}
\end{align*}
These show that 
\begin{align}\label{eq: X-dr-2}
0 =& \Eb{h_{E}\prns{a, X} \mid G= O}  \nonumber \\
    +&\Eb{\frac{\Prb{G=E}\Prb{G=O\mid X}}{\Prb{G=O}\Prb{G=E\mid X}}\frac{\indic{A=a}}{\Prb{A=a\mid X, G=E}}\prns{h\prns{\Sc, \Sb, A, X} - \bar{h}_E\prns{A, X} }\mid G=E} \nonumber \\
   -&\Eb{\frac{\indic{A=a}}{\Prb{A=a \mid X, G=O}} \frac{p(X \mid A=a, G=O)}{p(X \mid A=a, G=E)}q_0\left(S_{2}, S_{1}, A, X\right)h\left(S_{3}, S_{2}, A, X\right)\mid G=O}.
\end{align}

Combining \Cref{eq: mu-sel-bridge,eq: X-dr-2} leads to the conclusion. 
\end{proof}

\subsection{Proofs for Appendix}
\begin{proof}[Proof for \Cref{prop: sel-linear}]
First note that 
\begin{align*}
\frac{p\prns{\Sb, U, X\mid A =a , G = E}}{p\prns{\Sb, U, X \mid A = a, G = O}} 
    &= 
\frac{p\prns{ U, X\mid A = a, G = E}}{p\prns{U, X \mid A = a, G = O}}  \\
    &= \frac{\Prb{A=a\mid U, X, G = E}}{\Prb{A=a\mid U, X, G = O}}
    \frac{\Prb{A=a\mid G = O}}{\Prb{A=a\mid G = E}} \\
    &= \frac{\Prb{A=a\mid G = O}}{\Prb{A=a\mid U, X, G = O}},
\end{align*}
where the first equation follows from \Cref{lemma: data-overlap}, the second equation follows from Bayes rule, and the third equation follows from the fact that $\Prb{A=a\mid U, X, G = E} = \Prb{A=a\mid G = E} = \frac{1}{2}$.
Therefore, we have 
\begin{align}\label{eq: density-ratio}
\frac{p\prns{\Sb, U, X\mid A =a , G = E}}{p\prns{\Sb, U, X \mid A = a, G = O}} = 
\frac{\Eb{\bracks{1+\exp\prns{(-1)^a\prns{\kappa_1^\top U + \kappa_2^\top X}}}^{-1}}}{\bracks{1+\exp\prns{(-1)^a\prns{\kappa_1^\top U + \kappa_2^\top X}}}^{-1}}.
\end{align}

Second, $(S_1, S_2) \mid A, U, X, G = O$ follows a joint Gaussian distribution whose conditional expectation is 
\begin{align*}
\begin{bmatrix}
\tau_1 A + \beta_1 X + \gamma_1 U \\
\prns{\tau_2 + \alpha_2 \tau_1}A + \prns{\beta_2 + \alpha_2\beta_1}X + \prns{\gamma_2 + \alpha_2\gamma_1 }U
\end{bmatrix}
\end{align*}
and conditional covariance matrix is 
\begin{align*}
\begin{bmatrix}
\sigma_1^2 I_1 & \sigma_1^2 \alpha_2^\top \\
\sigma_1^2 \alpha_2 & \sigma_1^2 \alpha_2 \alpha_2^\top + \sigma_2^2 I_2
\end{bmatrix}.
\end{align*}
It follows that $S_1 \mid S_2, A, U, X, G = O$ also has a Gaussian distribution function with conditional expectation
\begin{align*}
\lambda_1 S_2 + \lambda_2 A + \lambda_3 X + \lambda_4 U
\end{align*}
and conditional variance 
\begin{align*}
\Sigma_{1\mid 2} = \sigma_1^2 I_1  -\sigma_1^4 \alpha_2^\top\prns{\sigma_1^2 \alpha_2\alpha_2^\top + \sigma_2^2 I_2}^{-1}\alpha_2.
\end{align*}
where 
\begin{align*}
&\lambda_1 = \sigma_1^2\alpha_2^\top\prns{\sigma_1^2 \alpha_2\alpha_2^\top + \sigma_2^2 I_2}^{-1}, \\
&\lambda_2 = \prns{I_1 - \sigma_1^2 \alpha_2^\top\prns{\sigma_1^2 \alpha_2\alpha_2^\top + \sigma_2^2 I_2}^{-1}\alpha_2}\tau_1 - \sigma_1^2\alpha_2^\top\prns{\sigma_1^2 \alpha_2\alpha_2^\top + \sigma_2^2 I_2}^{-1}\tau_2 \\
&\lambda_3 = \prns{I_1 - \sigma_1^2 \alpha_2^\top\prns{\sigma_1^2 \alpha_2\alpha_2^\top + \sigma_2^2 I_2}^{-1}\alpha_2}\beta_1 - \sigma_1^2\alpha_2^\top\prns{\sigma_1^2 \alpha_2\alpha_2^\top + \sigma_2^2 I_2}^{-1}\beta_2 \\
&\lambda_4 = \prns{I_1 - \sigma_1^2 \alpha_2^\top\prns{\sigma_1^2 \alpha_2\alpha_2^\top + \sigma_2^2 I_2}^{-1}\alpha_2}\gamma_1 - \sigma_1^2\alpha_2^\top\prns{\sigma_1^2 \alpha_2\alpha_2^\top + \sigma_2^2 I_2}^{-1}\gamma_2.
\end{align*}

Third, for $a = 1$, we posit a selection bridge function of the following form:
$$q_0\prns{S_2, S_1, 1, X} = c_1\exp\prns{\tilde\theta_2^\top S_2 + \tilde\theta_1^\top S_1 + \tilde \theta_0^\top X} + c_0.$$
It follows that 
\begin{align*}
&\Eb{q_0\prns{S_2, S_1, 1, X} \mid S_2, A= 1, U, X, G=O} \\
=& c_1\exp\prns{\tilde\theta_2^\top S_2  + \tilde \theta_0^\top X}\Eb{\exp\prns{\tilde\theta_1^\top S_1} \mid S_2, A= 1, U, X, G=O} + c_0 \\
=& c_1\exp\prns{\tilde\theta_2^\top S_2  + \tilde \theta_0^\top X}\exp\prns{\tilde\theta_1^\top \prns{\lambda_1 S_2 + \lambda_2 A + \lambda_3 X + \lambda_4 U} + \frac{1}{2}\tilde\theta_1^\top \Sigma_{1\mid 2}\tilde\theta_1} + c_0 \\
=& c_1 \exp\prns{\frac{1}{2}\tilde\theta_1^\top \Sigma_{1\mid 2}\tilde\theta_1}\exp\prns{\prns{\tilde\theta_1^\top\lambda_1 + \tilde\theta_2^\top}S_2 + \tilde\theta_1^\top\lambda_2 A + \prns{\tilde\theta_1^\top \lambda_3 + \tilde\theta_0^\top}X + \tilde\theta_1^\top\lambda_4 U} + c_0
\end{align*}

Thus we only need the above to match \Cref{eq: density-ratio} for $a = 1$. 
This is possible once $\lambda_4$ has full column rank: then there exists $\tilde\theta_1$ such that $\tilde\theta_1^\top\lambda_4 = \kappa_2^\top$. Then we can choose $\tilde\theta_1, \tilde\theta_0, c_1, c_0$ accordingly. 
Analogously, we can also show the existence of a selection bridge function $q_0\prns{S_2, S_1, 0, X} $ of the same form for $a = 0$.
\end{proof}

\begin{proof}[Proof for \Cref{corollary: identify-exp}]
We first prove the conclusion for \Cref{eq: identification-1} in \Cref{thm: identification1}.
Following the proof for \Cref{thm: identification1},  we have 
\begin{align*}
&\Eb{h_0\prns{\Sc, \Sb, A, X} \mid A = a, G = E} \\
    =& \Eb{\Eb{Y(a) \mid \Sb(a), U, X, G = O} \mid G = E} \\
    =& \Eb{\Eb{Y(a) \mid \Sb(a), U, X, G = E} \mid G = E} = \Eb{Y\prns{a} \mid G = E},
\end{align*}
where the second follows from the assumption that $Y(a) \perp G \mid S(a), U, X$.

Next, we prove the conclusion for \Cref{eq: identification-2} in \Cref{thm: identification2}. Following the proof for \Cref{thm: identification2},  we have

\begin{align*}
&\Eb{q_0\prns{S_2, S_1, A, X}Y \mid A = a, G = O} \\ 
    =& \Eb{\Eb{Y(a) \mid S_2(a), U, X, G= O} \mid A = a, G = E} \\
    =& \Eb{\Eb{Y(a) \mid S_2(a), U, X, G= E} \mid G = E}  \\
    =& \Eb{Y(a) \mid G = E} = \mu(a),
\end{align*} 
where the second equation follows from the assumption that $Y(a) \perp G \mid S(a), U, X$.

Finally, according to the proof of \Cref{thm: identification-DR}, if conditions in \Cref{thm: identification1} hold and $h = h_0$ satisfies \Cref{eq: bridge-obs}, then 
\begin{align*}
&\Eb{h\prns{\Sc, \Sb, A, X} \mid A = a, G = E} + \Eb{q\prns{\Sb, \Sa, A, X}\prns{Y - h\prns{\Sc, \Sb, A, X}} \mid A = a, G = O} \\
=& \Eb{h\prns{\Sc, \Sb, A, X} \mid A = a, G = E}.
\end{align*}
If conditions in \Cref{thm: identification2} hold and $q = q_0$ satisfies \Cref{eq: bridge2-obs-1} or \Cref{eq: bridge2-obs-2}, then 
\begin{align*}
&\Eb{h\prns{\Sc, \Sb, A, X} \mid A = a, G = E} + \Eb{q\prns{\Sb, \Sa, A, X}\prns{Y - h\prns{\Sc, \Sb, A, X}} \mid A = a, G = O} \\
=& \Eb{q\prns{\Sb, \Sa, A, X}Y\mid A=a, G=O}.
\end{align*}
Then the conclusion follows from our proof above. 

\end{proof}

{
\begin{proof}[Proof for \Cref{corollary: sel-bridge-dist}]
The proof for \Cref{corollary: sel-bridge-dist} straitforwardly follows from the proof for \Cref{corollary: covariate-exp-dr} and \Cref{corollary: covariate-exp} by replacing all $Y$ with $r(Y)$. 
\end{proof}
}

\begin{proof}[Proof for \Cref{lemma: neyman-orthogonality}]
We denote the map in \Cref{eq: orthogonality-map} as $\Phi(\eta)$. Then we need to prove that 
\begin{align*}
\dot{\Phi}_j(\eta^*)[\eta_j-\eta_j^*] \coloneqq \frac{\partial}{\partial t} \Phi(\eta^*_1, \dots, \eta^*_j + t(\eta_j - \eta^*_j), \dots, \eta_7^*)\vert_{t=0} = 0, ~ \text{for any } \eta_j \text{ and }j \in \braces{1, \dots, 7}.
\end{align*}
First, we note that 
\begin{align*}
&\dot{\Phi}_1(\eta^*)[\eta_1-\eta_1^*] \\
=& \sum_{a\in\braces{0, 1}}\prns{-1}^{1-a}\bigg\{\mathbb{E}\bigg[\frac{\Prb{G=E}\Prb{G=O\mid X}}{\Prb{G=O}\Prb{G=E\mid X}}\frac{\indic{A=a}}{\Prb{A=a\mid X, G=E}}\prns{h - h_0}(S_3, S_2, a, X) \mid G = E\bigg] \\
   -\mathbb{E}\bigg[
   &\frac{\Prb{G=E\mid A=a}\Prb{G=O\mid X}}{\Prb{G=O\mid A=a}\Prb{G=E\mid X}}\frac{\indic{A=a}}{\Prb{A=a\mid X, G=E}}q_0\left(S_{2}, S_{1}, a, X\right)\prns{h - h_0}(S_3, S_2, a, X)\mid G=O\bigg]\bigg\}\\
=& \sum_{a\in\braces{0, 1}}\prns{-1}^{1-a}\braces{\Eb{\prns{h - h_0}(S_3(a), S_2(a), a, X) \mid G=O} - \Eb{\prns{h - h_0}(S_3(a), S_2(a), a, X) \mid G=O}} = 0,
\end{align*}
where the second equation follows from \Cref{eq: eliminate-1,eq: eliminate-2} in the proof for \Cref{corollary: covariate-exp-dr}.

Second, we have that  
\begin{align*}
&\dot{\Phi}_2(\eta^*)[\eta_2-\eta_2^*]\\
&=\sum_{a\in\braces{0,1}}\prns{-1}^{1-a}\bigg\{\Eb{\bar{h}_E(a, X) -\bar{h}_{0, E}(a, X) \mid G = O} \\
&-\Eb{\frac{\Prb{G=E}\Prb{G=O\mid X}}{\Prb{G=O}\Prb{G=E\mid X}}\frac{\indic{A=a}}{\Prb{A=a\mid X, G=E}}\prns{\bar{h}_{E}\prns{a, X}-\bar{h}_{E, 0}\prns{a, X}}\mid G= E}\bigg\} \\
&=\sum_{a\in\braces{0,1}}\prns{-1}^{1-a}\braces{\Eb{\bar{h}_E(a, X) -\bar{h}_{0, E}(a, X) \mid G = O} - \Eb{\bar{h}_E(a, X) -\bar{h}_{0, E}(a, X) \mid G = O}} = 0,
\end{align*}
where the equation follows from the proof for \Cref{corollary: covariate-exp-dr}.

Third, we have 
\begin{align*}
\dot{\Phi}_3(\eta^*)[\eta_3-\eta_3^*] 
   &= \sum_{a\in\braces{0,1}}\prns{-1}^{1-a}\mathbb{E}\bigg[\frac{\Prb{G=E\mid A=a}\Prb{G=O\mid X}}{\Prb{G=O\mid A=a}\Prb{G=E\mid X}}\frac{\indic{A=a}}{\Prb{A=a\mid X, G=E}}\\
   &\qquad \times\prns{q-q_0}\prns{S_2, S_1, a, X}\left(Y-h_0\left(S_{3}, S_{2}, A, X\right)\right)\mid G = O\bigg] \\
   &= \sum_{a\in\braces{0,1}}\prns{-1}^{1-a}\mathbb{E}\bigg[\frac{\Prb{G=E\mid A=a}\Prb{G=O\mid X}}{\Prb{G=O\mid A=a}\Prb{G=E\mid X}}\frac{\indic{A=a}}{\Prb{A=a\mid X, G=E}}\\
   &\qquad \times\prns{q-q_0}\prns{S_2, S_1, a, X}\Eb{Y-h_0\left(S_{3}, S_{2}, A, X\right)\mid S_2, S_1, A=a, X, G=O}\mid G = O\bigg] \\
   &= 0.
\end{align*}
Fourth, we have 
\begin{align*}
&\dot{\Phi}_4(\eta^*)[\eta_4-\eta_4^*]\\ 
   =& \sum_{a\in\braces{0,1}}\prns{-1}^{1-a}\bigg\{\mathbb{E}\bigg[\frac{\Prb{G=E}\Prb{G=O\mid X}}{\Prb{G=O}\Prb{G=E\mid X}}\\
   &\times \frac{\indic{A=a}}{\mathbb{P}^2\prns{A=a\mid X, G=E}}\prns{\eta_4^*-\eta_4}\Eb{h_0\prns{\Sc, \Sb, a, X} - \bar{h}_{E, 0}\prns{a, X}\mid A=a, X, G=O}\mid G=O\bigg] \\
   &+\mathbb{E}\bigg[\frac{\Prb{G=E\mid A=a}\Prb{G=O\mid X}}{\Prb{G=O\mid A=a}\Prb{G=E\mid X}}\frac{\indic{A=a}}{\mathbb{P}^2\prns{A=a\mid X, G=E}}\prns{\eta_4^*-\eta_4} \\
   &\times q_0\left(S_{2}, S_{1}, a, X\right)\Eb{Y-h_0\left(S_{3}, S_{2}, A, X\right)\mid S_2,S_1,A=a, X}\mid G= O\bigg] = 0.
\end{align*}
Following this proof for $\dot{\Phi}_4(\eta^*)[\eta_4-\eta_4^*]=0$, we can similarly show that $\dot{\Phi}_j(\eta^*)[\eta_j-\eta_j^*] = 0$ for $j=5, 6, 7$.
\end{proof}

{
\begin{proof}[Proof for \Cref{thm: DR-est-normality-X}]
We can represent the estimator $\hat\tau$ as $\hat\tau = \frac{1}{K}\sum_{k=1}^K \sum_{a\in\{0, 1\}}(-1)^{1-a}\hat\mu_k(a)$, where 
\begin{align*}
\hat\mu_k(a) 
    &= \frac{1}{n_{E, k}}\sum_{i \in \mathcal{D}_{E, k}}\indic{A_i = a}\hat\alpha_k(A_i, X_i)\prns{\hat h_k(S_{3, i}, S_{2, i}, A_i, X_i) - \hat{\bar{h}}_k(A_i, X_i)} \\
    &+ \frac{1}{n_{O, k}}\sum_{i\in \mathcal{D}_{O, k}}\hat{\bar{h}}_k(a, X_i) +  \indic{A_i = a}\hat\beta_k(A_i, X_i)\hat q_k(S_{2, i}, S_{1, i}, A_i, X_i)\prns{Y_i - \hat h_k(S_{3, i}, S_{2, i}, A_i, X_i)} \\
    &= \frac{1}{n_{E, k}}\sum_{i \in \mathcal{D}_{E, k}}\phi_2(Y_i, S_i, a, X_i; \hat\eta_k) + \frac{1}{n_{O, k}}\sum_{i \in \mathcal{D}_{O, k}}\prns{\phi_1(Y_i, S_i, a, X_i; \hat\eta_k) + \phi_3(Y_i, S_i, a, X_i; \hat\eta_k)}.
\end{align*}
We further denote 
\begin{align*}
\tilde\mu_k(a) = \frac{1}{n_{E, k}}\sum_{i \in \mathcal{D}_{E, k}}\phi_2(Y_i, S_i, a, X_i; \eta^*) + \frac{1}{n_{O, k}}\sum_{i \in \mathcal{D}_{O, k}}\prns{\phi_1(Y_i, S_i, a, X_i; \eta^*) + \phi_3(Y_i, S_i, a, X_i; \eta^*)}.
\end{align*}
We have 
\begin{align*}
\hat\mu_k(a) - \tilde\mu_k(a) =& \frac{1}{n_{E, k}}\sum_{i \in \mathcal{D}_{E, k}}\prns{\phi_2(Y_i, S_i, a, X_i; \hat\eta_k) - \phi_2(Y_i, S_i, a, X_i; \eta^*)} \\
+& \frac{1}{n_{O, k}}\sum_{i \in \mathcal{D}_{O, k}}\prns{\phi_1(Y_i, S_i, a, X_i; \hat\eta_k) - \phi_1(Y_i, S_i, a, X_i; \eta^*)} \\
+& \frac{1}{n_{O, k}}\sum_{i \in \mathcal{D}_{O, k}}\prns{\phi_3(Y_i, S_i, a, X_i; \hat\eta_k) - \phi_3(Y_i, S_i, a, X_i; \eta^*)}.
\end{align*}
By following the proof of \Cref{thm: dist-dr}, we can show that for any consistent estimator $\hat\eta_k$, 
\begin{align*}
\hat\mu_k(a) - \tilde\mu_k(a) =& \Eb{\phi_2(Y, S, a, X; \hat\eta_k) - \phi_2(Y, S, a, X; \eta^*) \mid G = E, \mathcal{D}_{O, -k}\cup \mathcal{D}_{E, -k}} \\
 +& \Eb{\phi_1(Y, S, a, X; \hat\eta_k) - \phi_1(Y, S, a, X; \eta^*) \mid G = O, \mathcal{D}_{O, -k}\cup \mathcal{D}_{E, -k}} \\
 +& \Eb{\phi_3(Y, S, a, X; \hat\eta_k) - \phi_3(Y, S, a, X; \eta^*) \mid G = O, \mathcal{D}_{O, -k}\cup \mathcal{D}_{E, -k}} + o_{\mathbb{P}}(n^{-1/2}),
\end{align*}
where $\mathcal{D}_{O, -k} = \cup_{j \ne k}\mathcal{D}_{O, j}$ and $\mathcal{D}_{E, -k} = \cup_{j \ne k}\mathcal{D}_{E, j}$. 
We can further verify that 
\begin{align}
& \Eb{\phi_2(Y, S, a, X; \hat\eta_k) - \phi_2(Y, S, a, X; \eta^*) \mid G = E, \mathcal{D}_{O, -k}\cup \mathcal{D}_{E, -k}} \nonumber \\
 +& \Eb{\phi_1(Y, S, a, X; \hat\eta_k) - \phi_1(Y, S, a, X; \eta^*) \mid G = O, \mathcal{D}_{O, -k}\cup \mathcal{D}_{E, -k}} \nonumber \\
 +& \Eb{\phi_3(Y, S, a, X; \hat\eta_k) - \phi_3(Y, S, a, X; \eta^*) \mid G = O, \mathcal{D}_{O, -k}\cup \mathcal{D}_{E, -k}}  \nonumber \\
 =& \Eb{\indic{A = a}\prns{\alpha_0(A, X) - \hat\alpha_k(A, X)}\prns{\hat{\bar{h}}_k(A, X) - {\bar{h}}_{E, 0}(A, X)} \mid G = E, \mathcal{D}_{O, -k}\cup \mathcal{D}_{E, -k}}, \label{eq: dr-bound-1} \\
 +& \mathbb{E}[\indic{A = a}\prns{\hat\alpha_k(A, X) - \alpha_0(A, X)}\nonumber\\
 &\qquad\qquad\qquad\qquad \times \prns{\hat{{h}}_k(S_3, S_2, A, X) - {{h}}_0(S_3, S_2, A, X)} \mid G = E, \mathcal{D}_{O, -k}\cup \mathcal{D}_{E, -k}], \label{eq: dr-bound-2} \\
 +& \mathbb{E}[\indic{A = a}\beta_0(A, X)\prns{\hat q_k(S_{2}, S_{1}, A, X) -  q_0(S_{2}, S_{1}, A, X)}  \nonumber \\
 &\qquad\qquad\qquad\qquad \times \prns{h_0(S_{3}, S_{2}, A, X) - \hat h_k(S_{3}, S_{2}, A, X)} \mid G = O, \mathcal{D}_{O, -k}\cup \mathcal{D}_{E, -k}], \label{eq: dr-bound-3} \\
 +& \mathbb{E}[\indic{A = a}\prns{\hat \beta_k(A, X) - \beta_0(A, X)}\hat q_k(S_{2}, S_{1}, A, X), \nonumber\\
 &\qquad\qquad\qquad\qquad \times \prns{h_0(S_{3}, S_{2}, A, X) - \hat h_k(S_{3}, S_{2}, A, X)} \mid G = O, \mathcal{D}_{O, -k}\cup \mathcal{D}_{E, -k}].  \label{eq: dr-bound-4}
\end{align}
By the Cauchy-Schwartz inequality, 
\begin{align*}
\abs{\Cref{eq: dr-bound-1}} 
    &\lesssim \|\hat\alpha_k - \alpha_0\|_{\mathcal{L}_2(\mathbb{P})}\|\hat{\bar{h}}_k - \bar{h}_{E, 0}\|_{\mathcal{L}_2(\mathbb{P})} = o_{\mathbb{P}}(n^{-1/2}), \\
\abs{\Cref{eq: dr-bound-2}} &\lesssim  \|\hat\alpha_k - \alpha_0\|_{\mathcal{L}_2(\mathbb{P})}\|T(\hat h_k - h_0)\|_{\mathcal{L}_2(\mathbb{P})} = o_{\mathbb{P}}(n^{-1/2}), \\
\abs{\Cref{eq: dr-bound-3}} &\lesssim \min\{\|T(\hat h_k - h_0)\|_{\mathcal{L}_2(\mathbb{P})}\|\hat q_k - q_0\|_{\mathcal{L}_2(\mathbb{P})}, \|T^\star(\hat q_k - q_0)\|_{\mathcal{L}_2(\mathbb{P})}\|\hat h_k - h_0\|_{\mathcal{L}_2(\mathbb{P})}\} = o_{\mathbb{P}}(n^{-1/2}), \\
\abs{\Cref{eq: dr-bound-4}} & \lesssim \|\hat\beta_k - \beta_0\|_{\mathcal{L}_2(\mathbb{P})}\|T(\hat h_k - h_0)\|_{\mathcal{L}_2(\mathbb{P})} = o_{\mathbb{P}}(n^{-1/2}).
\end{align*}
Therefore, 
\begin{align*}
\hat\mu_k(a) = \tilde\mu_k(a) + o_{\mathbb{P}}(n^{-1/2}).
\end{align*}
It follows that 
\begin{align*}
\hat\tau - \tau 
    &= \frac{1}{K}\sum_{k=1}^K \sum_{a\in\{0, 1\}}(-1)^{1-a}\hat\mu_k(a) - \tau \\
    &= \frac{1}{K}\sum_{k=1}^K \sum_{a\in\{0, 1\}}(-1)^{1-a}\tilde\mu_k(a) - \tau + o_{\mathbb{P}}(n^{-1/2}) \\
    &= \sum_{a\in\{0, 1\}}(-1)^{1-a}\bigg\{\frac{1}{n_{E}}\sum_{i \in \mathcal{D}_{E}}\phi_2(Y_i, S_i, a, X_i; \eta^*) \\
    &+ \frac{1}{n_{O}}\sum_{i \in \mathcal{D}_{O}}\prns{\phi_1(Y_i, S_i, a, X_i; \eta^*) - \mu(a)} + \phi_3(Y_i, S_i, a, X_i; \eta^*)\bigg\} + o_{\mathbb{P}}(n^{-1/2}).
    \end{align*}
    Then according to the Central Lmit Theorem, we have the asserted conclusion.
\end{proof}}

\begin{proof}[Proof for \cref{corollary: S1-pretreat}]
We can first follow the proof for \Cref{thm: identification1} to show that for any $h_0\prns{S_3, S_2, A, X}$ that satisfies \Cref{eq: bridge-U}, 
\begin{align*}
&\Eb{h_0\prns{S_3, S_2, A, X} \mid S_1, A = a, X, G = E} \\
    &\qquad\qquad = \Eb{\Eb{Y(a) \mid S_2(a), U, X, G = O} \mid S_1, A = a, X, G = E}.
\end{align*}
The rest of the proof is analogous to \Cref{corollary: covariate-exp}. 
\end{proof}

\begin{proof}[Proof for \cref{corollary: S1-S2-pretreat}]
We can first follow the proof for \Cref{thm: identification1} to show that for any $h_0\prns{S_3, S_2, A, X}$ that satisfies \Cref{eq: bridge-U}, 
\begin{align*}
&\Eb{h_0\prns{S_3, S_2, A, X} \mid S_2, S_1, A = a, X, G = E} \\
    &\qquad\qquad = \Eb{\Eb{Y(a) \mid S_2(a), U, X, G = O} \mid S_2, S_1, A = a, X, G = E}.
\end{align*}
The rest of the proof is analogous to \Cref{corollary: covariate-exp}. 
\end{proof}

\begin{proof}[Proof for \Cref{corollary: partial}]
First, note that under \Cref{assump: partial,assump: partial-2}, we can follow the proofs for \Cref{cor: implication1,cor: implication2} to show that $S_3 \perp G \mid S_2, A = a, U, X$, and $\prns{Y, S_3} \perp S_1 \mid S_2, A, U_{\diamond}, X,  G = O$.

Second, following the proof for \Cref{lemma: bridge-obs}, we can show that for any function $h_0$ that satisfies  \Cref{eq: bridge-obs}, it must also satisfy 
\begin{align*}
\Eb{Y \mid \Sb, A, U_{\diamond}, X, G = O} = \Eb{h_0\prns{\Sc, \Sb, A, X} \mid \Sb, A, U_{\diamond}, X, G = O}.
\end{align*}
Finally, we can follow the proof for \Cref{corollary: covariate-exp} to show that for any function $h_0$ that satisfies \Cref{eq: bridge-U-partial}, \Cref{eq: identification-1-X} in \Cref{corollary: covariate-exp} holds. This concludes the proof for \Cref{corollary: partial}. 
\end{proof}

\begin{proof}[Proof of \Cref{prop: experiment}]
We already have $Z_2 \perp G \mid Z_1$. Thus, we only need to verify $G \perp Z_1$. Note that
\begin{align*}
p(z_1 \mid G = 1) = & 
{p(z_1, A = 1 \mid G = 1)} + {p(z_1, A = 0 \mid G = 1)}  \\
   = & \frac{\Prb{G = 1 \mid A = 1, Z_1 = z_1}\Prb{A = 1}p(z_1)}{\Prb{G = 1}} \\
   & + \frac{\Prb{G = 1 \mid A = 0, Z_1 = z_1}\Prb{A = 0}p(z_1)}{\Prb{G = 1}} \\
   = & p(z_1)\frac{C}{\Prb{G = 1}} \propto p(z_1),
 \end{align*} 
 which proves the desired result.
\end{proof}

{
\begin{proof}[Proof of \Cref{thm:ext-valid}]
    From the definition of external validity bridge function, we have
    \[
    \frac{p(U \mid S_2, X, A = a, G = O)}{p(U \mid S_2, X, A = a, G = E)} = \E[\tilde{q}\prns{S_1, S_2, X, A} \mid S_2, X, U, A = a, G = E] \cdot \frac{p(S_2, X \mid A = a, G = E)}{p(S_2, X \mid A = a, G = O)}.
    \]
    Then
    \begin{align*}
        & p(U \mid S_2, X, A = a, G = O) \\
        & = \E[\tilde{q}\prns{S_1, S_2, X, A} \mid S_2, X, U, A = a, G = E] \cdot \frac{p(S_2, X \mid A = a, G = E)}{p(S_2, X \mid A = a, G = O)} \cdot p(U \mid S_2, X, A = a, G = E) \\
        & = \E[\tilde{q}\prns{S_1, S_2, X, A} \mid S_2, X, U, A = a, G = E] \cdot \frac{p(S_2, X \mid G = E)}{p(S_2, X \mid A = a, G = O)} \cdot p(U \mid S_2, X, G = E),
    \end{align*}
    where for the second equality we use $(S_2, X, U) \perp A \mid G = E$. 
    Then
    \begin{align}
        & p(U \mid S_2, X, G = O) \nonumber \\
        & = \sum_a p(U \mid A = a, S_2, X, G = O) \Prb{A = a \mid S_2, X, G = O} \nonumber\\ 
        & = \sum_a \Prb{A = a \mid S_2, X, G = O} \E[\tilde{q}\prns{S_1, S_2, X, A} \mid S_2, X, U, A = a, G = E] \nonumber\\
        & \qquad\qquad \cdot \frac{p(S_2, X \mid G = E)}{p(S_2, X \mid A = a, G = O)} \cdot p(U \mid S_2, X, G = E) \nonumber\\
        & = \sum_a \E[\tilde{q}\prns{S_2, X, A, Z} \mid S_2, X, U, A = a, G = E] \Prb{A = a \mid G = O} \nonumber\\
            &\qquad\qquad\qquad  \times \frac{p(S_2, X \mid G = E)}{p(S_2, X \mid G = O)}  p(U \mid S_2, X, G = E) \label{eq:ext-valid}\\
        & = \E\left[\frac{\Prb{A \mid G = O}}{\Prb{A \mid G = E}} \tilde{q}\prns{S_1, S_2, X, A} \mid X, U, G = E\right] \cdot \frac{p(S_2, X \mid G = E)}{p(S_2, X \mid G = O)} \cdot p(U \mid S_2, X, G = E),\label{eq:ext-valid-2}
    \end{align}
    where to get~\eqref{eq:ext-valid} we use that 
    \[
    \frac{\Prb{A = a \mid S_2, X, G = O}}{p(S_2, X \mid A = a, G = O)} = \frac{\Prb{A = a \mid G = O}}{p(S_2, X \mid G = O)}
    \]
    and for the last equality we use again that $A \perp \prns{S_2, U, X} \mid G = E$  so that 
    \begin{align*}
        & \sum_a \E[\tilde{q}\prns{S_1, S_2, X, A} \mid S_2, X, U, A = a, G = E] \Prb{A = a \mid G = O} \\
        & = \sum_a \E[\tilde{q}\prns{S_1, S_2, X, A} \mid S_2, X, U, A = a, G = E] \frac{\Prb{A = a \mid G = O}}{\Prb{A = a \mid G = E}}  \Prb{A = a \mid G = E} \\
        & = \sum_a \E[\tilde{q}\prns{S_1, S_2, X, A} \mid S_2, X, U, A = a, G = E] \frac{\Prb{A = a \mid G = O}}{\Prb{A = a \mid G = E}}  \Prb{A = a \mid S_2, X, U, G = E} \\
        & =  \E\left[\frac{\Prb{A \mid G = O}}{\Prb{A \mid G = E}} \tilde{q}\prns{S_1, S_2, X, A} \mid S_2, X, U, G = E\right].
    \end{align*}
    From above, we have that
    \begin{align*}
        & \E[Y(a) \mid S_2, X, G = O] = \E [\E[Y(a) \mid S_2, U, X, G = O] \mid S_2, X, G = O ]\\
        & = \E [\E[Y(a) \mid S_2, U, X, G = E] \mid S_2, X, G = O] \\
        & = \E [\E[h(S_3, S_2, X, A) \mid A = a, S_2, U, X, G = E] \mid S_2, X, G = O] \\
        & = \E \left[\E[h(S_3, S_2, X, A) \mid A = a, S_2, U, X, G = E] \frac{p(U \mid S_2, X, G = O)}{p(U \mid S_2, X, G = E)}\mid S_2, X, G = E\right]
    \end{align*}
    Where for the second equality we use that $G \perp Y(a) \mid S_2, U, X$ 
    Then from~\eqref{eq:ext-valid-2}, we further have
    \begin{align*}
        & \E[Y(a) \mid S_2, X, G = O] = \\
        & \E \left[\E[h(S_3, S_2, X, A) \mid U,S_2,  X, G = E, A = a] \E\left[\frac{p(A \mid G = O)}{p(A \mid G = E)} \tilde{q}\prns{S_1, S_2, X, A} \mid S_2, X, U, G = E\right] \mid S_2, X, G = E\right] \\
        & \qquad \cdot \frac{p(S_2, X \mid G = E)}{p(S_2, X \mid G = O)}.
    \end{align*}
    From here, and that for any function $f(S_1, S_2, X, A)$,
    \begin{align*}
        & \E \left[\E[h(S_3, S_2, X, A) \mid U, S_2, X, G = E, A = a] \E\left[f(S_1, S_2, X, A) \mid S_2, X, U, G = E\right] \mid S_2, X, G = E\right] \\
        & = \E \Bigg[\E[h(S_3, S_2, X, A) \mid U, S_2, X, G = E, A = a] \\ 
        & \qquad\quad \cdot \E\left[ \sum_{a'} f(S_1, S_2, X, a') \Prb{A = a' \mid G = E} \mid S_2, X, U, G = E\right] \mid S_2, X, G = E\Bigg] \\
        & = \E \Bigg[\E[h(S_3, S_2, X, A) \mid U, S_2, X, G = E, A = a] \\
        & \qquad\quad \cdot \E\left[ \sum_{a'} f(S_1, S_2, X, a') \Prb{A = a' \mid G = E} \mid S_2, X, U, A = a, G = E\right] \mid S_2, X, G = E, A = a\Bigg] \\
        & = \E\left[ h(S_3, S_2, X, A) \sum_{a'} f(S_1, S_2, X, a') \Prb{A = a' \mid G = E} \mid S_2, X, A = a, G = E\right],
    \end{align*}
    where for the second equality we use that $S_1 \perp A \mid S_2, X, U, G = E$ 
    and that $U \perp A \mid S_2, X, G = E$. Finally, we have that
    \begin{align*}
        & \E[Y(a) \mid S_2, X, G = O] \\
        & = \E \left[h(S_3, S_2, X, A) \sum_{a'}\frac{\Prb{A = a' \mid G = O}}{\Prb{A = a' \mid G = E}} \tilde{q}\prns{S_1, S_2, X, a'} \Prb{A = a' \mid G = E}\ \mid S_2, X, A = a, G = E\right] \\
        & \qquad \cdot \frac{p(S_2, X \mid G = E)}{p(S_2, X \mid G = O)} \\
        & = \E \left[h(S_3, S_2, X, A) \sum_{a'}\Prb{A = a' \mid G = O} \tilde{q}\prns{S_1, S_2, X, a'} \mid S_2, X, A = a, G = E\right] \cdot \frac{p(S_2, X \mid G = E)}{p(S_2, X \mid G = O)},
    \end{align*}
which proves the desired result. 
It then follows that 
\begin{align*}
    \mu(a) & = \E[\E[Y(a) \mid S_2, X, G = O] \mid G = O] = \E\left[m(S_2, a, X) \frac{p\prns{S_2, X \mid G = E}}{p\prns{S_2, X \mid G = O}} \mid G = O\right] \\
    & = \E[m(S_2, a, X) \mid G = E].
\end{align*}
Therefore, we have 
\begin{align*}
\tau = \mu(1) - \mu(0) = \E[m(S_2, 1, X) - m(S_2, 0, X) \mid G = E].
\end{align*}
\end{proof}
}

\begin{proof}[Proof for \Cref{thm: control-fun}]
{
We first define the following stochastic processes: 
\begin{align*}
\Vcal(a) &= \braces{p_{S_3(a)}(s_3 \mid S_2(a), S_1(a), A, X, G = O): s_3 \in \mathcal{S}_3}, \\
\tilde\Vcal(a) &= \braces{p_{S_3(a)}(s_3 \mid S_2(a), S_1(a), A = a, X, G = O): s_3 \in \mathcal{S}_3}, \\
\Wcal(a) &= \braces{p(u \mid S_2(a), S_1(a), A,  X, G = O): u \in \Ucal}, \\
\tilde\Wcal(a) &= \braces{p(u \mid S_2(a), S_1(a), A = a,  X, G = O): u \in \Ucal}. 
\end{align*}
To prove the desired conclusion, we only need to prove that 
\begin{align}\label{eq: proof-target}
\Eb{r(Y(a)) \mid X, G = O} = \Eb{\Eb{r(Y) \mid \Vcal, S_2, A,X, G = O} \mid A = a, X, G = E}.
\end{align}
Then the conclusion follows from the iterated law of conditional expectation. 
We will prove \Cref{eq: proof-target} above by showing that 
\begin{align}\label{eq: proof-target2}
&\Eb{\Eb{r(Y) \mid \Vcal, S_2, A,X, G = O} \mid A = a, X, G = E} \nonumber \\
=& \Eb{\Eb{r(Y(a)) \mid \Vcal(a), S_2(a), X, G = O} \mid A = a, X, G = E}.
\end{align}
Then \Cref{eq: proof-target} follows from the fact that 
\begin{align*}
&\Eb{\Eb{r(Y(a)) \mid \Vcal(a), S_2(a), X, G = O} \mid A = a, X, G = E} \\
=& \Eb{\Eb{r(Y(a)) \mid \tilde\Vcal(a), S_2(a), X, G = O} \mid A = a, X, G = E} \\
=& \Eb{\Eb{r(Y(a)) \mid \tilde\Vcal(a), S_2(a), X, G = O} \mid X, G = E} \\
=& \Eb{\Eb{r(Y(a)) \mid \tilde\Vcal(a), S_2(a), X, G = O} \mid X, G = O} \\
=& \Eb{r(Y(a)) \mid X, G = O},
\end{align*}
where the second equality follows from the fact that $(S_2(a), S_1(a)) \perp A \mid X, G = E$ and that $\tilde \Vcal(a)$ is  determined by $(S_2(a), S_1(a), X)$, the third equality follows from the fact that $(S_2(a), S_1(a)) \perp G = O \mid X$, and the last equality follows from the iterated law of conditional expectation. 
}

{
Now we focus on proving \Cref{eq: proof-target2}. For brevity, we omit $X$ in all derivations so all conditional expectations below should be understood as conditioning on $X$ implicitly. We prove \Cref{eq: proof-target2} in two steps. 
\textbf{Step I: }  we first derive the relation between $\tilde\Wcal(a)$ and $\tilde\Vcal(a)$ and the relation between $\Wcal(a)$ and $\Vcal(a)$, under the completeness condition in 
\Cref{assump: completeness} condition \ref{assump: completeness-1}. By the law of total probability, we have
\begin{align*}
&p_{S_3(a)}({s_3 \mid S_2(a) = s_2, S_1(a) = s_1, A = a, G= O}) \\
=& \int p_{S_3(a)}({s_3 \mid S_2(a) = s_2, S_1(a) = s_1, A = a, U = u, G= O})p(u \mid S_2(a) = s_2, S_1(a) = s_1, A = a, G = O) \diff u \\
=& \int p_{S_3(a)}({s_3 \mid S_2(a) = s_2, U = u, G= O})p(u \mid S_2(a) = s_2, S_1(a) = s_1, A = a, G = O) \diff u \\
=& \Phi_{s_2}\bracks{p_U(\cdot \mid S_2(a) = s_2, S_1(a) = s_1, A = a, G = O)}(s_3),
\end{align*}
where the second equality follows from the fact that $S_3(a) \perp (S_1(a), A) \mid X, U, G = O$, and $\Phi_{s_2}$ is a mapping defined as follows: for any function $g: \Ucal \mapsto \R{}$, 
\begin{align*}
\phi_{s_2}[g(u)](s_3) = \int p_{S_3(a)}({s_3 \mid S_2(a) = s_2, U = u, G= O})g(u) \diff u.
\end{align*}
Now we show that this mapping is injective. To see this, consider  any two functions $g_1: \Ucal \mapsto \R{}$ and $g_2: \Ucal \mapsto \R{}$ such that $\Phi_{s_2}[g_1](s_3) = \Phi_{s_2}[g_2](s_3)$ for all $s_3$ such that $p(s_3 \mid S_2 = s_2, A = a, G = O) > 0$. Note that  we have 
\begin{align*}
\phi_{s_2}[g(u)](s_3) 
    &= \int p_{S_3(a)}({s_3 \mid S_2(a) = s_2, U = u, G= O})g(u) \diff u \\
    &= \int p_{S_3(a)}({s_3 \mid S_2(a) = s_2, A = a, U = u, G= O})g(u) \diff u \\
    &= \int p_{S_3}({s_3 \mid S_2 = s_2, A = a, U = u, G= O})g(u) \diff u \\
    &= \int p(u \mid S_3 = s_3, S_2 = s_2, A = a, G = O)\frac{p(s_3 \mid S_2 = s_2, A = a, G = O)}{p(u \mid S_2 = s_2, A = a, G = O)}g(u)\diff u. 
\end{align*}
According to the completeness condition in 
\Cref{assump: completeness} condition \ref{assump: completeness-1}, $\phi_{s_2}[g(u)](s_3) = 0$ for all $s_3$ such that $p(s_3 \mid S_2 = s_2, A = a, G = O) > 0$ if and only if $g(u) = 0$ for all $u$ such that $p(u \mid S_2 = s_2, A = a, G = O) > 0$. This in turn implies that $\Phi_{s_2}[g_1](s_3) - \Phi_{s_2}[g_2](s_3) = \Phi_{s_2}[g_1 - g_2](s_3) = 0$ for all  $s_3$ such that $p(s_3 \mid S_2 = s_2, A = a, G = O) > 0$ if and only if $g_1(u) = g_2(u)$ for all $u$ such that $p(u \mid S_2 = s_2, A = a, G = O) > 0$. Therefore, $\phi_{s_2}$ is an injective mapping. It follows that there exists another mapping $\Psi_{s_2}$ such that 
\begin{align*}
p(u \mid s_2(a) = s_2, S_1(a) = s_1, A = a, G = O) = \Psi_{s_2}[p_{S_3(a)}({s_3 \mid S_2(a) = s_2, S_1(a) = s_1, A = a, G= O})](u).
\end{align*}
Therefore, we have $\tilde\Wcal(a) = \Psi_{S_2(a)}[\tilde\Vcal(a)]$. 
}

{
By the same token, we also have 
\begin{align*}
&p_{S_3(a)}({s_3 \mid S_2(a) = s_2, S_1(a) = s_1, A, G= O}) = \Phi_{s_2}\bracks{p_U(\cdot \mid S_2(a) = s_2, S_1(a) = s_1, A, G = O)}(s_3).
\end{align*}
We thus also have $\Wcal(a) = \Psi_{S_2(a)}[\Vcal(a)]$.}

{
\textbf{Step II: } We next prove \Cref{eq: proof-target2}. 
Note that 
\begin{align*}
\Eb{\Eb{r(Y) \mid \Vcal, S_2, A, G = O} \mid A = a, G = E} = \Eb{\Eb{r(Y(a)) \mid \Vcal(a), S_2(a), A = a, G = O} \mid A = a, G = E}.
\end{align*}
By the iterated law of conditional expectation, 
\begin{align*}
&\Eb{r(Y(a)) \mid \Vcal(a), S_2(a), A = a, G = O} \\
=& \Eb{\Eb{r(Y(a)) \mid \Vcal(a), S_2(a), S_1(a), A = a, G = O} \mid \Vcal(a), S_2(a), A = a, G = O} \\
=& \Eb{\Eb{r(Y(a)) \mid \Vcal(a), S_2(a), S_1(a), A, G = O} \mid \Vcal(a), S_2(a), A = a, G = O} \\
=& \Eb{\Eb{r(Y(a)) \mid S_2(a), S_1(a), A, G = O} \mid \Vcal(a), S_2(a), A = a, G = O} \\
=& \Eb{\Eb{r(Y(a)) \mid S_2(a), S_1(a), A = a, G = O} \mid \Vcal(a), S_2(a), A = a, G = O} \\
=& \Eb{\Eb{\Eb{r(Y(a)) \mid S_2(a), S_1(a), A = a, U, G = O} \mid S_2(a), S_1(a), A = a, G = O} \mid \Vcal(a), S_2(a), A = a, G = O} \\
=& \Eb{\Eb{\Eb{r(Y(a)) \mid S_2(a), U, G = O} \mid S_2(a), S_1(a), A = a, G = O} \mid \Vcal(a), S_2(a), A = a, G = O},
\end{align*}
where the thid equality holds because $V(a)$ is fully determined by $S_2(a), S_1(a), A$, and the last equality holds because $Y(a) \perp S_1(a) \mid S_2(a), U, X, G = O$.
}

{
Here 
\begin{align*}
&\Eb{\Eb{r(Y(a)) \mid S_2(a), U, G = O} \mid S_2(a), S_1(a), A = a, G = O} \\
=& \int \Eb{r(Y(a)) \mid S_2(a), U = u, G = O}p(u \mid S_2(a), S_1(a), A = a, G = O)\diff u \\
=& \int \Eb{r(Y(a)) \mid S_2(a), U = u, G = O}[\tilde\Wcal(a)](u)\diff u \\
=& \int \Eb{r(Y(a)) \mid S_2(a), U = u, G = O}[\Psi_{S_2(a)}[\tilde\Vcal(a)]](u)\diff u
\end{align*}
It then follows that 
\begin{align*}
&\Eb{\Eb{r(Y(a)) \mid \Vcal(a), S_2(a), A = a, G = O} \mid A = a, G = E} \\
=& \Eb{\Eb{\int \Eb{r(Y(a)) \mid S_2(a), U = u, G = O}[\Psi_{S_2(a)}[\tilde\Vcal(a)]](u)\diff u \mid \Vcal(a), S_2(a), A = a, G = O}\mid A=a, G = E} \\
=& \Eb{\Eb{\int \Eb{r(Y(a)) \mid S_2(a), U = u, G = O}[\Psi_{S_2(a)}[\Vcal(a)]](u)\diff u \mid \Vcal(a), S_2(a), A, G = O}\mid A=a, G = E} \\
=& \Eb{\Eb{\int \Eb{r(Y(a)) \mid S_2(a), U = u, G = O}[\Psi_{S_2(a)}[\Vcal(a)]](u)\diff u \mid \Vcal(a), S_2(a), G = O}\mid A=a, G = E}. 
\end{align*}
where the last equality follows from the fact that conditionally on $S_2(a)$, the inner term within $\Eb{\cdot \mid \Vcal(a), S_2(a), A = a, G = O}$ in the second equality above only depends on $V(a)$. 
}

{
Moreover, we have 
\begin{align*}
&\int \Eb{r(Y(a)) \mid S_2(a), U = u, G = O}[\Psi_{S_2(a)}[\Vcal(a)]](u)\diff u \\
=& \int \Eb{r(Y(a)) \mid S_2(a), S_1(a), A, U = u, G = O}[\Wcal(a)](u)\diff u \\
=& \int \Eb{r(Y(a)) \mid S_2(a), S_1(a), A, U = u, G = O}p(u \mid S_2(a), S_1(a), A, G = O)\diff u \\
=& \int \Eb{r(Y(a)) \mid \Vcal(a), S_2(a), S_1(a), A, U = u, G = O}p(u \mid \Vcal(a), S_2(a), S_1(a), A, G = O)\diff u \\
=&  \Eb{r(Y(a)) \mid \Vcal(a), S_2(a), S_1(a), G = O},
\end{align*}
where the first equality uses the fact that $Y(a) \perp (S_1(a), A) \mid S_2(a), U, X, G = O$, the third equality uses the fact that $\Vcal(a)$ is fully determined by $S_2(a), S_1(a), A$, and the last equality uses the iterated law of conditional expectation.
}

{
This means that 
\begin{align*}
&\Eb{\int \Eb{r(Y(a)) \mid S_2(a), U = u, G = O}[\Psi_{S_2(a)}[\Vcal(a)]](u)\diff u \mid \Vcal(a), S_2(a), G = O} \\
=&\Eb{\Eb{r(Y(a)) \mid \Vcal(a), S_2(a), S_1(a), G = O} \mid \Vcal(a), S_2(a), G = O} \\
=&\Eb{r(Y(a)) \mid \Vcal(a), S_2(a), G = O}.
\end{align*}
It follows that 
\begin{align*}
\Eb{\Eb{r(Y) \mid \Vcal, S_2, A, G = O}\mid A = a, G = E} = \Eb{\Eb{r(Y(a)) \mid \Vcal(a), S_2(a), G = O} \mid A = a, G = E}.
\end{align*}
This finishes proving \Cref{eq: proof-target2} (with $X$ being implicitly conditioned on everywhere).
} 

\end{proof}

\end{document}